\documentclass[twocolumn,longbibliography,nofootinbib, superscriptaddress,10pt,aps,prx]{revtex4-2}
\usepackage[dvips]{graphicx} 
\usepackage{amsfonts}
\usepackage{csquotes} 
\usepackage{amssymb}
\usepackage{amscd}
\usepackage{amsmath}    
\usepackage{amsthm}
\usepackage{bm} 
\usepackage{bbm}
\usepackage{booktabs}
\usepackage{enumerate}
\usepackage{enumitem}
\usepackage{epsfig}
\usepackage{subfigure}
\usepackage{subfloat}
\usepackage{xcolor}	
\usepackage{physics}
\usepackage[most]{tcolorbox}
\usepackage{tikz}
\usetikzlibrary{quantikz2}
\usepackage{float}
\usepackage{appendix}
\usepackage{algpseudocode}  
\usepackage{dsfont}
\usepackage[colorlinks, linkcolor=blue, anchorcolor=blue, citecolor=blue]{hyperref}
\usepackage{MnSymbol}
\usepackage[linesnumbered,ruled,vlined]{algorithm2e}
\usepackage{times}

\newlist{protocolsteps}{enumerate}{1}
\setlist[protocolsteps]{
	label=\arabic*.,
	leftmargin=2.1em,
	labelsep=0.5em,
	topsep=0.2em,
	itemsep=0.18em,
	parsep=0pt,
	partopsep=0pt
}

\newcounter{protocol}
\renewcommand{\theprotocol}{\arabic{protocol}}

\newenvironment{protocol}[2]
{%
	\noindent
	\begin{minipage}{\linewidth}
		\refstepcounter{protocol}%
		\label{#2}%
		\raggedright
		\setlength{\parindent}{0pt}%
		
		\hrule height \heavyrulewidth
		\vspace{0.55ex}
		
		\textbf{Protocol~\theprotocol.} #1\par
		
		\vspace{0.35ex}
		\hrule height \lightrulewidth
		\vspace{0.65ex}
	}
	{%
		\vspace{0.15ex}
		\hrule height \heavyrulewidth
	\end{minipage}
}

\newcommand{\TableL}[2]{%
	\parbox[c]{#1}{\raggedright #2\par}%
}

\newcommand{\TableC}[2]{%
	\parbox[c]{#1}{\centering #2\par}%
}

\newtcolorbox[auto counter]{mybox}[2]{
	enhanced,
	breakable,
	label=#1,
	colback=blue!5!white,
	colframe=blue!75!black,
	fonttitle=\bfseries,
	title=Box \thetcbcounter: #2
}

\newtheorem{theorem}{Theorem}
\newtheorem{lemma}{Lemma}
\newtheorem{corollary}{Corollary}
\newtheorem{claim}{Claim}

\newtheorem{fact}{Fact}

\newtheorem{definition}{Definition}
\newtheorem{proposition}{Proposition}

\newtheorem{result}{Result}

\newcommand{\bE}{\mathbb{E}}
\newcommand{\bZ}{\mathbb{Z}}

\newcommand{\bF}{\mathbb{F}}
\newcommand{\bmz}{{\bm{z}}}
\newcommand{\bmx}{{\bm{x}}}
\newcommand{\bmb}{{\bm{b}}}
\newcommand{\cO}{\mathcal{O}}
\newcommand{\cL}{\mathcal{L}}
\newcommand{\cP}{\mathcal{P}}
\newcommand{\poly}{\mathrm{poly}}

\newcommand{\wt}{\mathrm{wt}}

\newcommand{\Var}{\mathrm{Var}}
\DeclareMathOperator{\spn}{span}
\newcommand{\Bbrange}{\sfB^{n_B}}
\newcommand{\Ber}{\mathrm{Bernoulli}}
\newcommand{\sfB}{{\mathsf{B}}}

\newcommand{\dtr}{\mathrm{d}_{\mathrm{tr}}}

\newcommand{\freeset}{\mathcal{F}_{\mathsf{P}}}
\newcommand{\freeprojset}{\mathcal{FP}_{\mathsf{P}}}
\newcommand{\fid}{\mathrm{Fid}_{\mathsf{P}}}
\newcommand{\LQ}{\mathrm{LQ}_{\mathsf{P}}}
\newcommand{\Bzrange}{\{0,1\}^{n_B}}
\newcommand{\Axrange}{\mathsf{X}^{n_A}}
\newcommand{\Bxrange}{\mathsf{X}^{n_B}}
\newcommand{\zsum}{\bmz \in \Bzrange}
\newcommand{\propwitness}{O_{\mathsf{P}}(\psi)}
\newcommand{\accept}{\mathsf{accept}}
\newcommand{\reject}{\mathsf{reject}}

\newcommand{\pfreeprojset}{\mathcal{FP}_{\mathsf{P}}^p}
\newcommand{\fidp}{\mathrm{Fid}_{\pfreeprojset}}
\newcommand{\fidpprime}{\mathrm{Fid}_{\mathcal{FP}_{\mathsf{P}}^{p'}}}
\newcommand{\culprobp}{F_{\pfreeprojset}(\psi,t)}
\newcommand{\culprobpprime}{F_{\mathcal{FP}_{\mathsf{P}}^{p'}}(\psi,t)}
\newcommand{\threswitnessnewp}{\tilde{O}(\psi,t,p)}

\newcommand{\revise}[1]{{\color{black} #1}}

\begin{document}
	\title{Certifying localizable quantum properties with constant sample complexity}
	\author{Zhenyu Du}
	\affiliation{Center for Quantum Information, Institute for Interdisciplinary Information Sciences, Tsinghua University, Beijing 100084, China}
	
	\author{Jinchang Liu}
	\affiliation{Institute for Interdisciplinary Information Sciences, Tsinghua University, Beijing 100084, China}
	
	\author{Elias X.\ Huber}
	\affiliation{Yenching Academy, Peking University, Beijing 100871, China}
	
	\author{Zi-Wen Liu}
	\email{zwliu0@tsinghua.edu.cn}
	\affiliation{Yau Mathematical Sciences Center, Tsinghua University, Beijing 100084, China}
	\author{Xiongfeng Ma}
	\email{xma@tsinghua.edu.cn}
	\affiliation{Center for Quantum Information, Institute for Interdisciplinary Information Sciences, Tsinghua University, Beijing 100084, China}

	\begin{abstract} 
		Characterizing increasingly complex quantum systems is a central task in quantum information science, yet experimental costs often scale prohibitively with system size. Certifying key properties using simple local measurements is highly desirable but challenging. 
		In this work, we introduce a highly general certification framework based on a physical phenomenon that we call localizable quantumness: for generic many-body states, essential quantum properties are robustly preserved within the projected ensembles on small subsystems after performing local projective measurements on the rest of the system. 
		Leveraging this insight, we develop protocols to certify global properties---including multipartite entanglement, circuit complexity, and quantum magic---by witnessing them on a small, accessible subsystem.
		\revise{
			Remarkably, randomizing the local measurement bases extends this capability to certify state fidelity. 
			Relying solely on local Pauli measurements, these protocols achieve constant sample complexity and robustness for almost all quantum states, including a wide range of physically relevant states.
			For certifying the fidelity of $n$-qubit states, this $\cO(1)$ scaling dramatically improves upon state-of-the-art protocols requiring $\cO(n^4)$ samples.}
		Our unified framework provides both a practical toolkit for large-scale quantum certification and a novel lens into complex many-body systems.
	\end{abstract}
	
	\maketitle
	
	\section{Introduction}
	
	The rapid advancement of quantum technologies is enabling control over increasingly large and complex quantum systems, paving the way for next-generation quantum computers and simulators. This progress, however, brings a central challenge: how can we efficiently characterize these devices at scale? 
	While characterizing an unknown quantum system is generally hard~\cite{Haah2016OptimalTomo, Liu2022Fundamental}, in practical experiments one usually has good knowledge of the state-preparation procedure \cite{Wang2016TenPhotons, Cao2023GenerationGME, Kim2023QuantumUtility, Acharya2024QuantumEC}. 
	\emph{Quantum certification}—verifying a quantum system by comparing with prior knowledge of its state—provides a powerful approach \cite{Eisert2020CertificationReview, Kliesch2021CertificationSurvey, Carraso2021TheoreticalExperimental}. 
	This approach not only carries fundamental importance in quantum information science~\cite{Mayers2004SelfTesting}, but also underpins a broad range of applications, from probing novel quantum phenomena~\cite{Hoke2023MeasurementInducedEntanglement, Niroula2024MagicPhaseTransition} and device benchmarking \cite{Knill2008RandomizedBenchmarking, Arute2019Supremacy, Liu2022BenchmarkingSampling, Mark2023BenchmarkingErgodic, Choi2023RandomStateBenchmarking} to quantum cryptography \cite{Ekert1991CryptoBell, Reichardt2013ClassicalCommand, Liu2025CertifiedRandomness}.
	
	A central challenge in certification is to reduce measurement requirements and resource costs~\cite{Flammia2011DirectFidelityEstimation, Pallister2018OptimalVerification, Takeuchi2018VerificationManyQubit, Elben2020CrossPlatform, Zwerger2019DIGME, Rodriguez2021CertificationGME, Huang2024Certifying, Gupta2025SingleQubitCertification}. 
	As the system size grows, the exponential cost of standard characterization techniques like full state tomography becomes prohibitive.
	Another straightforward approach---projecting an experimental system onto the ideal state---is impractical, as it requires measurement capabilities comparable to the state preparation itself. 
	These challenges create an imperative need for scalable and experimentally feasible methods to certify quantum states and their essential properties.

	This need has motivated a shift toward certification protocols that leverage simple local measurements, which are readily available on most quantum platforms~\cite{Zwerger2019DIGME, Rodriguez2021CertificationGME, Huang2024Certifying, Gupta2025SingleQubitCertification, Li2025UniversalVerification}.
	The central difficulty, however, is that many crucial quantum properties, such as multipartite entanglement, are inherently global. 
	Prior attempts to certify these global quantumness with local measurements face crucial limitations: 
	Protocols are often designed case by case for specific states or properties~\cite{Zwerger2019DIGME, Rodriguez2021CertificationGME, Cao2023GenerationGME}, require demanding adaptive measurement schemes that are challenging for near-term devices~\cite{Gupta2025SingleQubitCertification, Li2025UniversalVerification}, or guarantee soundness only for pure states~\cite{Zwerger2019DIGME, Haug2023ScalableMeasureMagic, Huang2024Shallow, Landau2024Shallow}, severely limiting their applicability in realistic, noisy systems. 
	Furthermore, existing non-adaptive protocols for certifying state fidelity or generic properties typically suffer from sample complexity that grows with system size and a diminishing robustness~\cite{Huang2024Certifying}. This prompts a critical question: 
	\begin{center}
		\emph{Is there a general physical principle for robustly and efficiently certifying global quantum properties with simple local measurements?}
	\end{center}
	
	In this work, we address this question in the affirmative by introducing and formalizing a fundamental phenomenon we term localizable quantumness. Inspired by localizable entanglement~\cite{Verstraete2004LocalizedEntanglement, Popp2005LocalizableEntanglement}, we show that the underlying physical principle is far more general.  
	For generic many-body states, essential quantum properties including entanglement, circuit complexity, and magic are robustly encoded throughout the system. 
	Consequently, local projective measurements on a large complement do not destroy these properties, but rather concentrate them into the remaining small, accessible subsystem.
	The measurement-induced projected ensemble on the subsystem thereby carries a distinctive feature that serves as a local signature of the global quantum property.
	
	Leveraging this phenomenon, we develop a unified and highly efficient certification framework that applies broadly to diverse properties (see Table~\ref{tab:results_summary} for the properties considered in detail). 
	Our protocol first localizes the global property into a small subsystem via projective measurements on the large complement system, and then witnesses the localized property on the small subsystem. 
	The significant reduction in system size allows us to efficiently certify the quantumness of the subsystem, thereby substantially reducing experimental overhead and enabling the certification of entanglement, circuit complexity, and quantum magic for generic states using only local Pauli measurements.
	We prove that the method achieves i) constant sample complexity, ii) constant-level robustness against deviations from the target states, and iii) soundness for mixed states, thereby overcoming key limitations of previous approaches and achieving substantial improvement in sample complexity~\cite{Haug2023ScalableMeasureMagic, Huang2024Shallow, Landau2024Shallow, Kim2024LearningPhaseOfMatter, Niroula2024MagicPhaseTransition}.
	Our framework also improves the sample complexity from $\mathcal{O}(n)$ to $\mathcal{O}(\log n)$ for certifying complex properties such as fully inseparable entanglement~\cite{Eisert2006MultiparticleEntanglement, Rodriguez2021CertificationGME}.
	We further demonstrate its efficiency across various Hamiltonian systems and large-scale magic-injection circuits.

	\revise{
		Furthermore, extending this framework by measuring the large complement system in random Pauli bases demonstrates remarkable power in certifying global state fidelity.
		Using random matrix theory, we prove its constant sample complexity and constant robustness for certifying almost all quantum states. 
		We also establish these guarantees for random graph states via an error localization mechanism: Pauli errors acting on generic stabilizer states can be localized to small subsystems via local measurements. 
		This $\cO(1)$ scaling substantially improves upon the state-of-the-art non-adaptive local measurement protocol requiring $\cO(n^4)$ samples~\cite{Huang2024Certifying}.
		Taken together, our unified framework delivers optimal constant-sample-complexity scaling for both property and state certification tasks across generic many-body systems.}

	Our paper is organized as follows. 
	In Section~\ref{sec:summary}, we present the key ideas underlying our protocols and provide an overview of the main results. 
	Section~\ref{sec:protocol_framework} introduces the mathematical framework, formalizes the certification protocol, and analyzes its performance. 
	Sections~\ref{sec:complexity}, \ref{sec:certifying_entanglement}, and \ref{sec:certifying_magic} apply our methods to certifying circuit complexity, entanglement, and magic, respectively. 
	Section~\ref{sec:fidelity} develops the random-basis–enhanced protocol for fidelity certification. 
	Finally, Section~\ref{sec:discussion} concludes with broader implications and future directions.
	
	\begin{table*}[!htbp]
		\centering
		\caption{
			Summary of the quantum properties certified by our framework, together with the corresponding sample complexities \revise{and the representative state classes for which these sample complexity guarantees are proven or numerically demonstrated.}
		}
		\label{tab:results_summary}
		
		\small
		\renewcommand{\arraystretch}{1.18}
		
		\begin{tabular*}{\textwidth}{@{\extracolsep{\fill}}lll@{}}
			\toprule
			
			\TableL{0.25\textwidth}{
				\textbf{Property} 
			}
			&
			\TableC{0.2\textwidth}{
				\textbf{Sample complexity}
			}
			&
			\TableL{0.52\textwidth}{
				\revise{
					\textbf{Representative states}
				}
			}
			\\
			
			\midrule
			
			\TableL{0.25\textwidth}{
				(Measurement-assisted) circuit complexity $d$
				in $D$ spatial dimensions
			}
			&
			\TableC{0.16\textwidth}{
				$\exp\bigl(\cO(d^D)\bigr)$
			}
			&
			\TableL{0.52\textwidth}{
				\revise{
					Deep-thermalized states
					(Proposition~\ref{prop:gap_deep_thermalize})
					\newline
					Random brickwork-circuit states
					(Proposition~\ref{prop:brickwork_state})
				}
			}
			\\
			
			\addlinespace[0.25em]
			\hline
			\addlinespace[0.25em]

			\TableL{0.25\textwidth}{
				Bipartite entanglement
			}
			&
			\TableC{0.16\textwidth}{
				$\cO(1)$
			}
			&
			\TableL{0.52\textwidth}{
				\revise{
					Deep-thermalized states
					(Proposition~\ref{prop:gap_deep_thermalize})
					\newline
					Random brickwork-circuit states
					(Proposition~\ref{prop:brickwork_state})
					\newline
					Practical Hamiltonian ground states\textsuperscript{*}
					(Sec.~\ref{subsec:HamiltonianEntanglement})
				}
			}
			\\
			
			\addlinespace[0.25em]
			\hline
			\addlinespace[0.25em]
			
			\TableL{0.25\textwidth}{
				Fully inseparable entanglement
			}
			&
			\TableC{0.16\textwidth}{
				$\cO(\log n)$
			}
			&
			\TableL{0.52\textwidth}{
				\revise{
					Haar-random states
					(Sec.~\ref{subsec:fully-inseparable};
					Lemma~\ref{lem:Haar-random-states-localizable-entanglement})
					\newline
					Random stabilizer states\textsuperscript{*}
					(Sec.~\ref{subsec:fully-inseparable})
				}
			}
			\\
			
			\addlinespace[0.25em]
			\hline
			\addlinespace[0.25em]
			
			\TableL{0.25\textwidth}{
				Quantum magic
			}
			&
			\TableC{0.16\textwidth}{
				$\cO(1)$
			}
			&
			\TableL{0.52\textwidth}{
				\revise{
					Haar-random states (Sec.~\ref{sec:certifying_magic}; Lemma~\ref{lem:Haar-random-states-localizable-magic}) 
					\newline
					Magic-injection-and-scrambling circuits\textsuperscript{*}
					(Sec.~\ref{sec:certifying_magic})
				}
			}
			\\
			
			\addlinespace[0.25em]
			\hline
			\addlinespace[0.25em]
			
			\TableL{0.25\textwidth}{
				State fidelity
			}
			&
			\TableC{0.16\textwidth}{
				$\cO(1)$
			}
			&
			\TableL{0.52\textwidth}{
				\revise{
					Random graph states
					(Theorem~\ref{thm:constant-gap-graph-states})
					\newline
					Haar-random states
					(Theorem~\ref{thm:large_spectral})
					\newline
					Random brickwork-circuit states\textsuperscript{*}
					(Sec.~\ref{subsec:fidelity-Haar-random})
				}
			}
			\\
			
			\bottomrule
		\end{tabular*}
		
		\vspace{0.35em}
		\parbox{\textwidth}{%
			\raggedleft
			\footnotesize
			\revise{
				\textsuperscript{*} Numerics.\par
			}
		}
		
	\end{table*}
	
	\section{Overview of Methods and Results}\label{sec:summary}
	This work presents a unified and practical framework for quantum certification based solely on local Pauli measurements. 
	In certification tasks, one typically has reliable prior knowledge of the quantum system. Specifically, a pure target state $\psi_{AB}$ is known (for example, from a noiseless quantum circuit description), and the goal is to certify whether the experimentally realized mixed state $\rho_{AB}$ possesses certain desired properties. 
	We partition the system into a small subsystem $A$ and its large complement $B$. 
	\revise{Performing computational basis measurements on $B$ generates a measurement-induced ensemble on $A$, commonly referred to as a \emph{projected ensemble}.
		Specifically, the experimental state yields } $\mathcal{E}_\rho = \{p_\rho(\bmz), \rho_\bmz\}$ on $A$, while the target state generates $\mathcal{E}_\psi = \{p_\psi(\bmz), \psi_\bmz\}$, as illustrated in Fig.~\ref{fig:projected_state_set}a. 
	\revise{While such ensembles have recently attracted substantial attention for probing complex dynamics and deep thermalization~\cite{Choi2023RandomStateBenchmarking, Cotler2023EmergentDesign, Ho2022EmergentChaotic, Ippoliti2023DynamicalPurification, Chang2024ChargeConserving, Mark2024MaximumEntropy}, we reveal a powerful new capability: these projected ensembles can serve as robust local signatures of the original global quantum properties. }
	By locally comparing $\rho_\bmz$ and $\psi_\bmz$ on the small subsystem $A$, our approach efficiently certifies the global quantum properties of $\rho$. We demonstrate that a broad class of quantum properties can be certified in this way with constant sample complexity. This yields highly efficient certification protocols for complex quantum structures, such as nontrivial measurement-equivalent quantum phases and fully inseparable entanglement.

	\begin{figure}[!h]
		\centering
		\includegraphics[width=.48\textwidth]{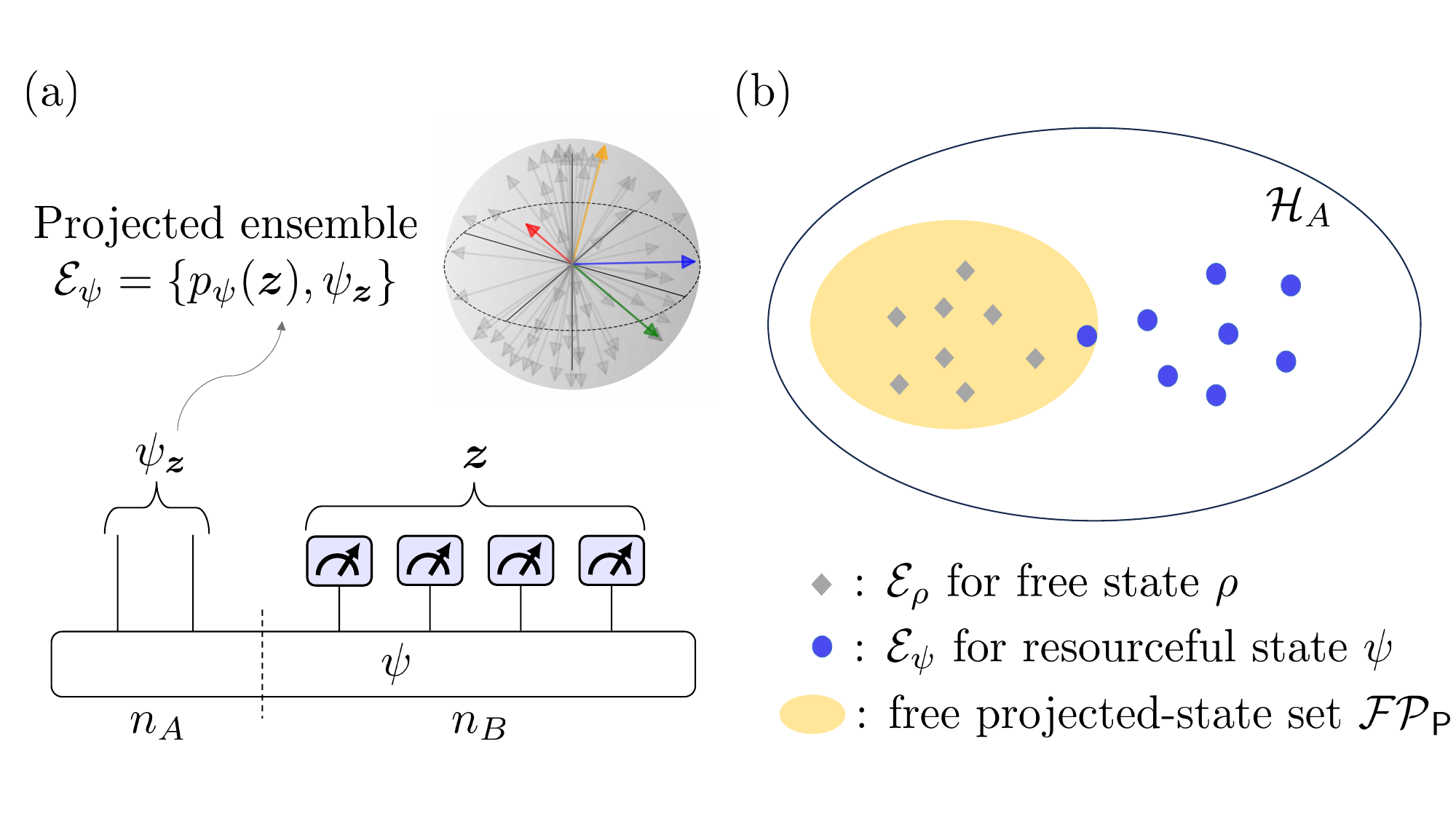}
		\caption{Localizable quantumness within projected states. 
			(a) Projected ensemble: The target state is partitioned into subsystems $A$ and $B$. Measuring $B$ in the computational basis yields an outcome $\bmz$ and a corresponding projected state $\psi_{\bmz}$ on $A$ with probability $p_{\psi}(\bmz)$, thereby defining the projected ensemble $\mathcal{E}_{\psi}$. 
			(b) Localizable quantumness: If the target state exhibits localizable quantumness, then its projected ensemble $\mathcal{E}_{\psi}$ differs significantly from the projected ensemble of any free state $\rho$. For instance, most projected states lie outside the free projected-state set $\freeprojset$. This difference can be utilized to certify the quantumness of the target state by examining the small subsystem $A$.}
		\label{fig:projected_state_set}
	\end{figure}
	
	Our certification framework relies on two key ingredients: (1) establishing that many quantum properties are localizable, meaning they concentrate into the small subsystem $A$ once the complement $B$ is measured, and (2) converting this observation into an efficient certification protocol by introducing and measuring a quantity we term the conditional fidelity \revise{(see Protocol~\ref{prot:certification})}. Together, these ingredients yield our main Result~\ref{thm:performance_certification_protocol_informal}, which guarantees that any such property can be certified both efficiently and robustly. We now outline these two ingredients, with details provided in Sec.~\ref{sec:protocol_framework}.
	
	The central concept of our approach is \emph{localizable quantumness}. For a broad class of physically relevant many-body states, essential quantum properties of the global system remain robustly encoded in small subsystems even after local projective measurements on the complement. Concretely, if the global target state $\psi$ possesses a localizable quantum property $\mathsf{P}$ (see Table~\ref{tab:results_summary} for examples)
	, then the projected ensemble $\mathcal{E}_\psi$ necessarily differs from $\mathcal{E}_\rho$ for any state $\rho$ that lacks $\mathsf{P}$, as shown in Fig.~\ref{fig:projected_state_set}b. This difference represents a local signature of the global property, thereby transforming the otherwise challenging task of global certification into a tractable local problem: certifying the properties of the projected ensemble $\mathcal{E}_\rho$.
	
	To harness localizable quantumness, we avoid the intractable task of directly evaluating projected-ensemble–averaged nonlinear quantities, such as localizable entanglement~\cite{Verstraete2004LocalizedEntanglement, Popp2005LocalizableEntanglement}. Instead, we compare the projected ensemble of the experimental state with that of the target. To this end, we introduce the notion of \emph{conditional fidelity}, which evaluates the local fidelity between each projected state $\rho_{\bmz}$ and its corresponding $\psi_{\bmz}$, conditioned on the measurement outcome $\bmz$ on $B$. 
	A naive strategy—estimating each fidelity $\tr(\psi_{\bm z}\rho_{\bm z})$—is inefficient because we have only single-copy access to $\rho_{\bm z}$ without post-selection. 
	In contrast, we estimate the ensemble-averaged local fidelity $\sum_{\bmz}p_{\rho}(\bmz) \tr(\psi_{\bm z}\rho_{\bm z})$ via randomized local Pauli measurements on the small subsystem $A$.
	\revise{This is achieved using the post-selection-free local shadow primitive introduced in Ref.~\cite{McGinley2024PostselectionFreeLearning} and also employed in the shadow-overlap protocol~\cite{Huang2024Certifying}}.
	By estimating the conditional fidelity, we can certify whether the property $\mathsf{P}$ is present or the state lies in the free set $\mathcal{F}_{\mathsf{P}}$ (states lacking $\mathsf{P}$), thereby turning the localizable quantumness into a robust witness of global quantum properties. In summary, we establish the following:
	
	\begin{result}[Certifying localizable quantumness, informal]\label{thm:performance_certification_protocol_informal}
		For any property $\mathsf{P}$ with associated free set $\mathcal{F}_{\mathsf{P}}$, if the target state $\psi$ exhibits localizable quantumness on a small subsystem $A$, then the protocol achieves: 
		\begin{enumerate}
			\item Using non-adaptive local measurements; 
			\item Soundness: all (possibly mixed) $\rho \in \mathcal{F}_{\mathsf{P}}$ are rejected with high probability;
			\item Completeness (Robustness): any $\rho$  within a constant trace distance of $\psi$ is accepted with high probability;
			\item Constant sample complexity: the number of required copies is $\mathcal{O}(1)$ for constant subsystem size $|A|= \cO(1)$.
		\end{enumerate}
	\end{result}

	The requirement of localizable quantumness is highly general~\cite{Verstraete2004LocalizedEntanglement, Popp2005LocalizableEntanglement, Zwerger2019DIGME, Ho2022EmergentChaotic, Choi2023RandomStateBenchmarking, Cotler2023EmergentDesign, Du2025SpacetimeComplexity, Niroula2024MagicPhaseTransition}. 
	\revise{As summarized in Table~\ref{tab:results_summary},  we establish both analytically and numerically that this phenomenon is indeed a universal feature across a wide range of quantum many-body systems.
		
		Built upon this ubiquitous physical phenomenon, our certification framework achieves powerful performance guarantees.}
	First, the guarantee of soundness for mixed states marks a significant advance over previous approaches, which often applied only to pure states when certifying complex quantum properties~\cite{Huang2024Shallow,Landau2024Shallow, Haug2023ScalableMeasureMagic, Kim2024LearningPhaseOfMatter}. 
	Second, the constant robustness also represents an important advantage over protocols whose robustness diminishes with system size~\cite{Huang2024Certifying, Gupta2025SingleQubitCertification}, ensuring practical certification in realistic noisy settings. 
	Finally, the guarantee of constant sample complexity yields substantial improvements over prior methods, as detailed below. 
	Combined with the use of only local Pauli measurements, this ensures experimental efficiency and scalability. 
	Taken together, these features establish a versatile and powerful framework for certifying diverse quantum properties.

	\textit{Complexity certification.}  
	A highlight of our results is the first explicit construction of complexity witnesses. 
	In Sec.~\ref{sec:complexity}, we show that circuit complexity, i.e., the minimal circuit depth required to prepare a state can be certified to exceed a constant threshold with constant sample complexity and robustness. 
	The key idea is to certify reliable signatures of circuit complexity, such as high entanglement, on a carefully chosen small subsystem. 
	This approach also extends to certifying \emph{measurement-assisted} circuit complexity, which can be exponentially smaller than standard state complexity and is therefore more challenging to certify~\cite{Piroli2021MeasurementAssistedQuantumCircuits, Lu2022MeasurementShortcut}. 
	Our framework thus opens a new pathway for probing intricate structures in many-body quantum systems. 
	As two important examples, it enables efficient certification of long-range entanglement~\cite{Chen2010LongRangeEntanglement, Wen2013topological} and nontrivial measurement-equivalent quantum phases~\cite{Tantivasadakarn2023HierarchyMeasurement}.  
	\begin{corollary}[Certifying quantum circuit complexity, informal]
		If the target state exhibits high localizable entanglement on a specific small subsystem, then constant (measurement-assisted) quantum circuit complexity can be certified with constant sample complexity while achieving constant robustness. This enables certification of long-range entanglement and nontrivial measurement-equivalent quantum phases with $\omega(1)$ sample complexity\footnote{In this work, ``$\omega(1)$ sample complexity'' means $\cO(f(n))$ samples for any function $f(n)=\omega(1)$.}.
	\end{corollary}
	\revise{
		Concretely, these performance guarantees are established for deep-thermalized states~\cite{Choi2023RandomStateBenchmarking, Cotler2023EmergentDesign, Ho2022EmergentChaotic, Ippoliti2023DynamicalPurification, Chang2024ChargeConserving, Mark2024MaximumEntropy} and physically relevant random brickwork-circuit states (see Table~\ref{tab:results_summary}).}
	
	\textit{Entanglement certification.}  
	In Sec.~\ref{sec:certifying_entanglement}, we show that our method enables highly efficient certification of entanglement. For bipartite entanglement, we show that localizable entanglement on a small subsystem suffices to certify entanglement across exponentially many bipartitions with constant sample complexity. 
	More strikingly, we establish a highly efficient method for certifying \emph{fully inseparable entanglement}, namely entanglement across all possible bipartitions~\cite{Eisert2006MultiparticleEntanglement}, using only logarithmic sample complexity. Fully inseparable entanglement is a crucial resource for quantum error correction~\cite{Rodriguez2021CertificationGME} and for the activation of genuine multipartite entanglement~\cite{Palazuelos2022FullyInseparableGMEActivation}, making our certification protocol especially valuable. 
	\begin{corollary}[Certifying fully inseparable entanglement, informal]
		If the target state exhibits localizable entanglement on certain small subsystems, then fully inseparable entanglement can be certified with $\mathcal{O}(\log n)$ sample complexity while achieving constant robustness. 
	\end{corollary}
	\revise{
		Beyond establishing the performance of this entanglement certification method for deep-thermalized states and random brickwork-circuit states, we also numerically demonstrate its efficiency across random stabilizer states and various Hamiltonian systems (see Table~\ref{tab:results_summary}).}
	
	\textit{Magic certification.}  
	We further apply our framework to quantum magic (or non-stabilizerness) certification in Sec.~\ref{sec:certifying_magic}, with the free set consisting of convex mixtures of all stabilizer states. Since local Pauli measurements are free operations in the magic resource theory and preserve stabilizerness, the presence of localizable quantum magic after measuring a large complement in the computational basis directly certifies the magic of the target state. 
	\begin{corollary}[Certifying quantum magic, informal]
		If the target state exhibits localizable magic on a small subsystem, then quantum magic can be certified with $\mathcal{O}(1)$ sample complexity while achieving constant robustness. 
	\end{corollary}
	
	\revise{
		These performance guarantees are established for typical Haar-random states.
		Furthermore, because magic is a fundamental resource for fault-tolerant quantum computing~\cite{Bravyi2005Universal}, we numerically demonstrate the scalability of our approach on magic-injection circuits for systems up to 512 qubits.}
	To the best of our knowledge, this is the first protocol to certify quantum magic using only single-copy local measurements while achieving soundness for mixed states. 
	Consequently, our framework offers the potential for exponential improvements in sample complexity for relevant experiments, such as detecting magic phase transitions~\cite{Niroula2024MagicPhaseTransition} and benchmarking near-term fault-tolerant demonstrations.
	
	\textit{Fidelity certification.}  
	Finally, we show that global state fidelity is also robustly encoded within small subsystems in Sec.~\ref{sec:fidelity}. 
	Specifically, we propose a fidelity certification protocol that compares local fidelities on a small subsystem after measuring the large complement in local Pauli bases.
	\revise{
		Our protocol achieves constant sample complexity and robustness for almost all quantum states, including random graph states.
	}
	Since every stabilizer state is local-Clifford equivalent to a graph state, this yields the first proof that generic stabilizer states can be efficiently certified using solely local random Pauli measurements.
	\revise{This constant scaling dramatically improves upon state-of-the-art non-adaptive methods,} whose costs grow \revise{polynomially} with system size and whose robustness against state imperfection diminishes~\cite{Huang2024Certifying, Gupta2025SingleQubitCertification}, as well as over adaptive local-measurement approaches that require challenging classical feedback and long-lived quantum memory to store states between measurement rounds~\cite{Gupta2025SingleQubitCertification, Li2025UniversalVerification, Coladangelo2026TwoBases} \revise{(see Sec.~\ref{sec:fidelity} for a detailed comparison)}.
	
	\begin{result}[Certifying state fidelity, informal]
		Conditional fidelity robustly certifies the global fidelity of typical pure target states. 
		\revise{Specifically, the protocol provably achieves constant sample complexity and constant robustness for almost all quantum states (Theorem~\ref{thm:large_spectral}), and for random graph states (Theorem~\ref{thm:constant-gap-graph-states}).}
	\end{result}
	
	\section{Mathematical framework and protocol}\label{sec:protocol_framework}
	Having outlined the key ideas and results of our work, we now present the certification framework and protocol in detail. We begin by formalizing the notion of localizable quantumness, which serves as the foundation of our approach. Building on this, we show how to construct property witnesses through conditional fidelity. Then, we present the complete protocol scheme, together with an analysis of its performance guarantees.

	\subsection{Localizable quantumness and conditional fidelity}
	We partition the $n$-qubit system into a small subsystem $A$ (with $n_A$ qubits) and its complement $B$ (with $n_B$ qubits, so that $n_A + n_B = n$). In the computational basis of $B$, the target state can be written as
	\begin{equation}
		\ket{\psi} =\sum_{\zsum}
		\sqrt{p_\psi(\bm z)}\,
		\ket{\psi_{\bm z}}_A\otimes\ket{\bm z}_B,
	\end{equation}
	which induces the projected ensemble $\mathcal{E}_{\psi} \coloneqq \{p_{\psi}(\bmz), \psi_{\bmz}\}$ (see Fig.~\ref{fig:projected_state_set}a). For a (possibly mixed) experimental state $\rho$, computational-basis measurements on $B$ yields outcome $\bmz$ with a projected state on $A$,
	\begin{equation}\label{eq:projected_state_rho}
		\rho_{\bmz} \coloneqq \frac{\tr_B[(I_A \otimes \ketbra{\bmz}_B) \rho]}{p_{\rho}(\bmz)}.
	\end{equation}
	with probability $p_{\rho}(\bmz) \coloneqq \tr[(I_A \otimes \ketbra{\bmz}_B) \rho]$. 
	
	To quantify the strength of localizable quantumness for a property $\mathsf{P}$, we first specify its associated free set $\freeset \subseteq \mathcal{D}(\mathcal{H}_{AB})$, consisting of all states that do not possess $\mathsf{P}$ (e.g., the separable states when $\mathsf{P}$ is entanglement). Local projective measurements on subsystem $B$ maps each $\rho\in\freeset$ to a family of projected states on $A$, thereby inducing the free projected-state set
	\begin{equation}\label{eq:free_proj_set}
		\freeprojset \coloneqq \{\rho_{\bm{z}} : \rho \in \freeset,\; p_{\rho}(\bm{z}) \neq 0\}.
	\end{equation}
	For the target state $\psi$, we then define the maximal fidelity at each outcome $\bm{z}$ by
	\begin{equation}
		\mathrm{Fid}_{\mathsf{P}}(\psi_{\bm{z}}) \coloneqq \sup_{\sigma \in \freeprojset}\,\mathrm{tr}(\sigma\,\psi_{\bm{z}}).
	\end{equation}
	Small values of $\mathrm{Fid}_{\mathsf{P}}(\psi_{\bm{z}})$ indicates strong localizable quantumness—i.e., $\psi_{\bm{z}}$ remains far from all property-free projected states on $A$ (Fig.~\ref{fig:projected_state_set}b). Averaging over outcomes yields our localizable-quantumness metric
	\begin{equation}\label{eq:cond_infid}
		\mathrm{LQ}_{\mathsf{P}}(\psi)\coloneqq \sum_{\bm{z}\in\{0,1\}^{n_B}} p_{\psi}(\bm{z})\,\bigl[1-\mathrm{Fid}_{\mathsf{P}}(\psi_{\bm{z}})\bigr].
	\end{equation}
	This quantity plays a central role in our certification protocols, serving as the gap that separates the target state from all property-free states (Lemma~\ref{lem:eta_witness}).
	Unlike localizable entanglement~\cite{Verstraete2004LocalizedEntanglement}, we do not optimize over measurement bases on $B$, as fixing a simple local measurement scheme minimizes experimental requirements.
	
	We also consider random-basis measurements on $B$: independently for each qubit, randomly choose a Pauli measurement basis $X/Y/Z$. Equivalently, the single-qubit POVM on $B$ is
	\begin{equation}\label{eq:POVM}
		\left\{\frac{1}{3}\ketbra{0},\, \frac{1}{3}\ketbra{1},\, \frac{1}{3}\ketbra{+},\, \frac{1}{3}\ketbra{-},\, \frac{1}{3}\ketbra{+i},\, \frac{1}{3}\ketbra{-i}\right\}.
	\end{equation}
	Let $\bm x\in\{0,1,+,-,+i,-i\}^{n_B}$ (abbreviated as $\mathsf{X}^{n_B}$) denote the outcome string. The projected state on $A$ and its probability are
	\begin{equation}
		\begin{split}
			\ket{\psi_{\bm x}}
			&=\frac{(I_A \otimes \bra{\bm x}_B)\ket{\psi}}{\sqrt{3^{n_B} \tilde{p}_{\psi}(\bmx)}}, \\
			\tilde{p}_\psi(\bm x)&=3^{-n_B}\,\mathrm{tr}\bigl[(I_A\otimes\ketbra{\bm x}_B)\,\psi\bigr].
		\end{split}
	\end{equation}
	The random-basis metric of localizable quantumness is then
	\begin{equation}
		\widetilde{\mathrm{LQ}}_{\mathsf{P}}(\psi)\coloneqq \sum_{\bm x\in\Bxrange} \tilde{p}_{\psi}(\bm x)\,\bigl[1-\mathrm{Fid}_{\mathsf{P}}(\psi_{\bm x})\bigr].
	\end{equation}
	This variant preserves minimal experimental requirements while enabling extensive data reuse and flexible post-processing from a single measurement dataset. 
	These advantages will be leveraged in our entanglement certification scheme. In what follows, we present the protocol under computational-basis measurements on $B$, and the extension to the random-basis setting is straightforward.
	
	We certify the target property $\mathsf{P}$ by comparing the projected ensemble of the experimental state $\rho$ with that of the target $\psi$. Formally, we define the observable
	\begin{equation}
		\begin{split}
			\propwitness &\coloneqq \sum_{\bmz} [\psi_{\bmz} - \fid(\psi_{\bm z}) I] \otimes \ketbra{\bmz} \\
			\eta_{\psi}(\rho) &\coloneqq \tr(\propwitness \rho) \\
			&= \sum_{\bmz} p_{\rho}(\bmz) [\tr(\psi_{\bmz} \rho_{\bmz}) - \fid(\psi_{\bm z})]
		\end{split}
	\end{equation}
	We refer to such an ensemble-averaged quantity as a conditional fidelity, since it averages local fidelities conditioned on measurement outcomes of $B$. Each local fidelity is shifted by the constant $\fid(\psi_{\bm z})$, ensuring that the operator $\psi_{\bm z}-\fid(\psi_{\bm z})I$ acts as a valid property witness, taking positive expectation only for projected states lying outside the free set $\freeprojset$. 
	
	\begin{lemma}[$\eta_{\psi}(\rho)$ is a valid witness]\label{lem:eta_witness}
		The conditional fidelity $\eta_{\psi}$ satisfies: 
		\begin{enumerate}
			\item Soundness: If $\rho\in\freeset$, then $\eta_{\psi}(\rho)\le 0$.
			\item Completeness: For the target $\psi$, $\eta_{\psi}(\psi)= \LQ(\psi)$.
		\end{enumerate}
		Consequently, the gap $ \LQ(\psi)$ separates $\psi$ from all property-free states, making $\eta_{\psi}$ a valid witness for $\mathsf{P}$.
	\end{lemma}
	
	\begin{proof}
		If $\rho\in\freeset$, then for every $\bm z$ with $p_\rho(\bm z)>0$ we have $\rho_{\bm z}\in\freeprojset$. The definition of maximal fidelity implies $\tr\bigl(\psi_{\bm z}\rho_{\bm z}\bigr) - \fid(\psi_{\bmz}) \le 0$. Averaging over outcomes yields $\eta_{\psi}(\rho)\le 0$.
		For the target state $\psi$,
		\begin{equation}\label{eq:target_value}
			\tr(\propwitness \psi)=\sum_{\bmz} p_{\psi}(\bmz)[1-\fid(\psi_{\bmz})].
		\end{equation}
		which equals $\LQ(\psi)$ by definition.
	\end{proof}
	
	By continuity, $\eta_{\psi}$ also certifies $\mathsf{P}$ for any state $\rho$ sufficiently close to $\psi$ (Fig.~\ref{fig:conditional_fidelity}b). 
	Moreover, replacing $\fid(\psi_{\bm z})$ in $\propwitness$ with suitable upper bounds $f(\psi_{\bmz})$ also suffices for certification. 
	This slightly reduce the gap in Lemma~\ref{lem:eta_witness} to
	\begin{equation}\label{eq:delta_upper_bounds}
		\Delta = \sum_{\bmz} p_{\psi}(\bmz) [ 1- f(\psi_{\bmz})],
	\end{equation}
	but in practice, such bounds are often easier to compute and readily available. 
	
	\begin{figure*}
		\centering
		\includegraphics[width=0.7\linewidth]{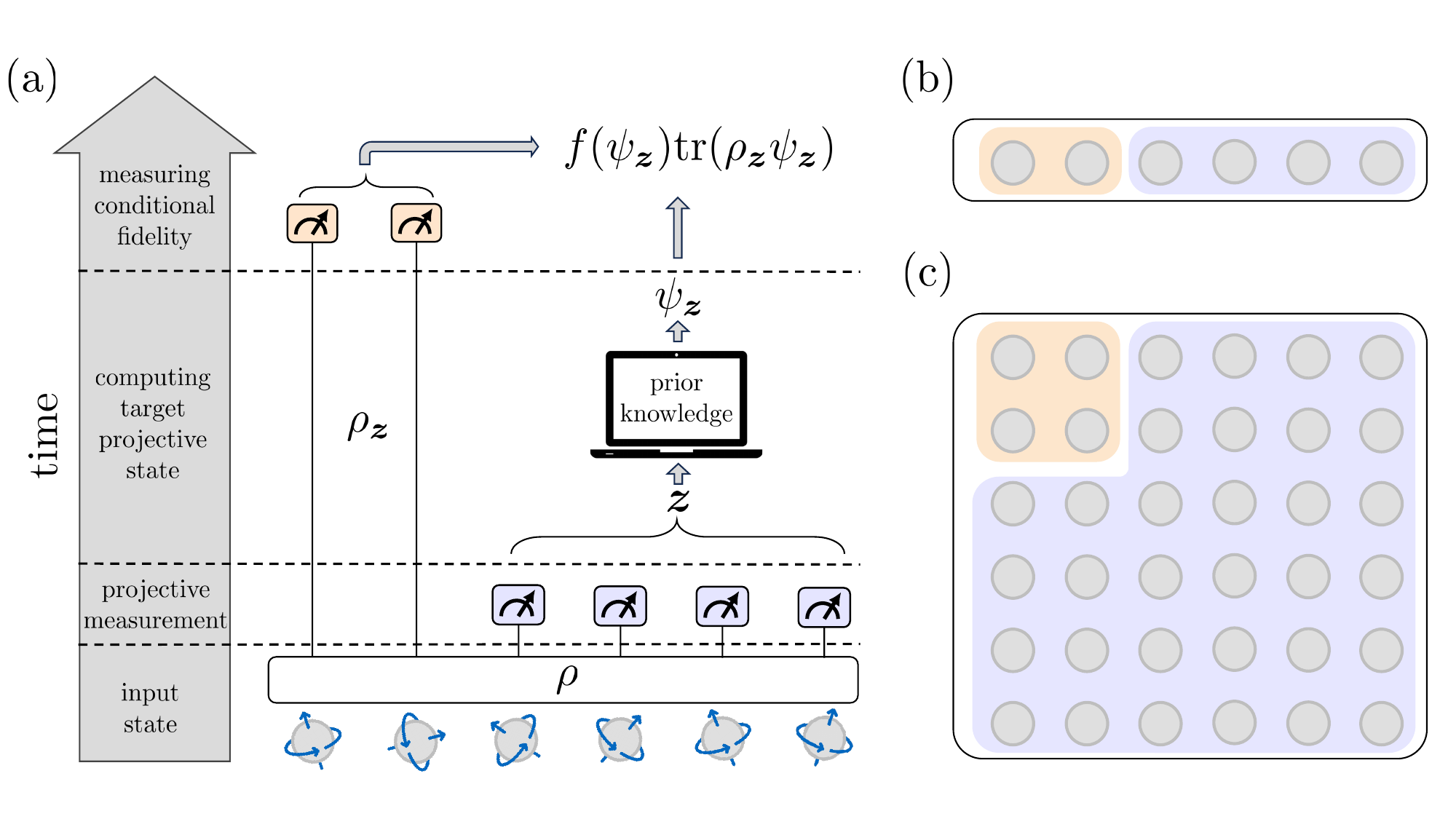}
		\caption{Certifying quantum properties via conditional fidelity.  (a) \revise{Workflow of Protocol~\ref{prot:certification}}.  Computational-basis measurements (shown in blue) are performed on a subsystem of $\rho$, yielding an outcome $\bm z$ and the corresponding projected state $\rho_{\bm z}$.  Randomized Pauli measurements (shown in orange) are then applied to $\rho_{\bm z}$.  Leveraging prior knowledge of $\psi_{\bm z}$, we post-process the measurement data to obtain estimators of $\tr(\psi_{\bmz} \rho_{\bmz})$.  (b) Decision rule and robustness. By selecting an appropriate threshold $\eta^*$ and estimating the conditional fidelity $\eta_{\psi}$ from the measurement data, states $\rho$ that are close to $\psi$ are accepted, while property-free states are rejected.
		}
		\label{fig:conditional_fidelity}
	\end{figure*}
	
	\subsection{Protocol scheme and performance}\label{subsec:measurement_scheme}
	We now turn the above construction into a practical, experimentally implementable protocol and analyze its performance (see Appendix~\ref{app:prelim} for proofs).
	Protocol~\ref{prot:certification} details our protocol scheme, which proceeds in three stages: (i) Choose a partition $A \cup B$ such that the subsystem $A$ exhibits large localizable quantumness. (ii) Perform measurements on $\rho$ and record the measurement outcomes. (iii) Using the measurement record together with prior knowledge of $\psi$, compute the estimator of $\eta_{\psi}(\rho)$ and decide whether $\rho$ possesses the target property.
	
	\begin{table}[!htbp]
		\begin{protocol}
			{Protocol for certifying quantum properties}
			{prot:certification}
			
			\small
			
			\textbf{Input:} Quantum state $\rho$; classical description of a target state
			$\psi$; failure probability $\delta\in(0,1)$.
			
			\textbf{Output:} $\accept$ or $\reject$.
			\par
			
			\medskip
			\textbf{Pre-computation:}
			\begin{protocolsteps}
				\item Find a partition $A\cup B$.
			\end{protocolsteps}
			
			\smallskip
			\textbf{Experiment:}
			\begin{protocolsteps}[start=2]
				\item
				Measure subsystem $B$ of $\rho$ in the computational basis,
				obtaining an outcome $\bm z$.
				
				\item
				Measure subsystem $A$ of $\rho_{\bm z}$ in a random Pauli basis,
				obtaining an outcome $\bm x$.
				
				\item
				Repeat Steps~2 and~3 for $T$ rounds, obtaining
				$\{(\bm z_i,\bm x_i)\}_{i=1}^{T}$.
			\end{protocolsteps}
			
			\smallskip
			\textbf{Post-processing:}\par
			\begin{protocolsteps}[start=5]
				\item
				Compute the estimator $\omega_{i}$ of $\tr(\psi_{\bmz_{i}}\rho_{\bmz_{i}})$ for $1 \le j \le T$ using local classical shadow, and compute the median-of-means estimator $\omega$ using $\{\omega_{i} - \fid(\psi_{\bmz_i})\}_{j=1}^{T}$.
				
				\item
				Set $\eta^\ast = \LQ(\psi)/3$. If $\omega \le \eta^\ast$, output $\reject$; otherwise, output $\accept$. 
			\end{protocolsteps}
		\end{protocol}
	\end{table}
	
	In the protocol, we need to evaluate $\LQ(\psi)$ for the chosen partition and $\fid(\psi_{\bm z})$ for the sampled projected states. 
	Rather than exhaustively enumerating the full projected ensemble $\mathcal{E}_\psi$, one can reliably estimate $\LQ(\psi)$ by sampling only a small subset of projected states. 
	As discussed, it is not necessary to compute $\fid(\psi_{\bm z})$ exactly; replacing it with suitable upper bounds $f(\psi_{\bm z})$ suffices. 
	This substitution reduces the gap from $\LQ(\psi)$ to a smaller value $\Delta$ (Eq.~\eqref{eq:delta_upper_bounds}), and we may adjust the threshold to $\eta^* = \Delta/3$ accordingly. 
	In all applications considered here, subsystem $A$ is of constant size, which makes the computation of $\fid(\psi_{\bm z})$ or its upper bound efficient. 
	Throughout this section, we therefore treat the partition $A \cup B$ together with the values of $\LQ(\psi)$ and $\fid(\psi_{\bm z})$ as given, deferring their explicit construction to the discussion of specific applications.
	
	In each experimental round $i$, we first measure subsystem $B$ in the computational basis $\{\ket{\bm z}\}_{\bm z\in\{0,1\}^{n_B}}$, obtaining a bit string $\bm z_i$ drawn from $p_\rho(\bm z)$. Conditional on this outcome, we then measure each qubit of $A$ in an independently and uniformly chosen random single-qubit Pauli basis using the POVM of Eq.~\eqref{eq:POVM}, recording the outcome $\bm x_i\in\Axrange$. Repeating for $T$ rounds yields the dataset $\{(\bm z_i,\bm x_i)\}_{i=1}^{T}$.

	In the post-processing step, we use the collected data to evaluate $\eta_{\psi}(\rho)$.  
	The informationally complete POVM on subsystem $A$ enables an unbiased estimator $\omega_i\coloneqq \tr\Bigl\{\psi_{\bmz_i} \Bigl[\bigotimes_{j \in A} (3\ketbra{\bmx_{i,j}} - I_j)\Bigr]\Bigr\}$ of $\tr(\psi_{\bmz_i}\rho_{\bmz_i})$ using \revise{local classical shadow~\cite{Huang2020Predicting, McGinley2024PostselectionFreeLearning, Huang2024Certifying}}. The final estimator $\omega$ is given by applying the median-of-means procedure to $T$ estimators $\{\omega_i - \fid(\psi_{\bmz_i})\}_{i=1}^{T}$. 
	
	Finally, we certify the property by comparing the estimator to $\eta^\ast = \LQ(\psi)/3$. This decision threshold is placed between the maximal free-state value ($\le 0$) and the ideal value $\LQ(\psi)$, leaving margins of $\LQ(\psi)/3$ above $0$ and $2\LQ(\psi)/3$ below $\LQ(\psi)$. These margins absorb statistical fluctuations and experimental noise. Our protocol is summarized in Fig.~\ref{fig:conditional_fidelity}.
	We have the following performance guarantee: 
	\begin{theorem}[Performance of Protocol~\ref{prot:certification}, formal version of Result \ref{thm:performance_certification_protocol_informal}]\label{thm:performance_certification_protocol}
		Given a property $\mathsf{P}$, a target state $\psi$, and a bipartition $A\cup B$. Protocol~\ref{prot:certification} satisfies
		\begin{enumerate}
			\item (Soundness): If the input state $\rho\in\freeset$, the protocol outputs $\reject$ with probability at least $1-\delta$.
			\item (Completeness): If $\rho$ satisfies $\dtr(\psi,\rho)\le\varepsilon$ with $\varepsilon =\LQ(\psi) / 6$, the protocol outputs $\accept$ with probability at least $1-\delta$.
		\end{enumerate}
		Moreover, the protocol performs local Pauli measurements on $\rho$ with sample complexity 
		\begin{equation}\label{eq:T_value}
			T = \frac{243 \ln(\delta^{-1}) \sigma^2}{\LQ^2(\psi)},
		\end{equation}
		where $\sigma^{2}=4^{n_{A}}+1$.  
		
		When $n_{A}=\cO(1)$, $\LQ(\psi)=\Omega(1)$, the sample complexity is $T=\cO(1)$, while achieving $\varepsilon =\Omega(1)$ robustness.
	\end{theorem}
	Therefore, strong localizable quantumness $\LQ(\psi)$ on a small subsystem guarantees sample-efficient certification, validating Result~\ref{thm:performance_certification_protocol_informal}. 
	This condition holds for a broad class of states, which is established in subsequent sections.
	
	An important benefit of our protocol is its constant robustness against state deviation. 
	Specifically, our scheme accepts any input state $\rho$ within $\varepsilon=\Omega(1)$ trace distance of the target state $\psi$.
	Compared to certification methods that require the experimental state to be $\varepsilon$-close to the target with $\varepsilon \sim 1/\poly(n)$~\cite{Huang2024Certifying, Gupta2025SingleQubitCertification}, the constant robustness of our method significantly relaxes experimental requirements.
	Given the high gate fidelities achieved by current platforms~\cite{Abanin2025OTOC, Kretschmer2025UnconditionalSeparation}, this robustness effectively enlarges the certifiable circuit volume by a substantial factor. 
	Furthermore, in the early fault-tolerant regime with quantum error correction techniques, our protocol serves as a practical tool for certifying logical preparation procedures, such as magic injection and cultivation.
	
	Moreover, our protocol reliably rejects any mixed state in $\freeset$, offering an advantage over previous protocols that detect properties only for pure input states, such as those for circuit complexity~\cite{Huang2024Shallow,Landau2024Shallow} and magic~\cite{Haug2023ScalableMeasureMagic}.
	The decision boundary $\eta^\ast$ can also be tuned closer to the property-free value $0$ to further enhance robustness. 
	Meanwhile, this also increases the sample complexity, as higher precision is required to estimate $\eta_{\psi}(\rho)$. Therefore, $\eta^\ast$ can be tuned to balance robustness and sample complexity.

	\subsection{Remarks on computational efficiency}
	
	While our protocol is experimentally efficient, its practical application also relies on the computational cost of classical processing. In our protocol, this cost is dominated by obtaining classical descriptions of the projected states $\psi_{\bm{z}}$. Although, in general, computing $\psi_{\bm{z}}$ may require exponential time, for a broad class of physically relevant states, this computation is tractable.
	Tensor-network representations enable efficient classical computation for a wide class of large quantum systems~\cite{Vidal2003MPS,Verstraete2006MPSGroudState,Jordan2008TensorNetwork2D,Vidal2008MERA,Pan2022SimulationBatchTensor}. Examples include states produced by one-dimensional logarithmic-depth circuits~\cite{Shi2006TreeTensor}, the surface-code ground state~\cite{Darmawan2017TensorNetworkSurfaceCode}, and other long-range-entangled states~\cite{Tantivasadakarn2024LongRangeEntanglement}. Neural-network quantum states~\cite{Carleo2017SolvingManybody,Torlai2018NeuralNetworkTomography,Carrasquilla2019GenerativeModel} also capture highly entangled states while permitting efficient evaluation of $\psi_{\bm{z}}$. In many scenarios, such as magic-state injection protocols~\cite{Bravyi2005Universal}, the projected states are known exactly by construction. 
	Furthermore, for small subsystem sizes $n_A$, computing the projected state is feasible under the amplitude query model by querying individual amplitudes of $\psi_{\bm{z}}$, consistent with the framework established in Ref.~\cite{Huang2024Certifying}.
	Hence, our certification scheme is both sample-efficient and computationally tractable for a broad range of states.
	
	Our protocol also requires estimating the maximal fidelity $\fid(\psi_{\bmz})$, or a suitable upper bound, on the subsystem $A$. 
	This value can be computed in time  $2^{\mathcal{O}(n_A)}$ for entanglement and circuit complexity via singular value decomposition, or in time $2^{\cO(n_A^2)}$ for quantum magic using the methods from Ref.~\cite{Hamaguchi2024HandbookQuantifying}.
	Since our protocol only requires a small subsystem size $n_A$, this computational procedure remains efficient.
	
	Crucially, our framework decouples the quantum experiment from the classical computation by using non-adaptive local measurements. 
	During the quantum stage, no prior knowledge of state preparation is required, so data collection is highly efficient. 
	Consequently, the same dataset can be analyzed offline for exponentially many different properties or even different target states~\cite{Huang2020Predicting, Huang2024Certifying}, which can yield exponential improvements in sample complexity over previous protocols, as exemplified in Sec.~\ref{sec:certifying_entanglement}. 
	Given that quantum resources are currently far expensive than classical ones, this decoupling is a significant practical advantage.
	
	\section{Certifying quantum circuit complexity} \label{sec:complexity}
	
	Having established our certification framework, we now illustrate its power and versatility by applying it to four central tasks in quantum information: certifying quantum circuit complexity, entanglement, magic, and state fidelity.
	
	We begin with circuit complexity. It quantifies the minimal resources required to prepare a quantum state and plays a fundamental role connecting quantum computation~\cite{Aaronson2016Complexity, Yi2024Complexity}, quantum learning~\cite{Huang2024Shallow,Zhao2024LearningBounded,Landau2024Shallow}, many-body physics~\cite{Chen2010LongRangeEntanglement, Wen2013topological, Brandao2021ModelsComplexity,Cotler2023EmergentDesign,Du2025SpacetimeComplexity, Mark2024MaximumEntropy}, and black-hole physics~\cite{Brown2018SecondLaw,Susskind2018BlackHole,Oszmaniec2024RecurrenceComplexity}. Certifying that a prepared state has high complexity is therefore crucial both for probing novel structures of quantum matter and for benchmarking experimental platforms.

	\subsection{Certifying unitary circuit complexity}\label{subsec:certifying_complexity}
	
	We first consider the standard notion of circuit complexity for unitary circuit architectures. Formally, for a given circuit architecture specified by a graph $\mathsf{G}=(V,E)$ where vertices represent qubits and edges indicate the qubit pairs on which native two-qubit gates may act, we define the circuit complexity of a quantum state as the minimum circuit depth required to prepare it within $\mathsf{G}$.  For a pure target state $\psi$, the circuit complexity is given by
	\begin{equation}
		C_{\mathsf{G}}(\psi) := \min\left\{d : \ket{\psi} = \prod_{i=1}^d \left(\bigotimes_j V_{i,j}\right) \ket{0}^{\otimes n} \right\}
	\end{equation}
	where $V_{i,j}$ is a two-qubit gate in the $i$-th layer with $\operatorname{supp}(V_{i,j})\in E$ or a single-qubit gate, and the gates in the same layer act on disjoint qubits.
	For a mixed state $\rho$, we adopt the convex-roof extension.  Any $\rho$ can be written as $\rho=\sum_{k} p_{k}\psi_{k}$ with $p_{k}\ge0$ and $\sum_{k}p_{k}=1$. 
	One may prepare $\rho$ by sampling $k$ with probability $p_{k}$ and then preparing $\psi_{k}$ using a circuit of depth $C_{\mathsf{G}}(\psi_{k})$.  The quantum circuit complexity of $\rho$ can therefore be no larger than the largest depth needed among the $\psi_{k}$.  Formally,
	\begin{equation}\label{eq:def_mixed_state_complexity}
		\begin{split}
			C_{\mathsf{G}}(\rho) := \min\Bigl\{ d : \rho = \sum_i p_i \psi_i, \text{ with } p_i \ge 0, \sum_i p_i = 1, \\
			\text{ and } C_{\mathsf{G}}(\psi_i) \le d \text{ for all } i \Bigr\}.
		\end{split}
	\end{equation}

	We apply our framework to certify that the complexity of a state exceeds a certain depth $d$ on a $D$-dimensional lattice.
	While learning-based protocols exist~\cite{Huang2024Shallow, Kim2024LearningPhaseOfMatter, Landau2024Shallow}, they are often restricted to pure states and require a sample complexity that grows with system size, limiting their usage in realistic noisy settings. 
	In contrast, our approach overcomes these limitations and exponentially improves sample complexity, achieving constant sample complexity for certifying constant-depth complexity ($d=\cO(1)$) and remaining sound for mixed states.
	
	Our approach is based on a key physical insight: in a low-depth circuit, information can propagate only a limited distance. This light-cone effect dictates that local measurements on the periphery of a system cannot substantially alter the entanglement structure deep within its center. 
	Consequently, low-complexity states yield projected states with limited entanglement.
	
	Formally, we focus here on an $m\times m$ two-dimensional case ($D=2$, with circuit complexity denoted $C_2$). Our method readily extends to higher dimensions.
	Let $A$ be the set of qubits in a $2w\times2w$ square and $B$ its complement, as illustrated in Fig.~\ref{fig:projected_ensemble_complexity}a, where the width $w$ is chosen slightly larger than the threshold depth $d$.
	The key observation is that, for a circuit of depth $d$, the entanglement between the left and right halves of region $A$ ($L$ and $R$ in Fig.~\ref{fig:projected_ensemble_complexity}a) is bounded. 
	This is because only gates near the boundary between $L$ and $R$ can generate entanglement across this cut, and a shallow circuit contains few such gates. 
	We formalize this as the following lemma. Detailed proofs are given in Appendix~\ref{app:additional_proofs_complexity}.
	
	\begin{figure*}[!htbp]
		\centering
		\includegraphics[width=.7\textwidth]{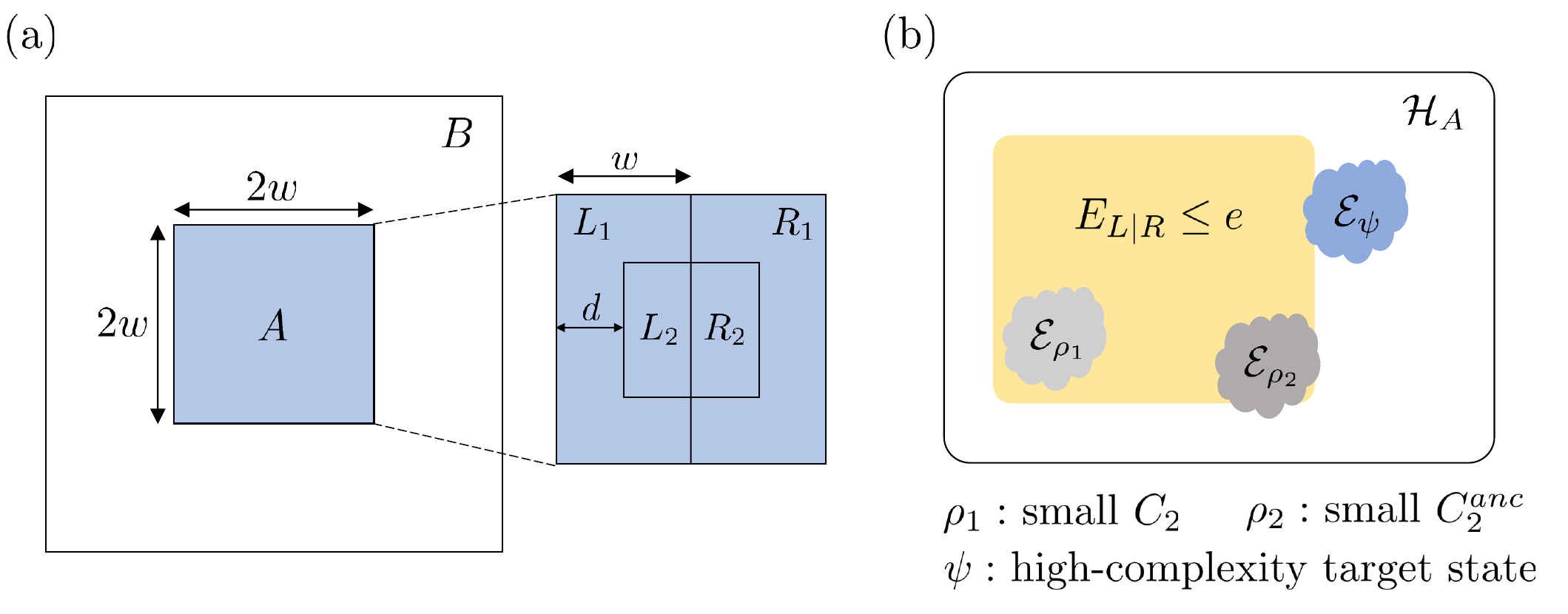}
		\caption{Projected ensembles for certifying circuit complexity $d$.
			(a) Partitioning of a two-dimensional lattice: the full $m\times m$ lattice is divided into subsystems $A$ and $B$, where $A$ is a $2w\times 2w$ square. Subsystem $A$ is further split into $L$ and $R$, with $L=L_1\cup L_2$ and $R=R_1\cup R_2$. The regions $L_1$ and $R_1$ comprise all qubits whose distance to $B$ is at most $d$.
			(b) Entanglement in the projected ensemble certifies complexity. The free projected-state set $\freeprojset$ consists of states with entanglement across $L\mid R$ below a threshold $e$. If a state $\rho_1$ has small unitary circuit complexity $C_2$, all of its projected states lie in $\freeprojset$. If a state $\rho_2$ has small measurement-assisted complexity $C_2^{\mathrm{anc}}$, a large fraction of its projected states lie in $\freeprojset$. Our protocol efficiently certifies the complexity of a high-complexity target $\psi$ whose projected states are far from $\freeprojset$.
		}
		\label{fig:projected_ensemble_complexity}
	\end{figure*}
	
	\begin{lemma}[Entanglement in states of low unitary circuit complexity]\label{lem:entanglement_free_projected_state}
		Let $\ket{\phi}$ be an $(m\times m)$-qubit pure state with circuit depth $C_2(\phi)\le d$. Consider the partition $A\cup B$ with $A=L\cup R$, as depicted in Fig.~\ref{fig:projected_ensemble_complexity}. Then, for every measurement outcome $\bmz \in \Bzrange$,
		\begin{equation}
			E_{L\mid R}(\phi_{\bm z})\le 8wd,
		\end{equation}
		where $E_{L\mid R}(\xi_{LR})\coloneqq S(\xi_L) = -\tr(\xi_L \log \xi_L)$ is the entanglement entropy of the state $\xi_{LR}$ across the bipartition $L\mid R$.
	\end{lemma}
	
	In essence, Lemma~\ref{lem:entanglement_free_projected_state} characterizes the projected free set $\freeprojset$ for the property $\mathsf{P}$ of ``circuit complexity $C_2>d$'': it contains only projected states with low entanglement across $L\mid R$. Conversely, if the projected states of the target are highly entangled, they must lie far from $\freeprojset$. Therefore,  the high entanglement entropy of the projected ensemble provides a clear signature of circuit complexity, as illustrated in Fig.~\ref{fig:projected_ensemble_complexity}b, enabling efficient certification of circuit complexity.
	
	\begin{theorem}[Certifying 2D quantum circuit complexity]\label{thm:2Dcomplexity}
		Let $\psi$ be the target state and consider the bipartition $A\cup B$ shown in Fig.~\ref{fig:projected_ensemble_complexity}, with $w=\eta d$ for some $\eta>4$.  Suppose projected ensemble of $\psi$ on $A$ satisfies
		\begin{equation}
			\Pr_{\bmz \sim p_{\psi}(\bmz)}\left[E_{L \mid R}(\psi_{\bmz}) \ge c \abs{L} + 1\right] \ge p =\Omega(1),   
		\end{equation}
		for a constant $c=4\eta^{-1}+\Omega(1)$.  Then Protocol~\ref{prot:certification} certifies that the circuit complexity of $\psi$ exceeds $d$ with constant failure probability, using $T = \exp\bigl(\cO(\eta^2d^2)\bigr)$ copies of $\psi$ and achieving robustness $\varepsilon = \Omega(1)$.
	\end{theorem}
	
	Therefore,  for constant circuit complexity $d=\cO(1)$, our protocol achieves constant sample complexity and robustness. 
	This constant-depth regime is particularly important: in finite spatial dimensions, depths $d = \omega(1)$ are associated with long-range entanglement in many-body systems~\cite{Hastings2005Topology, Chen2010LongRangeEntanglement,Wen2013topological}. 
	Consequently, our protocol provides a sample-efficient method for certifying long-range entanglement by certifying depths just beyond the constant regime using $\omega(1)$ samples. 
	Moreover, the protocol requires only polynomially many samples up to depths $d = \cO((\log n)^{1/2})$ (or more generally $d=\cO((\log n)^{1/D})$ in spatial dimension $D$).
	
	\subsection{Certifying measurement-assisted circuit complexity}
	We now turn to a stronger notion---measurement-assisted circuit complexity~\cite{Tantivasadakarn2023HierarchyMeasurement, Du2025SpacetimeComplexity}. 
	Measurement-assisted circuits permit mid-circuit single-qubit measurements with classical feedforward, and have attracted significant recent interest because they can provide exponential depth shortcuts for preparing important many-body states, including topologically ordered phases~\cite{Briegel2001PersistentEntanglement, Piroli2021MeasurementAssistedQuantumCircuits, Lu2022MeasurementShortcut, bravyi2022adaptive, Tantivasadakarn2024LongRangeEntanglement}. 
	The recently introduced spacetime quantum circuit complexity also unifies notions of circuit width, depth, and measurements, and reveals a novel complexity concentration phenomenon in measurement-assisted circuits~\cite{Du2025SpacetimeComplexity}.
	
	Here, we focus on the definition of measurement-assisted circuit complexity as the minimum number of measurement layers required to prepare a target state \emph{deterministically}. This definition underpins the hierarchy of topological order and the notion of nontrivial measurement-equivalent quantum phases~\cite{Tantivasadakarn2023HierarchyMeasurement}.  
	Formally, fix an architecture $\mathsf{G}=(V,E)$ on $n=|V|$ qubits, a depth-$d$ measurement-assisted circuit consists of $d$ rounds.  In round $i\in [d]$, conditioned on the previous measurement record $m_{<i}=(m_1,\dots,m_{i-1})$, one applies parallel two-qubit gates $\{V_{i,\ell}^{m_{<i}}\}_{\ell}$ on pairwise disjoint edges of $E$ (and arbitrary single-qubit gates if desired), followed by single-qubit projective measurements $\{P_{i,q}^{m_{<i}}(b)\}_{b\in\{0,1\}}$ on a subset of qubits $q$, yielding outcomes $m_{i,q}\in\{0,1\}$. Gates and measurements within the same round act on disjoint qubits, and future operations may depend on the cumulative measurement record $m_{<i}$. Let $m=(m_1,\dots,m_d)$ denote the full outcome record. The unnormalized post-measurement state is
	\begin{equation}
		\ket{\widetilde{\phi}_m}
		=
		\Bigl[\prod_{i=1}^{d}\Bigl(\prod_{q} P_{i,q}^{\,m_{<i}}(m_{i,q})\Bigr)\Bigl(\prod_{\ell} V_{i,\ell}^{\,m_{<i}}\Bigr)\Bigr]\ket{0}^{\otimes n},
	\end{equation}
	with branch probability $p_m=\|\ket{\widetilde{\phi}_m}\|_2^{2}$ and normalized state $\ket{\phi_m}=\ket{\widetilde{\phi}_m}/\sqrt{p_m}$. The measurement-assisted circuit complexity of a pure state $\psi$ is
	\begin{equation}
		\begin{split}
			C_{\mathsf{G}}^{anc}(\psi)
			\coloneqq \min\Bigl\{d: \exists\ \text{depth-}d\ \text{measurement-assisted circuit} \\
			\text{s.t.}\ \forall m\ \text{with}\ p_m>0,\ \ket{\phi_m}=\ket{\psi}\Bigr\}.
		\end{split}
	\end{equation}
	For mixed states, we use the convex-roof extension same as in Eq.~\eqref{eq:def_mixed_state_complexity}.
	
	Low-depth measurement-assisted circuits are also expected to generate limited entanglement across appropriately chosen cuts. Nonetheless, incorporating the effects of measurements is subtle. For example, depending on the partition, measurements can induce significant entanglement in post-measurement states even when the underlying unitary depth is shallow~\cite{McGinley2025MeasurementInducedEntanglement2D}. We address this by choosing the specific bipartition in Fig.~\ref{fig:projected_ensemble_complexity} and leveraging the fact that entanglement entropy is a monotone that is non-increasing on average under local measurements~\cite{Vidal2000EntanglementMonotones}. The next lemma quantifies the average entanglement remaining on $A$ after measuring its complement $B$.
	\begin{lemma}[Entanglement in states of measurement-assisted circuit complexity]\label{lem:ent_adaptive_circuit}
		Let $\ket{\phi}$ be an $(m\times m)$-qubit pure state with measurement-assisted circuit complexity $C_2^{anc}(\phi)\le d$. Consider the partition $A\cup B$ with $A=L\cup R$ as depicted in Fig.~\ref{fig:projected_ensemble_complexity}. Then,
		\begin{equation}
			\bE_{\bmz \sim p_{\phi}(\bmz)}[E_{L\mid R}(\phi_{\bm z})]\leq 12wd.
		\end{equation} 
	\end{lemma}
	Therefore, although certain projected states $\phi_{\bm z}$ may exhibit high entanglement, the average entanglement remains small. As a result, most projected states lie in a low-entanglement set (Fig.~\ref{fig:projected_ensemble_complexity}b). Nonetheless, this crucial distinction makes the conditional fidelity witness in Lemma~\ref{lem:eta_witness} less effective. To address this, we adopt a linear witness of the form
	\begin{equation}\label{eq:adaptive_complexity_witness}
		\tilde{O} \coloneqq \sum_{\bm z}\bigl[t+(1-t)\Theta_t\bigl(\mathrm{Fid}(\psi_{\bm z})\bigr)\bigr]\, \ketbra{\psi_{\bm z}}_A\otimes \ketbra{\bm z}_B,
	\end{equation}
	where $t \in (0,1)$ and $\Theta_t(x) = \mathbf{1}[x \leq t]$ is a threshold function. 
	The precise construction of $\tilde{O}$, along with the set defining the maximal fidelity $\mathrm{Fid}(\psi_{\bm z})$, is provided in Appendix~\ref{app:additional_proofs_complexity}, together with the full proofs. 
	Using this construction, we establish the following:
	\begin{theorem}[Certifying 2D measurement-assisted circuit complexity]\label{thm:2Dadaptive_complexity}
		Let $\psi$ be the target state and consider the bipartition $A\cup B$ shown in Fig.~\ref{fig:projected_ensemble_complexity}, with $w=\eta d$ for some $\eta>6$.  Suppose the projected ensemble of $\psi$ on $A$ satisfies
		\begin{equation}
			\Pr_{\bmz \sim p_{\psi}(\bmz)}\left[E_{L \mid R}(\psi_{\bmz}) \ge c \abs{L} + 1\right] \ge p =\Omega(1),   
		\end{equation}
		for a constant $c=(6+\Omega(1))p^{-1}\eta^{-1}$.  Then, estimating $\tr(\tilde{O} \rho)$ above a threshold $\eta^*$ certifies that the measurement-assisted circuit complexity of $\psi$ exceeds $d$ with constant failure probability $\delta$ using $\exp\bigl(\cO(\eta^2d^2)\bigr)$
		copies of $\psi$ and achieving robustness $\varepsilon = \Omega(1)$.
	\end{theorem}
	
	In the regime $d=\mathcal{O}(1)$, our protocol again uses only $\mathcal{O}(1)$ sample complexity. This enables, in particular, the certification of nontrivial measurement-equivalent quantum phases (those with $C_{2}^{\mathrm{anc}}=\omega(1)$)~\cite{Tantivasadakarn2023HierarchyMeasurement} using only $\omega(1)$ samples. Although the hidden constant could be large, to our knowledge, these are the first protocols to achieve system-size-independent sample-complexity scaling for certifying unitary and measurement-assisted circuit complexity.
	
	\subsection{Performance on generic quantum states}\label{subsec:states_certifying_complexity}
	
	We now show that our approach applies to generic quantum states and a wide range of physically relevant states. The detailed proofs are given in Appendix~\ref{app:performance_certification_states}. 
	
	An important class of states is those drawn Haar-randomly, which could demonstrate the average-case performance of our protocol over the full Hilbert space. For Haar-random states, the projected states on subsystem $A$ are highly entangled.
	This can be established via the concept of deep thermalization~\cite{Cotler2023EmergentDesign, Choi2023RandomStateBenchmarking}: higher moments of the measurement-induced projected ensemble within a subsystem are close to their thermal values (i.e., to those of the Haar ensemble), so the ensemble forms an (approximate) state design.
	The high entanglement inherent in second-order state designs~\cite{Liu2018EntanglementDesign} suffices to satisfy the requirement in Theorems~\ref{thm:2Dcomplexity} and ~\ref{thm:2Dadaptive_complexity}, enabling sample-efficient certification of complexity for these deep-thermalized states.
	
	\begin{proposition}[Certifying circuit complexity of deep-thermalized states]\label{prop:gap_deep_thermalize}
		Let $\psi$ be an $n$-qubit state whose projected ensemble $\mathcal{E}_{\psi}$ on a specific $\Theta(d^2)$-qubit system $A$ (see Fig.~\ref{fig:projected_ensemble_complexity}) forms an $\epsilon$-approximate 2-design with $\epsilon = \cO(2^{-n_A/2})$. Then, our protocols can certify that $\psi$ has complexity $C_2$ (or $C_2^{anc}$) larger than $d$ with a sample complexity of $\exp(\cO(d^2))$, while achieving constant robustness.
	\end{proposition}
	
	Using deep thermalization results for Haar-random states~\cite{Cotler2023EmergentDesign}, we establish the average-case performance of our protocol for Haar-random states. 
	Moreover, deep thermalization has been observed in many chaotic physical quantum systems~\cite{Ippoliti2023DynamicalPurification, Chang2024ChargeConserving, Mark2024MaximumEntropy}, so our protocol also applies to these physically relevant states.
	
	To further assess practical performance, we consider more experimentally relevant ensembles: states generated by brickwork quantum circuits, where each layer consists of two-qubit gates arranged in a staggered pattern. Such architectures serve as canonical ensembles for studying the generic behavior of physical quantum states~\cite{Brandao2016LocalRandomCircuits,Nahum2017EntanglementGrowth,Haferkamp2022LinearGrowth,McGinley2025MeasurementInducedEntanglement2D}.
	We prove in Appendix~\ref{app:performance_certification_states} that for finite $D$-dimensional circuits, as long as the circuit depth $d'$ is slightly larger than the complexity $d$ we aim to certify—specifically $d' \ge kd$ for some global constant $k$—these depth-$d'$ circuits yield highly entangled projected states, which are therefore far from $\freeprojset$. This provides a sample-efficient method for certifying the complexity of shallow-brickwork-circuit states. We summarize the result below, which generalizes to any finite $D$-dimensional circuit.
	
	\begin{proposition}[Certifying circuit complexity of random brickwork-circuit states]\label{prop:brickwork_state}
		There exists a constant $k$ such that, for a state $\psi$ drawn from random $2D$ brickwork circuits with depth $d' \ge kd$, with probability at least 0.99, our protocols can certify that $\psi$ has complexity $C_2$ (or $C_2^{anc}$) larger than $d$ with a sample complexity of $\exp(\cO(d^2))$, while achieving constant robustness.
	\end{proposition}
	
	Taken together, these two propositions establish the sample efficiency of our protocol for both generic quantum states and experimentally realizable states on near-term quantum platforms.

	\section{Certifying quantum entanglement}\label{sec:certifying_entanglement}
	Entanglement is a fundamental feature of quantum mechanics, and certifying it has long been a central yet challenging task in quantum physics~\cite{Ekert1991CryptoBell, Mayers2004SelfTesting, Horodecki2009QunatumEntanglement, Guhne2009EntanglementDetection, Friis2019EntanglementCertification}.
	While fidelity-based witnesses are a powerful and widely used tool~\cite{Guhne2009EntanglementDetection, Riccardi2021Faithfulness, Cao2023GenerationGME}, they suffer from a major practical limitation: measuring fidelity typically requires global measurements, which are often experimentally prohibitive.
	Here, we demonstrate that our certification framework provides a highly efficient solution for certifying both bipartite and fully inseparable multipartite entanglement. By relying solely on local measurements, our approach substantially reduces the implementation complexity required to witness entanglement.

	We begin by showing that entanglement across exponentially many bipartitions can be certified with constant sample complexity. We then tackle the more demanding task of certifying fully inseparable entanglement. Remarkably, we demonstrate an exponential improvement over prior methods, reducing the required sample complexity from linear to logarithmic in the system size for generic quantum states.
	
	\subsection{Bipartite entanglement}
	We first consider certifying entanglement across a bipartition $A\mid B$. In general, certifying bipartite entanglement is a daunting task~\cite{Gurvits2003NPHardnessEntanglement, Liu2022Fundamental, Liu2025SeparationEntanglement}. Entanglement witnesses provide a powerful route: every entangled state admits a witness, which is very useful in experiments~\cite{Guhne2009EntanglementDetection}. Nevertheless, designing witnesses using simple local measurements is highly nontrivial~\cite{Dimic2018SingleCopyEntanglementDetection, Rodriguez2021CertificationGME}.
	
	Our framework circumvents this obstacle by reducing global certification to a local entanglement witnessing task on a small subsystem. Let $A_1 \subset A$ and $B_1 \subset B$ be two small, disjoint subsystems. The key insight is that, if a state $\rho$ is separable across $A\mid B$, then any state on $A_1B_1$ obtained by performing local measurements on the complement $R:=[n] \backslash (A_1 \cup B_1)$ must also be separable. (Note the difference in notation adopted in this section: The system $R$ corresponds to $B$ in Protocol~\ref{prot:certification} and $A_1B_1$ corresponds to $A$.) Therefore, for the property $\mathsf{P}$ of ``entanglement across $A\mid B$,'' the free projected-state set $\freeprojset$ on $A_1B_1$ is precisely the set of separable states on $A_1 B_1$.
	We define the localizable entanglement between $A_1$ and $B_1$ as
	\begin{equation}\label{eq:localizable_entanglement}
		\mathrm{LE}_{\psi}(A_1,B_1)\;\coloneqq\; \mathrm{LQ}_{\mathsf{Ent}(A,B)}(\psi),
	\end{equation}
	making explicit its dependence on the chosen subsystems. By Theorem~\ref{thm:performance_certification_protocol}, if the target state $\ket{\psi}$ robustly preserves this resource—namely, if $\mathrm{LE}_{\psi}(A_1,B_1)=\Omega(1)$—then our protocol for certifying the bipartite entanglement is efficient. 
	
	\begin{proposition}[Certifying bipartite entanglement]\label{prop:bipartite_entanglement}
		Consider two disjoint subsystems $A_1, B_1 \subseteq [n]$ of constant size, $n_{A_1}, n_{B_1} = \cO(1)$. If the target state $\psi$ satisfies
		\begin{equation}\label{eq:bipartite_conditional_entanglement_requirement}
			\mathrm{LE}_{\psi}(A_1,B_1) = \Omega(1),
		\end{equation}
		then Protocol~\ref{prot:certification} certifies that $\psi$ is entangled across any bipartition $A\mid B$ (with $A_1 \subseteq A, B_1 \subseteq B$) with $\cO(\ln \delta^{-1})$ sample complexity, while achieving constant robustness.
	\end{proposition}
	\revise{This condition of constant localizable entanglement is satisfied by broad classes of states. In particular, Propositions~\ref{prop:gap_deep_thermalize} and \ref{prop:brickwork_state} establish this behavior for deep-thermalized states and random brickwork-circuit states, respectively, with $n_{A_1}, n_{B_1} = \cO(1)$. Furthermore, Lemma~\ref{lem:Haar-random-states-localizable-entanglement} proves that for Haar-random states, two-qubit subsystems ($A_1=\{1\}$ and $B_1=\{2\}$) already satisfy $\mathrm{LE}_{\psi}(1,2)=\Omega(1)$ with overwhelming probability.}  
	
	Notably, if the protocol outputs $\accept$, then we simultaneously certify entanglement of $\rho$ across all bipartitions that separate $A_1$ and $B_1$. Indeed, if $\rho$ were separable across some bipartition $A\mid B$ with $A_1\subseteq A$ and $B_1\subseteq B$, then local measurements on $R$ cannot create entanglement between $A_1$ and $B_1$. Consequently, every projected state $\rho_{\bm z}$ on $A_1\cup B_1$ would be separable, and the certification would output $\reject$.
	
	We remark that prior definitions of localizable entanglement~\cite{Verstraete2004LocalizedEntanglement, Popp2005LocalizableEntanglement, Rodriguez2021CertificationGME} optimize over local measurement bases, whereas here we do not. Avoiding this optimization eliminates substantial classical optimizations and removes stringent device requirements such as adaptive local measurements, yielding a simple, nonadaptive protocol. The price to pay is that our fixed-basis approach may fail to certify certain families of states for which basis optimization is required—for example, the $n$-qubit GHZ state $(\ket{0^{n}}+\ket{1^{n}})/\sqrt{2}$.
	\revise{Nevertheless, the results above demonstrate that this fixed-basis localizable entanglement remains constant across generic and broadly applicable physical state ensembles.}
	
	It is also instructive to discuss the relationship between our witness and standard entanglement measures.
	The quantity $\fid$ corresponds to the maximum singular value ($p_{\max}$) of the Schmidt decomposition of projected states $\psi_{\bmz}$, which relates to the min-entropy of $(\psi_{\bmz})_{A_1}$ ($- \log p_{\max}$) rather than the standard entanglement entropy $S(\psi_{\bmz})$.
	While these two entropies can diverge for large systems, Proposition~\ref{prop:bipartite_entanglement} targets small subsystems where $n_A = \mathcal{O}(1)$.
	In this regime, the difference between the two entropies is bounded by a constant factor, ensuring comparable performance in characterizing the entanglement of the projected states.
	On the other hand, the relation between $\mathrm{LE}_{\psi}(A_1,B_1)$ and the global entanglement entropy $S(\psi)$ may be weak, as global states with disparate entanglement entropy behaviors can exhibit similar localizable entanglement.
	Nonetheless, the entanglement entropy is notoriously difficult to measure experimentally~\cite{Acharya2020EstimatingQuantumEntropy}. In contrast, fidelity-based metrics hold significant operational relevance~\cite{Leone2025ComputationalEntanglement} and are frequently employed in experimental certification~\cite{Guhne2009EntanglementDetection, Cao2023GenerationGME}. 
	By measuring conditional fidelity solely with local measurements, our method significantly reduces the implementation complexity required for certification tasks based on global quantities.

	Beyond its experimental simplicity, we now show that our framework also enables substantial reductions in sample complexity for certifying more complex entanglement structures.
	
	\subsection{Fully inseparable entanglement}\label{subsec:fully-inseparable}
	
	We now address the more stringent task of certifying fully inseparable entanglement, where a state is entangled across every possible bipartition~\cite{Eisert2006MultiparticleEntanglement}. 
	This structure is a key resource for quantum networks: it is the necessary and sufficient form of entanglement for activating genuine multipartite entanglement~\cite{Palazuelos2022FullyInseparableGMEActivation, Zwerger2019DIGME}. It also provides a meaningful benchmark for noisy devices aiming to implement quantum error correction~\cite{Rodriguez2021CertificationGME}. Consequently, certifying fully inseparable entanglement offers a stringent test of entangling power and a stepping stone toward quantum networks and fault-tolerant quantum computing.
	
	At first glance, certifying fully inseparable entanglement appears challenging, as one must rule out separability across exponentially many bipartitions.
	Nonetheless, Proposition~\ref{prop:bipartite_entanglement} provides a path forward. 
	It suffices to certify localizable entanglement between qubits $i$ and $i+1$ for each $i\in[n-1]$. 
	Indeed, if $\rho$ were separable across some bipartition $A\mid B$, then there exists $i\in A$ satisfying $i+1\in B$, therefore the test on the pair $\{i,i+1\}$ would fail.
	A straightforward implementation that runs our protocol separately for every $i$ would require $n-1$ distinct experiments, yielding total sample complexity at least linear in $n$, similar to existing methods~\cite{Rodriguez2021CertificationGME}.
	
	Here, we exponentially improve this scaling by decoupling the measurement scheme from the choice of pairs and reusing a single dataset across all tests, with the spirit underlying classical shadows~\cite{Huang2020Predicting}.
	Concretely, rather than measuring the complement of a specific pair in the computational basis, we perform single-qubit random Pauli measurements on all qubits. Data collection is thus independent of which pair will be tested in postprocessing. To leverage this dataset, we use the random-basis analogue of \eqref{eq:localizable_entanglement}:
	\begin{equation}
		\widetilde{\mathrm{LE}}_\psi(A_1,B_1)\;\coloneqq\;\widetilde{\mathrm{LQ}}_{\mathsf{Ent}(A_1\mid B_1)}(\psi).
	\end{equation}
	Protocol~\ref{prot:certification} extends directly to this random-basis setting: Measure the complement $R$ of $A_1B_1$ in a random Pauli basis instead of the $Z$ basis, then estimate $\tr(\psi_{\bm x}\rho_{\bm x})$ and $\fid(\psi_{\bm x})$ on $A_1\cup B_1$. 
	Replacing the condition in Eq.~\eqref{eq:bipartite_conditional_entanglement_requirement} of Proposition~\ref{prop:bipartite_entanglement} by $\widetilde{\mathrm{LE}}_\psi(A_1,B_1)=\Omega(1)$ yields the same performance guarantees by using the random-basis variant of Protocol~\ref{prot:certification}. 
	
	Now, we show our modified protocol reduces the sample complexity from linear to logarithmic in $n$ (proofs in Appendix~\ref{app:proof_entanglement_certification}). Concretely, we collect a single dataset from $T$ rounds of random Pauli measurements on $\rho$ and reuse the same dataset to certify localizable entanglement for the $(n-1)$ pairs $A_1 = \{i\}, B_1=\{i+1\}$ for $1\le i\le n-1$. If all tests output $\accept$, we certify fully inseparable entanglement. 
	This protocol also extends to any set of $n-1$ pairs that connects the $n$-vertex graph.
	
	\begin{proposition}[Certifying fully inseparable entanglement with logarithmic sample complexity]\label{prop:fully_inseparable_entanglement}
		Let $\psi$ be a quantum state satisfying 
		\begin{equation}\label{eq:fully_inseparable_requirement}
			\widetilde{\mathrm{LE}}_{\psi}(i,i+1) = \Omega(1)
		\end{equation} for all $1 \le i \le n-1$.
		Then the fully inseparable entanglement of $\psi$ can be certified with failure probability $\delta$ using $T=\cO(\log n + \log \delta^{-1})$ samples, while achieving constant robustness.
	\end{proposition}
	
	For generic states, such as Haar-random states, the condition $\widetilde{\mathrm{LE}}_{\psi}(i, i+1) = \Omega(1)$ holds simultaneously for all $i$ with high probability (Appendix~\ref{app:proof_entanglement_certification}, \revise{Lemma~\ref{lem:Haar-random-states-localizable-entanglement}}).
	To further demonstrate localizable entanglement in a physical setting, we evaluate its distribution across all nearest-neighbor pairs for random stabilizer states $\psi$ up to hundreds of qubits, as presented in Fig.~\ref{fig:StabilizerEntanglement} (details provided in Appendix~\ref{app:proof_entanglement_certification}).
	Here, $i_p$ denotes the site index where $\widetilde{\mathrm{LE}}_{\psi}(i_p,i_p+1)$ corresponds to the $p$-th quantile of the distribution ($0<p<1$). We observe that the deviation between the maximum $\widetilde{\mathrm{LE}}_{\psi}(i_0,i_0+1)$ and the minimum $\widetilde{\mathrm{LE}}_{\psi}(i_1,i_1+1)$ is small, with the minimum concentrating consistently around $0.2$. Furthermore, the variation across distinct random stabilizer states is negligible.
	These results further validate that our protocol provides a sample-efficient method for certifying this strong form of multipartite entanglement in generic quantum states.

	\begin{figure}[!htbp]
		\centering
		\includegraphics[width=.4\textwidth]{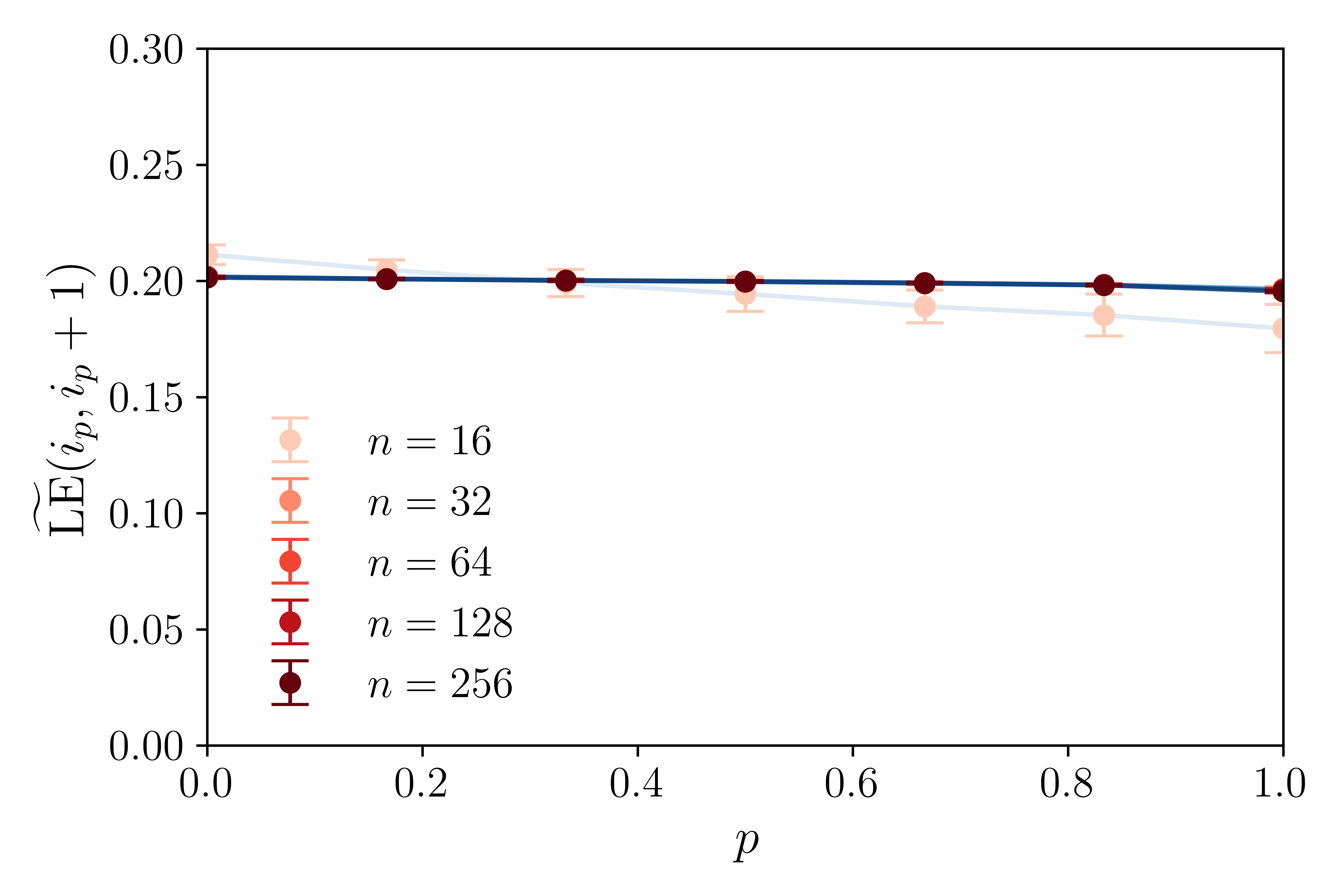} 
		\caption{Localizable entanglement $\widetilde{\mathrm{LE}}_{\psi}(i_p,i_p+1)$ for random stabilizer states $\psi$. Here, $i_p$ denotes the site index corresponding to the $p$-th quantile of the localizable entanglement across all sites. Error bars denote the standard deviation across the sampled states, which are negligible for large system sizes.}
		\label{fig:StabilizerEntanglement}
	\end{figure}

	\subsection{Case study: ground states of Hamiltonian systems}\label{subsec:HamiltonianEntanglement}
	Building on the general result that the random-basis localizable entanglement $\widetilde{\mathrm{LE}}_{\psi}$ can remain constant for generic many-body states, we now investigate this metric for ground states of concrete Hamiltonian models, thereby demonstrating the practical efficiency of our protocol.
	
	We begin with the spin-$\tfrac{1}{2}$ XXZ chain, a paradigmatic and exactly solvable model of quantum magnetism with a rich phase diagram~\cite{Takahashi1999OneDimensionalSolvable}. The XXZ chain has played an important role in entanglement theory~\cite{Calabrese2004EntanglementEntropyQFT, Gu2005GroundStateXXZ, Popp2005LocalizableEntanglement} and serves as a key experimental benchmark on multiple platforms, including trapped ions~\cite{Brydges2019Probing} and optical cavities~\cite{Luo2025XYZOpticalCavity}. The Hamiltonian is
	\begin{equation}
		H_{\mathrm{XXZ}}
		= -\sum_{i=1}^{n}\Bigl[\sigma_{x}^{i}\sigma_{x}^{i+1}
		+ \sigma_{y}^{i}\sigma_{y}^{i+1}
		+ \Delta\,\sigma_{z}^{i}\sigma_{z}^{i+1}\Bigr],
	\end{equation}
	where $\Delta$ is the anisotropy, with periodic boundary condition $\sigma_{\alpha}^{n+1}\equiv\sigma_{\alpha}^{1}$. 
	For a translational-invariant non-degenerate Hamiltonian, its unique ground state is also translational-invariant, so $\widetilde{\mathrm{LE}}_{\psi}(i,i+1)$ is identical for all $i$. 
	We numerically evaluate $\widetilde{\mathrm{LE}}_{\psi}(1,2)$ of the ground state $\psi$ for system sizes up to $n=64$ using the density matrix renormalization group (DMRG) algorithm~\cite{White1992DMRG}. 
	Details of the numeric experiments are provided in Appendix~\ref{app:proof_entanglement_certification}.
	
	The results are shown in Fig.~\ref{fig:HamiltonianEntanglement}a. 
	We observe that $\widetilde{\mathrm{LE}}_{\psi}(1,2)$ remains nonvanishing throughout the critical regime $|\Delta|<1$, with only mild finite-size drift as $n$ increases. 
	Notably, $\widetilde{\mathrm{LE}}_{\psi}(1,2)$ reaches maximal at the isotropic point $\Delta=-1$, at which the XXZ chain undergoes a Kosterlitz-Thouless transition from gapped N\'eel antiferromagnet ($\Delta < -1$) to critical Luttinger liquid ($\abs{\Delta}<1$)~\cite{Takahashi1999OneDimensionalSolvable}.
	Moreover, it vanishes at another phase transition point $\Delta=1$ where the ground state becomes a product state for $\Delta > 1$. 
	Therefore, the behavior of $\widetilde{\mathrm{LE}}_{\psi}(1,2)$ is able to indicate various phase transitions.
	
	\begin{figure}[!htbp]
		\centering
		\includegraphics[width=.5\textwidth]{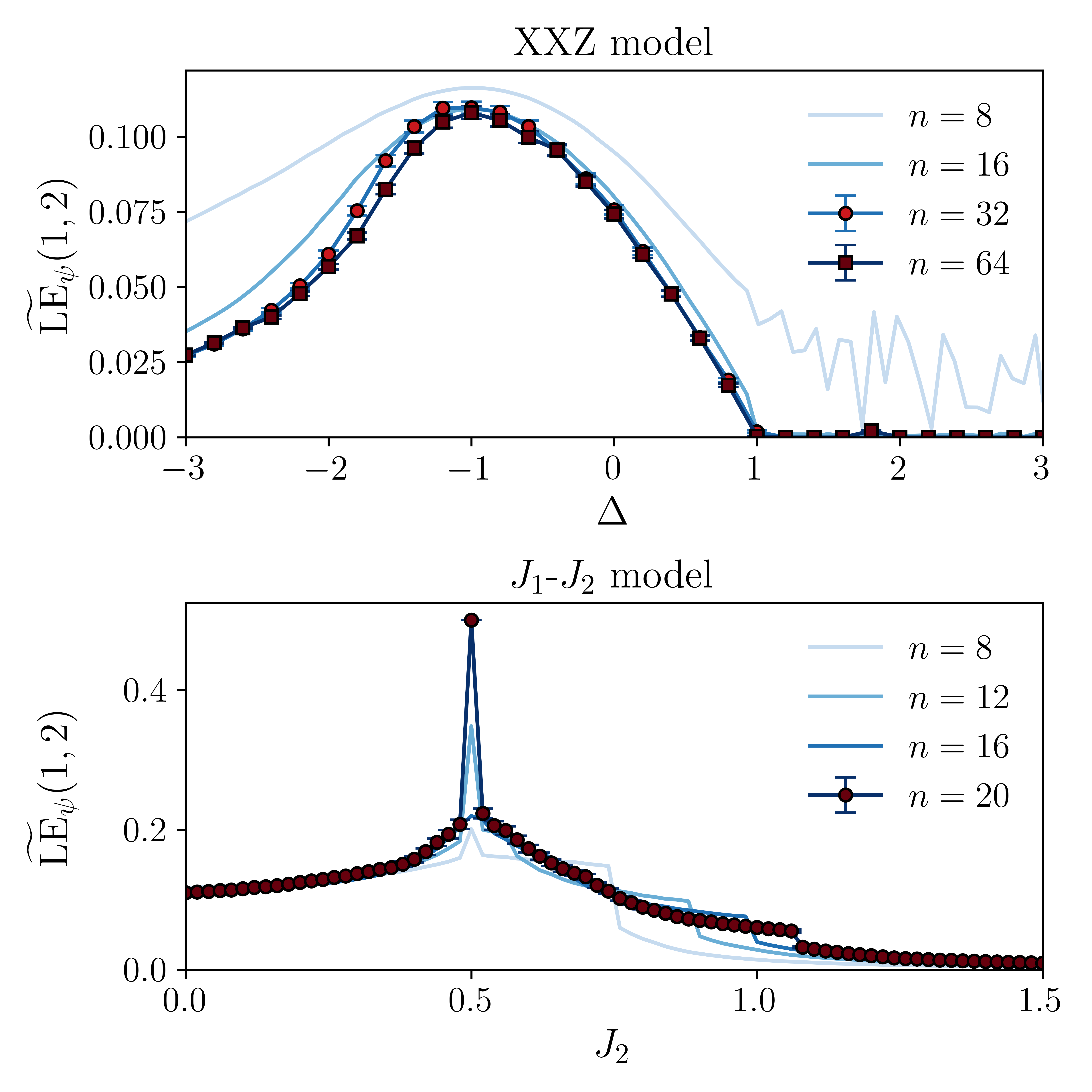} 
		\caption{Localizable entanglement $\widetilde{\mathrm{LE}}_{\psi}(1,2)$ in the ground states of Hamiltonian systems. Across broad parameter regimes, the localizable entanglement remains finite, enabling sample-efficient certification of entanglement. In addition, its behavior signifies key parameter points of the system.}
		\label{fig:HamiltonianEntanglement}
	\end{figure}

	We next consider the spin-$\tfrac{1}{2}$ $J_1$–$J_2$ Heisenberg chain, a paradigmatic frustrated extension of the nearest-neighbor model with next-nearest-neighbor couplings. Although not exactly solvable in general, its key parameter points are well established: a fluid-dimer transition occurs at $J_2/J_1\simeq 0.2411$~\cite{Okamato1992PhaseTransitionNextNearest}; the exactly solvable Majumdar-Ghosh point $J_2/J_1=\tfrac{1}{2}$~\cite{Majumdar1969MGPoint}; and, within the same phase, a Lifshitz point near $J_2/J_1\approx 0.52$~\cite{Bursill1995LifschitzPoint}. The model captures the physics of zigzag-chain materials~\cite{Pregelj2015PhaseFrustratedChain} and has been implemented on programmable superconducting platforms~\cite{Chowdhury2024HeisenbergChainSuperconducting}.
	
	Concretely, with periodic boundary conditions, the Hamiltonian is
	\begin{equation}
		H_{J_1J_2} = J_1\sum_{i=1}^{n} \bm{\sigma}^i \cdot \bm{\sigma}^{i+1} + J_2 \sum_{i=1}^{n} \bm{\sigma}^i \cdot \bm{\sigma}^{i+2}
	\end{equation}
	where $\bm{\sigma}^{\,i}=(\sigma_x^i,\sigma_y^i,\sigma_z^i)$ and $\sigma_{\alpha}^{n+k}\equiv\sigma_{\alpha}^{k}$.
	Here, we set $J_1=1$ and consider an antiferromagnetic interaction $J_2>0$.
	For the ground state $\psi$, we evaluate the random-basis localizable entanglement $\widetilde{\mathrm{LE}}_{\psi}(1,2)$. The results are shown in Fig.~\ref{fig:HamiltonianEntanglement}b.
	
	At the Majumdar–Ghosh point $J_2/J_1=\tfrac{1}{2}$, fluctuations in $\widetilde{\mathrm{LE}}_{\psi}(1,2)$ are from the twofold degeneracy of the ground space: in one ground state, sites $1$ and $2$ are maximally entangled, whereas in the other they are unentangled. Therefore, $\widetilde{\mathrm{LE}}_{\psi}(1,2)$ depends on the superposition coefficients of the two ground states. 
	As $J_2$ increases from $0$, $\widetilde{\mathrm{LE}}_{\psi}(1,2)$ grows and peaks near the Lifshitz point $J_2\approx 0.52$, then decreases for larger $J_2$. 
	Therefore,  the Lifshitz point is signaled by the discontinuity in the first derivative of $\widetilde{\mathrm{LE}}_{\psi}(1,2)$. 
	By contrast, we do not observe a clear signature near the phase transition $J_2/J_1\simeq 0.2411$, and understanding which types of parameters can be captured by localizable entanglement is an interesting question~\cite{Popp2005LocalizableEntanglement}.
	
	One can find wide parameter regimes where $\widetilde{\mathrm{LE}}_{\psi}(1,2)$ is robustly bounded away from zero across both models. 
	For non-degenerate Hamiltonians, translation invariance then implies $\widetilde{\mathrm{LE}}_{\psi}(i,i+1) = \widetilde{\mathrm{LE}}_{\psi}(1,2) =\Omega(1)$ for all nearest-neighbor pairs.  
	By Proposition \ref{prop:fully_inseparable_entanglement}, this yields logarithmic sample complexity and constant robustness for certifying fully inseparable entanglement in those Hamiltonian systems.

	\section{Certifying quantum magic}\label{sec:certifying_magic}
	Here, we further demonstrate the power of our framework by applying it to certify quantum magic (or non-stabilizerness), a critical resource for fault-tolerant quantum computing~\cite{Bravyi2016TradingClassicalQuantum,Bravyi2019StabilizerDecomposition,Seddon2021TigerMagic,Liu2022ManyBodyMagic}. 
	Its characterization has recently received significant attention~\cite{Leone2022StabilizerRenyi,Haug2023ScalableMeasureMagic,Haug2024MeasureStabilizerEntropy,Gross2021SchurWeylClifford,Hinsche2025SingleCopyStablizerTesting,Warmuz2025MagicMonotoneMixedStates,Wei2025LongRangeMagic}. 
	Nonetheless, robust certification of this critical resource remains challenging. 
	Existing protocols often lack soundness for mixed-state inputs or require demanding measurements~\cite{Haug2023ScalableMeasureMagic,Gross2021SchurWeylClifford,Hinsche2025SingleCopyStablizerTesting}. 
	We overcome these limitations and present the first protocol that certifies quantum magic using very few local measurements and achieves soundness for mixed states. Furthermore, it only uses constant sample complexity and achieves constant robustness.
	
	A pure $n$-qubit stabilizer state can be written as  
	$\ket{\psi} = C\ket{0}^{\otimes n}$,  
	where $C$ is a Clifford circuit acting on the initial all-zero state~\cite{Gottesman1997Stabilizer}.  
	The set of free states—namely, the stabilizer set—is defined as the convex hull of stabilizer states:
	\begin{equation}
		\begin{split}
			\mathrm{STAB}_{n} \coloneq \{\rho : \rho  = \sum_i p_i \ketbra{\psi_i}, \text{ with } p_i \ge 0, \sum_i p_i = 1, \\ 
			\text{ and } \ket{\psi_i} = C_i\ket{0}^{\otimes n}, C_i \text{ Clifford}\}.
		\end{split}
	\end{equation}
	A key feature of the stabilizer set is its closure under Pauli measurements~\cite{Gottesman1997Stabilizer}.
	Consequently, the free-projected set $\freeprojset$ is simply the set of stabilizer states on $A$, denoted $\mathrm{STAB}_A$.  
	The maximal fidelity $\fid(\phi)$ then becomes the well-known stabilizer fidelity~\cite{Bravyi2019StabilizerDecomposition,Liu2022ManyBodyMagic}:
	\begin{equation}
		\fid(\psi_{\bm z})
		= \max_{\ket{\phi}\in\mathrm{STAB}_A}
		\bra{\phi}\psi_{\bm z}\ket{\phi}.
	\end{equation}
	
	Let $\mathrm{LM}(\psi)\coloneqq \mathrm{LQ}_{\mathsf{Magic}}(\psi)$ denote the localizable magic of the target. By Theorem~\ref{thm:performance_certification_protocol}, if $n_A = \cO(1)$  and $\mathrm{LM}(\psi)=\Omega(1)$, then our protocol has constant sample complexity and robustness. 
	\revise{We prove that this condition holds for typical Haar-random states with $n_A = 1$ in Appendix~\ref{app:proof_magic_certification}.}
	This condition is also expected for generic physical states: their projected ensemble on $A$ is typically random states~\cite{Cotler2023EmergentDesign}, which possess abundant quantum magic and require many layers of non-Clifford gates to prepare~\cite{Turkeshi2025MagicSpreading, Leone2025NonCliffordCost}.
	Hence, our certification protocol is both efficient and robust for generic states.
	
	However, in many practical scenarios, quantum magic is a scarce and costly resource.
	A crucial scenario involves the injection of magic into Clifford circuits~\cite{Bravyi2005Universal, Niroula2024MagicPhaseTransition}. 
	In near-term demonstration of fault-tolerant architectures, the quantum circuits are likely to be dominated by Clifford operations, with a limited number of logical $T$ gates implemented via magic state distillation or cultivation~\cite{Gidney2024Cultivation}. 
	To demonstrate the efficiency of our certification protocol for this practical setting, we consider and examine a magic-injection-and-scrambling model (Fig.~\ref{fig:magic}a).
	We initialize the system in a product state comprising $n-k$ single-qubit stabilizer states $\ket{0}$ and $k$ magic states $\ket{\alpha} = 1/\sqrt{2}( \ket{0} + e^{i\alpha} \ket{1})$, where $\alpha \in [0, \pi/4]$ and $k$ controls the amount of injected magic (with $\alpha = \pi/4$ yielding maximal magic).
	A random $n$-qubit Clifford unitary $C$ is then applied to scramble the injected magic across the system.
	The resulting target state is given by
	\begin{equation}
		\ket{\psi} = C \left(\ket{0}^{\otimes n-k} \otimes \ket{\alpha}^{\otimes k}\right).
	\end{equation}
	Following local Pauli measurements on the $n_B$-qubit subsystem $B$, the resulting projected states $\psi_{\bm z}$ are defined on the complementary subsystem $A$, which consists of $n_A = n - n_B$ qubits.

	\begin{figure}[!htbp]
		\centering
		\includegraphics[width=.5\textwidth]{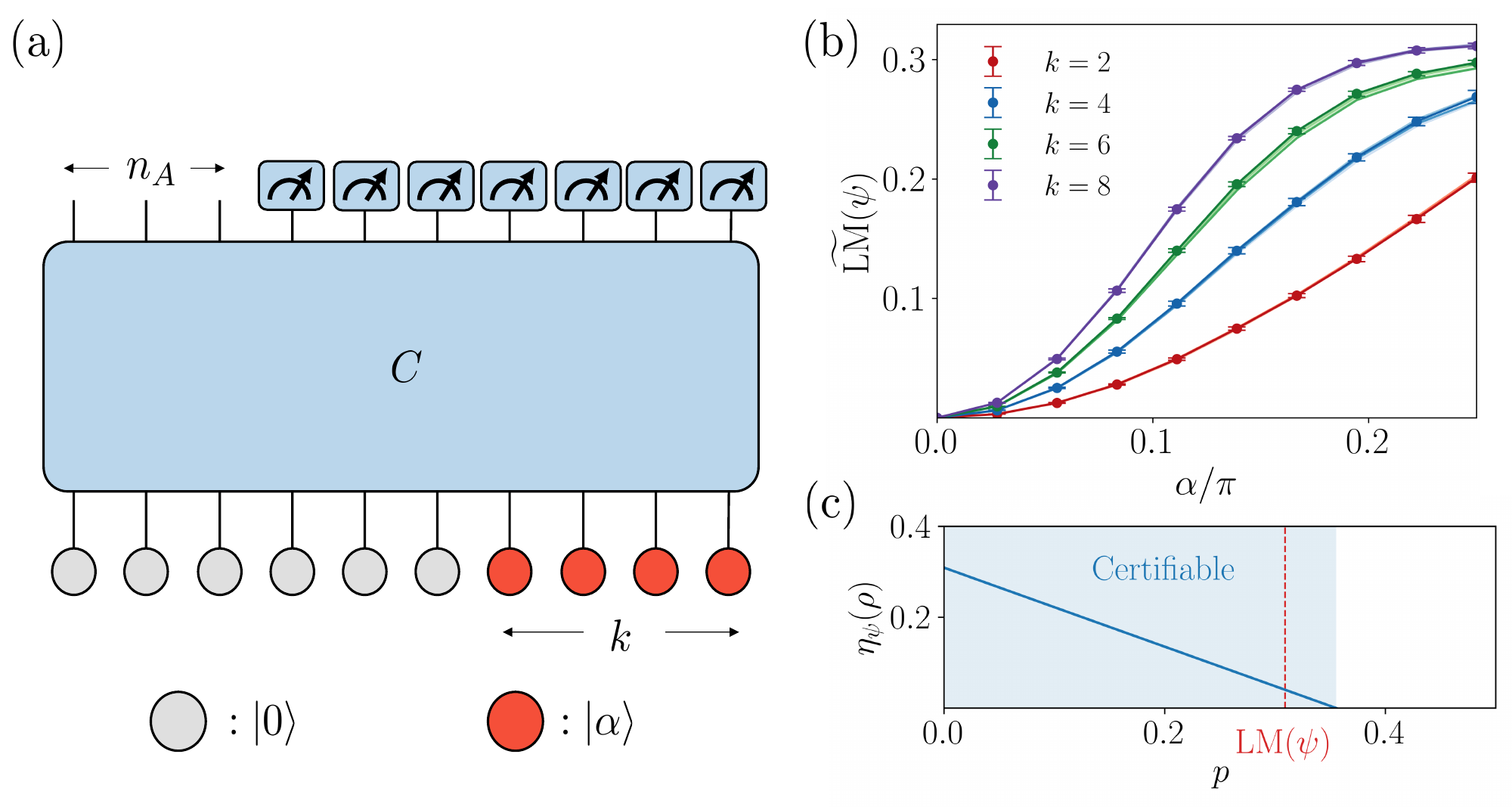} 
		\caption{Magic certification and its performance. (a) The system is initialized in the state $\ket{0}^{\otimes n - k} \otimes \ket{\alpha}^{\otimes k}$. A global random Clifford unitary $C$ is then applied to scramble the injected magic. Finally, single-qubit Pauli measurements are performed on the complementary system, yielding projected states on the small $n_A$-qubit subsystem. (b) Localizable magic $\widetilde{\mathrm{LM}}(\psi)$ with system size $n \in \{32, 64, 128, 256, 512\}$ for a fixed subsystem size $n_A=3$ and varying numbers of injected magic state $k \in \{2,4,6,8\}$. Darker colors correspond to larger system sizes. Each data point represents the average over 5 random Clifford realizations (each estimated using 1000 projected state samples), with error bars denoting the standard deviation across realizations.  (c) The relation between global depolarizing noise strength $p$ and $\eta_{\psi}(\rho)$ for $\rho = (1-p)\psi + p I/d$ with the same state $\psi$. 
			The value of $\eta_{\psi}(\rho)$ remains in the certifiable regime for $p \gtrsim \mathrm{LM}(\psi)$, demonstrating the strong robustness of our protocol.}
		\label{fig:magic}
	\end{figure} 
	
	We first numerically evaluate the localizable magic $\widetilde{\mathrm{LM}}(\psi)$ for a fixed subsystem size $n_A = 3$, varying the magic injection angle $\alpha$, magic state count $k$, and the total system size $n$ (Fig.~\ref{fig:magic}b). These simulations utilize established classical algorithms for $T$-doped Clifford circuits and stabilizer fidelity estimation~\cite{Bravyi2019StabilizerDecomposition, Hamaguchi2024HandbookQuantifying}.
	The data reveals several key features. First, $\widetilde{\mathrm{LM}}(\psi)$ increases with $k$,  demonstrating that global magic is effectively concentrated into the subsystem $A$, an insight that underlies protocols like magic-state injection. 
	Second, and most importantly for our purposes, $\widetilde{\mathrm{LM}}(\psi)$ quickly exceeds a constant threshold as $\alpha$ increases. 
	This implies that for a fixed random Pauli basis, the localizable magic $\widetilde{\mathrm{LM}}(\psi)$ exceeds a constant with high probability.
	According to Theorem~\ref{thm:performance_certification_protocol}, this directly shows that our protocol certifies magic in these states with constant sample complexity and robustness.
	It is also interesting that $\widetilde{\mathrm{LM}}$ increases monotonically with $\alpha$, suggesting that the projected-ensemble-averaged magic faithfully reflects, to a significant extent, the total magic of the full state.  
	Such a magic concentration phenomenon has likewise been observed in~\cite{Niroula2024MagicPhaseTransition, Zhang2024MagicConcentration}. 
	We remark that $\widetilde{\mathrm{LM}}(\psi)$ is not a magic monotone, as it can change under the addition of ancillary $\ket{0}$ qubits and Clifford rotations.  
	Nevertheless, it serves as a practical indicator of the global magic contained in the state.
	
	To further demonstrate robustness, we simulate our protocol on a state subject to global depolarizing noise, $\rho = (1-p)\psi + p\, I/d$, a good model for present-day experiments~\cite{Dalzell2024RandomCircuitsGlobalNoise}. Moreover, coherent errors can be converted into effective depolarizing noise via randomized compiling~\cite{Wallman2016NoiseTailoringRandomizedCompliling,Hashim2021RandomizedCompilingNoisyDevice}. 
	Fig.~\ref{fig:magic}c plots $\eta_{\psi}(\rho)$ versus the noise strength $p$ for a $512$-qubit state generated with $\alpha = \pi/4$. As shown, the conditional fidelity $\eta(\rho)$ remains well above zero even for substantial noise levels $p \approx 0.35$. This behavior can also be derived analytically:
	\begin{equation}
		\begin{split}
			\eta_{\psi}(\rho) &= (1-p)\mathrm{LM}(\psi)+p d^{-1} \Bigl[\sum_{\bmz} \bigl(1 - d_A\fid(\psi_{\bmz}) \bigr) \Bigr] \\
			&= (1-p)\mathrm{LM}(\psi) + p\Bigl[d_A^{-1} - d_B^{-1} \sum_{\bmz} \fid(\psi_{\bmz})\Bigr].
		\end{split}
	\end{equation}
	By adjusting the decision boundary to require $\eta^* < \eta_{\psi}(\rho)$, our protocol could certify the magic in $\rho$ whenever $\eta_{\psi}(\rho)>0$, which holds provided
	\begin{equation}\label{eq:noise_robustness}
		p < \frac{ \mathrm{LM}(\psi) }{d_B^{-1} \sum_{\bmz} \fid(\psi_{\bmz}) + \mathrm{LM}(\psi) -d_A^{-1}} \approx \frac{ \mathrm{LM}(\psi) }{1 - d_A^{-1}},
	\end{equation}
	consistent with the numerical simulation. Here the approximation uses the empirical observation $\mathrm{LM}(\psi) \approx 1-d_B^{-1}\sum_{\bm z}\fid(\psi_{\bm z}) $. 
	This level of robustness is much better than a direct application of Theorem~\ref{thm:performance_certification_protocol}, showing that robustness can be improved in realistic noise scenarios.
	Moreover, the monotonic decrease of $\eta_{\psi}(\rho)$ with increasing noise strength directly highlights the utility of our protocol for benchmarking noisy quantum devices.
	
	\section{Certifying quantum state fidelity using random-basis-enhanced conditional fidelity}\label{sec:fidelity}
	
	We \revise{now address} a particularly strong certification task: certifying the global state fidelity $ \langle \psi | \rho | \psi \rangle $.
	Recent protocols achieve this using only local measurements~\cite{Huang2020Predicting, Gupta2025SingleQubitCertification, Li2025UniversalVerification, Coladangelo2026TwoBases}, but they face severe trade-offs between experimental simplicity and certification power. 
	Adaptive protocols can, in principle, certify arbitrary states~\cite{Gupta2025SingleQubitCertification}. 
	\revise{For Haar-random states, recent adaptive local-measurement protocols have achieved a sample complexity of $\cO((\log n)^2)$ and a robustness of $\Omega(1/\log n)$~\cite{Coladangelo2026TwoBases}. 
		However, employing adaptive protocol requires challenging classical feedback and long-lived quantum memory to store the state between measurement rounds~\cite{Gupta2025SingleQubitCertification, Li2025UniversalVerification, Coladangelo2026TwoBases}. 
		Ref.~\cite{Coladangelo2026TwoBases} also demonstrated that constant scaling is achievable, but doing so required non-local measurements across $\cO(\log n)$ qubits. }

	By contrast, protocols that employ non-adaptive local measurement schemes are far more compatible with near-term hardware. 
	However, state-of-the-art non-adaptive methods for generic states, such as the shadow-overlap protocol, \revise{exhibit \revise{$\cO(n^4)$} sample complexity and $\Omega(n^{-2})$ noise robustness, and this sample complexity can only improve to $\cO(n^2)$ by using adaptive measurements~\cite{Huang2024Certifying}. 
		Consequently, whether constant robustness and sample complexity can be achieved for generic states using solely non-adaptive local measurements has remained a central open challenge in the field.}
	
	\revise{
		Here, we show that our enhanced conditional fidelity protocol settles this problem. 
		Relying solely on non-adaptive single-qubit Pauli measurements, it achieves $\cO(1)$ sample complexity and $\Omega(1)$ robustness for certifying generic states, including both Haar-random states and physically relevant random graph states.
		By achieving essentially optimal scaling guarantees, our protocol provides a definitive resolution to this open challenge in quantum certification.
	}
	
	\subsection{Random-basis-enhanced conditional fidelity}
	Our previous conditional fidelity mainly focuses on measuring subsystem $B$ in a single, fixed basis (e.g., the computational basis), which is insufficient for fidelity certification: the product state $|\psi_{0^{n_B}}\rangle_A \otimes |0^{n_B}\rangle_B$ achieves unit conditional fidelity even though its global fidelity with $|\psi\rangle$ could be exponentially small in $n_B$.
	The core issue is that a fixed measurement basis on $B$ fails to capture the rich correlations across the $A \mid B$ partition, which are revealed only in other bases (e.g., the relative phases among the projected states).
	
	Our solution is to measure subsystem $B$ in randomly chosen single-qubit Pauli bases. The intuition is that random basis choices produce projected states on subsystem $A$ that are closely correlated with one another, so the conditional fidelity captures the global fidelity structure and becomes a sensitive witness of the overall state fidelity. To formalize this idea, we define the following observable:
	\begin{equation}
		O_{\psi} = 3^{-n_B} \sum_{\bmx \in \Bxrange, \tilde{p}(\bmx) \neq 0} \ketbra{\psi_{\bmx}}_A \otimes \ketbra{\bmx}_B.
	\end{equation}
	By construction, $\tr(O_{\psi}\psi) = 1$.
	The ability to certify fidelity is determined by the spectral gap
	\begin{equation}
		\Delta(O_{\psi})
		\coloneqq
		1 - \max_{\phi \perp \psi} \tr(O_{\psi}\phi),
	\end{equation}
	as shown in the following:
	\begin{proposition}[Certifying state fidelity with local Pauli measurements]\label{prop:fidelity_performance}
		Let $\psi$ be the target state with $\Delta(O_{\psi}) > 0$, and fix parameters $0 < c < 1/2$, $0 < F < 1$ and $\delta>0$. Then, using local Pauli measurements on 
		\begin{equation}
			T = \frac{27 \ln(\delta^{-1})\, (4^{n_A}+1)}{c^{2}\Delta^{2}(1-F)^{2}}
		\end{equation}
		samples of $\rho$ suffices to:
		\begin{enumerate}
			\item Output $\reject$ with probability at least $1-\delta$ when $\tr(\rho\psi) \le F$;
			\item Output $\accept$ with probability at least $1-\delta$ when $\tr(\rho\psi) \ge 1-(1-2c)\Delta(1-F)$.
		\end{enumerate}
	\end{proposition}
	Therefore, a larger spectral gap $\Delta$ directly improves both sample complexity and robustness. 
	While the gap can vanish for certain target states (e.g., product states across $A\mid B$), we now show that it is typically bounded below by a constant for generic families of states of interest.
	
	\subsection{Provable constant gap for random graph states}
	We first prove that our protocol achieves a constant gap for random graph states, providing a natural and robust benchmark for evaluating its performance.
	Graph states form a canonical and widely used subclass of stabilizer states, serving as standard resources in quantum computation and quantum networks~\cite{Raussendorf2001oneway,Briegel2001PersistentEntanglement,Hein2004EntanglementGraphStates}.
	Moreover, every stabilizer state is local-Clifford equivalent to a graph state~\cite{Gottesman1997Stabilizer,Hein2004EntanglementGraphStates}, making graph states a representative family for studying generic phenomena for stabilizer states.
	
	Specifically, let $G=(V,E)$ be a simple undirected graph on $n$ vertices $V=\{1,\dots,n\}$ with symmetric adjacency matrix $A\in\bF_2^{n\times n}$ with zero diagonal entries.
	The associated $n$-qubit graph state $\ket{G}$ is given by $\ket{G} = \left(\prod_{i<j}\mathrm{CZ}_{i,j}^{A_{i,j}} \right)\ket{+}^{\otimes n}$. 
	We consider the random graph state ensemble obtained by sampling $A$ uniformly from the set of all symmetric, zero-diagonal matrices.
	Our main theorem for this ensemble is the following.
	\begin{theorem}[Constant gap for random graph states]\label{thm:constant-gap-graph-states}
		Let $\ket{\psi}$ be drawn uniformly from $n$-qubit graph states, and let $n_A=1$. Then, with probability $1-\exp(-\Omega(n))$, $\Delta(O_{\psi})=\Omega(1)$.
	\end{theorem}
	\begin{proof}[Proof sketch]
		For a target stabilizer state $\psi$, the observable $O_{\psi}$ has eigenstates in the form $E\ket{\psi}$, where $E$ is a Pauli error. 
		Consequently, bounding the second-largest eigenvalue of $O_{\psi}$ reduces to showing that every Pauli error on $\psi$ can be localized on system $A$ via the random measurements on $B$, therefore be detected by the conditional fidelity. 
		We further reduce this problem to a random linear-algebra problem induced by the random adjacency matrix $A$, and prove this is indeed the case.
		The detailed proof is given in Appendix~\ref{app:proof_spectral_gap}. 
	\end{proof}
	
	Our result of random graph states directly implies the same constant-gap behavior for the corresponding ensemble of random stabilizer states via local Clifford rotations, as these rotations preserve randomized Pauli measurements up to relabeling. 
	To our knowledge, this is the first proof that local random Pauli measurements suffice to certify generic stabilizer states with constant sample complexity and robustness.
	Compared with the standard certification method based on measuring stabilizer operators, our approach is state-independent, as it does not require prior knowledge of the target stabilizer group before measurements. 
	Hence, a single dataset of local Pauli measurements can be reused to certify many target stabilizer states simultaneously.
	
	We remark that prior analyses of non-adaptive local certification have mainly focused on Haar-random target states~\cite{Huang2024Certifying}. 
	Our Theorem~\ref{thm:constant-gap-graph-states} already gives a constant-gap guarantee for a physically motivated and efficiently preparable non-Haar ensemble. 
	This shows that constant-gap behavior is not restricted to highly random ensembles, but can arise in physically relevant states.
	
	\revise{
		\subsection{Provable constant gap for almost all states}\label{subsec:fidelity-Haar-random}
		
		We now establish the performance guarantees for generic pure states. 
		We focus first on Haar-random states, which benchmark the average-case performance of our protocol across the entire Hilbert space. 
		While we show in Appendix~\ref{app:proof_spectral_gap} that the spectral gap $\Delta(O_{\psi})$ is strictly positive almost surely, efficient certification demands a much stronger condition. 
		Specifically, an exponentially small positive gap would still yield poor robustness and prohibitive sample complexity. 
		Our main fidelity theorem overcomes this by proving that, for all but an exponentially small fraction of states, the spectral gap is strictly lower-bounded by a constant.
		
		\begin{theorem}[Constant gap for almost all states]\label{thm:large_spectral}
			Let $\ket{\psi}$ be drawn Haar-randomly from the $n$-qubit Hilbert space $\mathcal{H}_d$ with $d=2^n$ and $n_A \ge 1$. Then, with probability $1-\exp(-\Omega(n))$, $\Delta(O_{\psi}) = \Omega(1)$.
		\end{theorem}
	}
	
	\revise{
		\begin{proof}[Proof sketch]
			It suffices to consider $n_A=1$, since for any one-qubit subsystem $A_1\subseteq A$, one has $O_{\psi,A} \le O_{\psi,A_1}$ and hence $\Delta(O_{\psi,A})\ge\Delta(O_{\psi,A_1})$.
			Let the target state $\ket{\psi}$ and an orthogonal input pure state $\ket{\phi}$ be decomposed as $\ket{\psi}=\ket{0}\ket{u}+\ket{1}\ket{v}$ and $\ket{\phi}=\ket{0}\ket{a}+\ket{1}\ket{b}$, where the first qubit corresponds to subsystem $A$. For each random Pauli measurement outcome $\bmx$ on subsystem $B$, we denote the projected amplitudes as $u_{\bmx}=\langle\bmx|u\rangle$, and define $v_{\bmx}$, $a_{\bmx}$, and $b_{\bmx}$ analogously. 
			The spectral gap can then be expressed as
			\begin{equation}\label{eq:gap_dirichlet_form}
				\bra{\phi}(I-O_\psi)\ket{\phi} = \sum_{\bmx\in\Bxrange}3^{-n_B} \frac{|u_{\bmx}b_{\bmx}-v_{\bmx}a_{\bmx}|^2}{|u_{\bmx}|^2+|v_{\bmx}|^2}.
			\end{equation}
			By introducing two positive semidefinite operators, $H_0$ and $H_1$ (defined in Eq.~\eqref{eq:H0H1}), we apply the Cauchy--Schwarz inequality to lower-bound Eq.~\eqref{eq:gap_dirichlet_form} by the ratio $\langle\phi|H_0|\phi\rangle^2 / \langle\phi|H_1|\phi\rangle$. 
			Using Gaussian random matrix concentration inequalities and the Wick integration formula for Gaussian random variables, we prove that, with probability $1-\exp[-\Omega(n)]$, the following conditions hold simultaneously:
			\begin{equation}
				\begin{split}
					\ker(H_0) &= \operatorname{span}\{\ket{\psi}\}, \\
					\lambda_{\min}^{+}(H_0) &= \Omega(2^{-n}), \\
					\|H_1\| &= \mathcal{O}(2^{-2n}).
				\end{split}
			\end{equation}    
			The dimensional scaling factors exactly cancel in this ratio, yielding a constant lower bound of $\Omega(1)$. The complete proof is provided in Appendix~\ref{app:proof_spectral_gap}.
		\end{proof}
		
		Combining Theorem~\ref{thm:large_spectral} and Proposition~\ref{prop:fidelity_performance} establishes that for any chosen parameter $c<1/2$, a given fidelity $F$ can be certified using $\cO(1)$ samples for all but a $2^{-\Omega(n)}$ fraction of pure states. 
		Moreover, this approach tolerates an infidelity of $(1-2c)\Delta(1-F) = \Omega(1)$. 
		This performance substantially improves upon previous protocols requiring $\cO(n^4)$ samples and $\Omega(n^{-2})$ robustness~\cite{Huang2024Certifying}.
	}
	
	\revise{To test whether this behavior persists for more experimentally accessible targets, we evaluated the practical performance of our protocol using numerical simulations on physically relevant, non-Haar states.}
	We sampled $N_s=50$ states $\{\psi_j\}_{j=1}^{N_s}$ generated by local quantum circuits of depth $10n$. 
	The depth is chosen so that the states are far from Haar-random states, whose preparation requires exponentially deep circuits. 
	To estimate the spectral gap of $O_{\psi}$ for these states, we draw $N_M=100$ random local Pauli bases $\bmb_1,\bmb_2,\ldots,\bmb_{N_M} \in \{X,Y,Z\}^{n_B}$ and define the truncated observable
	\begin{equation}
		\begin{split}
			O^{(i)}_{\psi}
			= \frac{1}{i}\sum_{j=1}^{i} \sum_{\bm z: \tilde{p}(\bmx(\bmb_j,\bm z)) \neq 0}
			&\ketbra{\psi_{\bmx(\bmb_j,\bm z)}}_A \\ 
			\otimes &\ketbra{\bmx(\bmb_j,\bm z)}_B,
		\end{split}
	\end{equation}
	where $1 \le i \le N_M$ and $\bmx(\bmb_j,\bm z) \in \Bxrange$ denotes the basis state corresponding to measurement outcome $\bm z$ in basis $\bmb_j$. The previously defined $O_{\psi}$ is recovered in the limit $i \to \infty$.
	For $n_A=1$, we compute the averaged spectral gap
	\begin{equation}
		\bar{\Delta}_i
		= \frac{1}{N_s} \sum_{j=1}^{N_s} \Delta\left(O^{(i)}_{\psi_j}\right)
	\end{equation}
	over the sampled random states and plot the results in Fig.~\ref{fig:numerics_fidelity_enhanced}a. As $i$ increases, $\bar{\Delta}_i$ grows, indicating that the average gap $\bar{\Delta}_{\infty} = \Delta(O_{\psi})$ exceeds $\bar{\Delta}_{N_M}$; thus $\bar{\Delta}_{N_M}$ gives a good empirical lower bound of $\Delta(O_{\psi})$.
	
	\begin{figure}[!htbp]
		\centering
		\includegraphics[width=.5\textwidth]{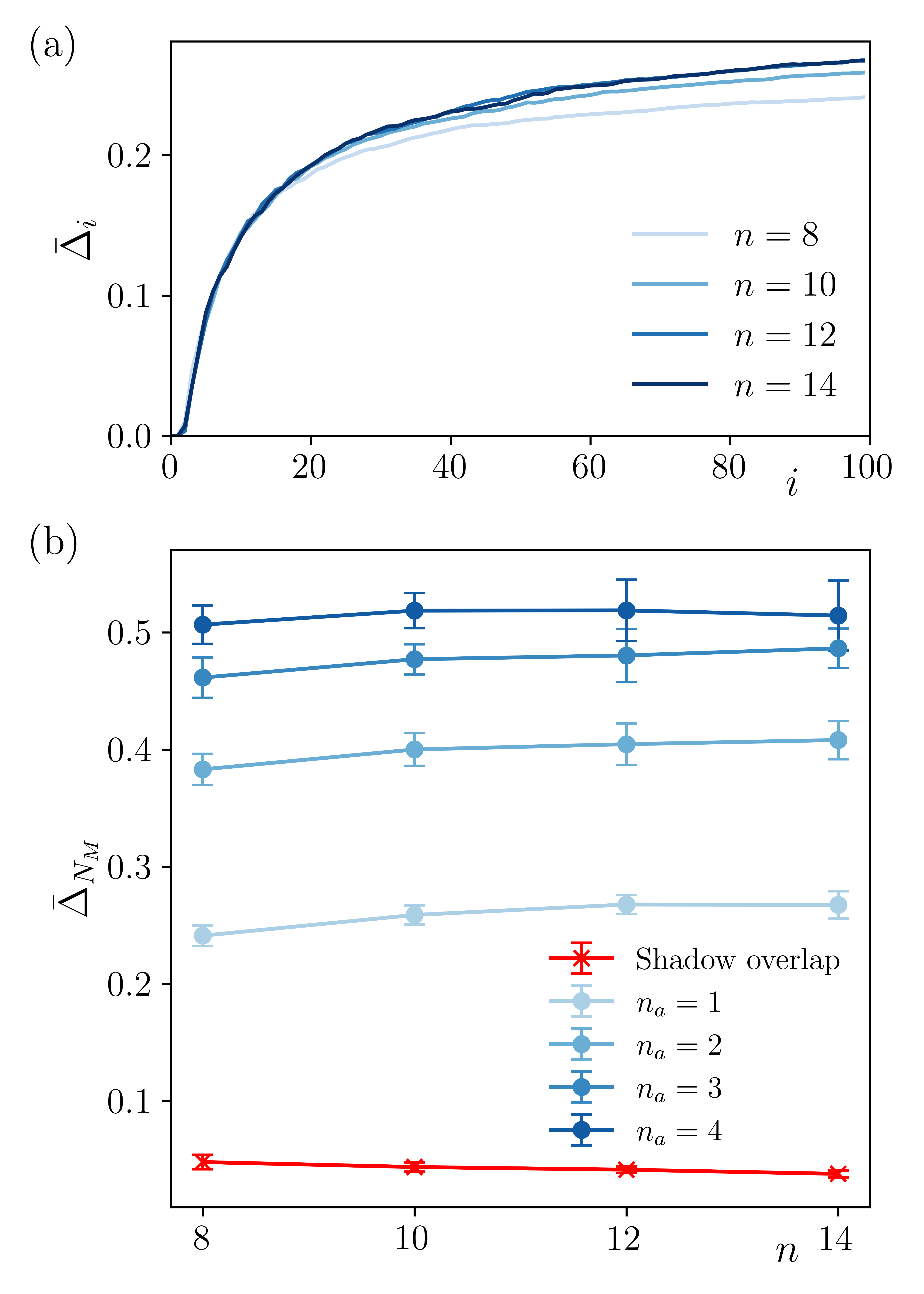} 
		\caption{Spectral gaps of conditional-fidelity observables. (a) The average spectral gap $\bar{\Delta}_i \coloneqq \frac{1}{N_s}\sum_{j=1}^{N_s}\Delta\left(O_{\psi_j}^{(i)}\right)$ for $1\le i\le N_M=100$ random bases increases with $i$, indicating that the limit $O_{\psi}$ (as $i\to\infty$) has a larger average gap than $\bar{\Delta}_{N_M}$. (b) $\bar{\Delta}_{N_M}$ versus $n$ (error bars show the standard deviation across states). Here $\bar{\Delta}_{N_M}$ serves as a proxy for $\bar{\Delta}_{\infty}$, suggesting that the average spectral gap of $O_{\psi}$ remains $\Omega(1)$ as $n$ grows. 
			Moreover, $\bar{\Delta}_{N_M}$ increases with $n_A$, enhancing the protocol’s robustness. 
			For comparison, the shadow-overlap protocol exhibits a smaller spectral gap that decays with $n$.}
		\label{fig:numerics_fidelity_enhanced}
	\end{figure} 
	
	\revise{
		As shown in Fig.~\ref{fig:numerics_fidelity_enhanced}a, the averaged gap rapidly increases and approaches a stable value as more local basis patterns are included.
		Furthermore, Fig.~\ref{fig:numerics_fidelity_enhanced}b shows that the finite-basis gap remains bounded away from zero over the accessible system sizes, in sharp contrast with the decreasing shadow-overlap gap~\cite{Huang2024Certifying}.
		These data clearly indicate that the constant-gap phenomenon established in Theorem~\ref{thm:large_spectral} is not restricted to Haar-random states but extends robustly to physically generated states.
	}
	While enlarging $n_A$ can further increase the spectral gap, thereby improving robustness, it also increases estimator variance as $\mathcal{O}(4^{n_A})$, yielding a trade-off in sample complexity.
	The ability to maintain a large constant gap shows that our protocol is highly practical for certifying and benchmarking intermediate-scale components of near-term devices, using significantly fewer quantum resources while tolerating larger device imperfections.

	\section{Discussion and outlook} \label{sec:discussion}
	
	\revise{In this work, we introduced the physical principle of localizable quantumness and developed a conditional fidelity certification framework, establishing a powerful approach for certifying central quantum properties. 
		By relying only on state-independent, non-adaptive local Pauli measurements, our protocols guarantee mixed-input soundness while yielding substantial improvements in both robustness and sample complexity. 
		Crucially, we resolved the open challenge regarding state fidelity certification by proving that a strictly constant spectral gap, $\Delta(O_{\psi}) = \Omega(1)$, and hence robust certification with constant sample complexity, is achievable for generic quantum states via local measurements. 
		Consequently, our unified framework achieves optimal, constant-sample-complexity scaling for certifying a broad range of many-body properties.
		While our current circuit-complexity witnesses primarily target relatively shallow circuits, measurements are known to concentrate exponentially large circuit complexity into small subsystems~\cite{Du2025SpacetimeComplexity}. 
		Harnessing this complexity concentration phenomenon to design sample-efficient witnesses for highly complex, deep circuits represents a compelling direction for future research.}
	
	Localizable quantumness offers a unifying lens on several previously established effects. For example, complexity concentration~\cite{Du2025SpacetimeComplexity}, localizable entanglement~\cite{Verstraete2004LocalizedEntanglement, Popp2005LocalizableEntanglement}, magic concentration~\cite{Zhang2024MagicConcentration, Niroula2024MagicPhaseTransition, Loio2025quantumstatedesignsmagic}, and magic-state injection~\cite{Bravyi2005Universal} can all be viewed as special cases in which a global resource localizes into a small subsystem after local projective measurements on the complement. 
	This highlights the broader significance of projected-ensemble properties such as $\LQ$ and their potential applications and impact on resource theories and many-body physics. 
	In condensed matter physics, it is especially compelling to study order parameters and phases through the lens of localizable quantumness~\cite{Ge2025BipartiteProjectedEnsemble}. 
	Another exciting direction is to develop a theory of localizable magic, paralleling localizable entanglement~\cite{Verstraete2004LocalizedEntanglement}, as a tool to probe quantum many-body magic.
	
	Extending localizable quantumness to mixed states and developing protocols that certify mixed-state properties, especially mixed-state entanglement and magic, represent especially promising directions.
	Recent work also shows that, with tolerable additional overhead, one can obtain good bounds on nonlinear quantities of the projected ensemble, enabling powerful certifications such as emergent quantum state designs~\cite{McGinley2024PostselectionFreeLearning}. It is therefore compelling to investigate whether nonlinear projected-ensemble quantities can certify stronger quantum properties or yield sharper, more powerful property-detection criteria.
	
	More broadly, our approach complements and extends existing paradigms for certification and estimation such as direct fidelity estimation \cite{Flammia2011DirectFidelityEstimation} and randomized measurements \cite{Huang2020Predicting, Elben2022toolbox} by showing that simple local measurements can certify global properties with favorable scaling. 
	A particularly promising avenue for future work is to tailor our method to other physical platforms, leveraging native, low-overhead measurements to certify key properties—for example, nonclassicality in bosonic systems \cite{Mari2011Nonclassicality} and non-Gaussianity in fermionic systems \cite{Hebenstreit2019AllPureNonGaussian, Gluza2018FidelityWitnessFermionic}.

	\begin{acknowledgments}
		We are grateful to Zhenhuan Liu and Muzhou Ma for insightful discussions in the early stages of this project. We thank Shuo Yang for helpful suggestions on the numerical experiments. 
		This work is supported by the National Natural Science Foundation of China (Grant No. 12575023), the Quantum Science and Technology-National Science and Technology Major Project (Grants No.~2021ZD0300804 and No.~2021ZD0300702), the CCF-QuantumCtek Superconducting Quantum Computing Special Cooperation Program (Grant No.~CCF-QC2025005), and the Turing AI Institute of Nanjing.
		Z.-W.L. is supported in part by NSFC under Grant No.~12475023, Dushi Program, and a startup funding from YMSC. 
		Parts of the numerical results were obtained using the ITensor~\cite{Fishman2022ITensor} and Qiskit~\cite{qiskit2024Quantum} software packages.
		
		\textbf{Data Availability:} The data that support the findings of this article are openly available~\cite{Du2025CertifyingLocalizableRepo}.
	\end{acknowledgments}

	
	\onecolumngrid
	\appendix

	\renewcommand{\thetheorem}{S\arabic{theorem}}
	\renewcommand{\thefact}{S\arabic{fact}}
	\renewcommand{\thelemma}{S\arabic{lemma}}
	\renewcommand{\theequation}{\thesection\arabic{equation}}
	\renewcommand{\thedefinition}{S\arabic{definition}}
	\renewcommand{\theproposition}{S\arabic{proposition}}
	\renewcommand{\thecorollary}{S\arabic{corollary}}
	\renewcommand{\theclaim}{S\arabic{claim}}
	\setcounter{theorem}{0}
	\setcounter{fact}{0}
	\setcounter{lemma}{0}
	\setcounter{equation}{0}
	\setcounter{definition}{0}
	\setcounter{proposition}{0}
	\setcounter{claim}{0}
	\setcounter{corollary}{0}
	
	\renewcommand{\theHtheorem}{S\arabic{theorem}}
	\renewcommand{\theHfact}{S\arabic{fact}}
	\renewcommand{\theHlemma}{S\arabic{lemma}}
	\renewcommand{\theHdefinition}{S\arabic{definition}}
	\renewcommand{\theHproposition}{S\arabic{proposition}}
	\renewcommand{\theHcorollary}{S\arabic{corollary}}
	\renewcommand{\theHclaim}{S\arabic{claim}}
	\renewcommand{\theHfigure}{S\arabic{figure}}

	\section{Additional analysis on the performance of Protocol \ref{prot:certification}}\label{app:prelim}
	
	In this appendix, we analyze the performance of Protocol~\ref{prot:certification}. We first collect the statistical tools used throughout our work, including local shadow tomography, standard concentration inequalities, and the median-of-means estimator. We then establish the performance guarantee based on those tools.
	
	\subsection{Local classical shadow}
	We briefly review the local classical shadow protocol~\cite{Huang2020Predicting}. 
	For a given $k$-qubit input state $\rho$, the protocol performs random single-qubit Pauli measurements in the $X$, $Y$, or $Z$ basis independently on each qubit.  
	This is equivalent to applying a POVM composed of six elements:
	\begin{equation}
		\left\{\frac{1}{3}\ketbra{+},\, \frac{1}{3}\ketbra{-},\, \frac{1}{3}\ketbra{+i},\, \frac{1}{3}\ketbra{-i},\, \frac{1}{3}\ketbra{0},\, \frac{1}{3}\ketbra{1}\right\}.
	\end{equation}
	Suppose the POVM yields measurement outcomes $\bmx \in \mathsf{X}^k$.
	The classical shadow estimator for $\tr(O\rho)$ is given by
	\begin{equation}
		\omega = \tr\Bigl\{O\Bigl[\bigotimes_{i \in [k]} (3\ketbra{\bmx_i} - I_i)\Bigr] \Bigr\}.
	\end{equation}
	This estimator is unbiased:
	\begin{equation}
		\bE[\omega] = \tr(O\rho),
	\end{equation}
	and its variance is bounded by
	\begin{equation}\label{eq:local_shadow_variance}
		\Var[\omega] \le 4^k \norm{O}_{\infty}^2
	\end{equation}
	
	Crucially, $\omega$ uses only a single copy of $\rho$, which is essential for estimating ensemble-averaged local fidelities in our protocol.
	
	\begin{lemma}
		Let $\bm z\sim p$ and $\omega$ be a random variable such that $\bE[\omega| \bm z]=o_{\bm z}\ge 0$ and $\operatorname{Var}[\omega| \bm z]\le \sigma^{2}$ for all $\bm z$. 
		Then, $\bE[\omega] = \sum_\bmz p(\bmz) o_\bmz$ and $\Var[\omega] \le \sigma^2 + \max_{\bmz}\abs{o_{\bmz}}^2$. 
	\end{lemma}
	\begin{proof}
		The proof is from the laws of total expectation and variance,
		\begin{equation}\label{eq:variance_conditional_fidelity}
			\begin{split}
				\bE \omega &= \bE_{\bmz}\bE[\omega|\bmz] \\
				&= \sum_{\bmz} p(\bmz) o_\bmz, \\
				\Var[\omega] &= \bE_z \Var[\omega|\bmz] + \Var_{\bmz}\bE[\omega|\bmz]  \\
				&\le \bE_{\bmz} \sigma^2 + \Var_{\bmz} o_{\bmz}\\
				&\le \sigma^2 + \max_{\bmz}\abs{o_{\bmz}}^2.
			\end{split}
		\end{equation}
	\end{proof}
	
	This enables unbiased estimation of projected-ensemble–averaged linear quantities~\cite{McGinley2024PostselectionFreeLearning}. 
	Combining the lemma with the variance bound in~\eqref{eq:local_shadow_variance} yields the following guarantee for estimating conditional fidelities.
	
	\begin{corollary}[Variance of conditional fidelity]\label{col:var_conditional_fidelity}
		Let an $n$-qubit state $\rho$ be partitioned as $A\cup B$, and measure $B$ in the computational basis to obtain $\bm z$ with probability $p_\rho(\bm z)$ and post-measurement state $\rho_{\bm z}$ on $A$. Construct $\omega$ by applying local classical shadows on $\rho_{\bmz}$ to estimate $\tr(\psi_{\bm z}\rho_{\bm z})$. Then, $\bE[\omega] = \sum_{\bmz} p_{\rho}(\bmz) \tr(\psi_{\bmz}\rho_{\bmz})$ and $\Var[\omega] \le 4^{n_A} + 1$.
	\end{corollary}

	\subsection{Median-of-means estimator}
	The median-of-means estimator is a standard statistical method that uses medians to suppress the contribution of outlier estimators, which has the benefit of exponentially suppressing the failure probability.
	We first introduce Hoeffding's inequality, which is useful for bounded-value estimators:
	
	\begin{proposition}[Hoeffding's inequality, {\cite[Theorem 8]{Kliesch2021CertificationSurvey}}]\label{prop:Hoeffding}
		Let $\{\omega_i\}_{i=1}^{T}$ be i.i.d. random variables with $a_i \le \omega_i \le b_i$, and define $S_T \coloneqq \sum_{i=1}^{T} \omega_i$.  For any $t>0$,
		\begin{equation}
			\Pr[S_T - \bE[S_T] \ge t] \le \exp\left[-\frac{2t^2}{\sum_{i=1}^T(b_i - a_i)^2}\right]
		\end{equation}
	\end{proposition}

	Then, we introduce the median-of-means estimator:
	\begin{proposition}[Median-of-means estimator]\label{prop:mom-estimator}
		Let $T=BK$ and $\{\omega_i\}_{i=1}^{T}$ be i.i.d. random variables with mean $\mu$ and variance $\sigma^{2}$.  
		For each block $l\in\{1,\dots,K\}$ define the empirical mean
		\begin{equation}
			o_l \coloneq  \frac{1}{B}\sum_{i=(l-1)B+1}^{lB} \omega_i,
		\end{equation}
		and let $\omega = \mathrm{median}\{o_1,o_2,\cdots,o_K\}$.
		Fix $\delta\in(0,1)$ and $\varepsilon>0$.  By choosing $B = \frac{6\sigma^2}{\varepsilon^2}$, $K = \frac{9}{2} \ln(\delta^{-1})$, we have that
		\begin{equation}
			\Pr[\abs{\omega-\mu} \ge \varepsilon] \le \delta.
		\end{equation}
		Equivalently, with probability at least $1-\delta$,
		\begin{equation}\label{eq:claim-distance-from-mean-MoM}
			\abs{\omega - \mu} < \sigma \sqrt{\frac{27 \ln (\delta^{-1})}{T}}
		\end{equation}
	\end{proposition}
	
	\begin{proof}
		For each block mean $o_l$ we have $\bE[o_l]=\mu$ and
		$\Var[o_l] = \sigma^2/B = \varepsilon^2 / 6$.  
		By Chebyshev’s inequality,
		\begin{equation}
			\Pr[\abs{o_l - \mu} \ge \varepsilon] \le \frac{\Var[o_l]}{\varepsilon^2} =\frac{1}{6}.
		\end{equation}
		Define the indicator $\tilde{o}_l \coloneqq \mathbf{1}\bigl[\lvert o_l-\mu\rvert \ge \varepsilon\bigr]$, then $\bE[\tilde{o}_l]\le 1/6$ and $0\le\tilde{o}_l\le 1$.  
		Applying Proposition~\ref{prop:Hoeffding} to $\sum_{l=1}^{K}\tilde{o}_l$, we have that:
		\begin{equation}
			\begin{split}
				\Pr[\abs{\omega-\mu}\ge \varepsilon] &\le \Pr[\sum_{i=1}^{K} \tilde{o}_l \ge \frac{K}{2}] \\ 
				&\le \Pr[\sum_{i=1}^{K} \tilde{o}_l - \bE[\sum_{i=1}^{K} \tilde{o}_l]\ge \frac{K}{3}] \\
				& \le \exp\left[-\frac{2(\frac{K}{3})^2}{K}\right] = \delta,
			\end{split}
		\end{equation}
		completing the proof.
	\end{proof}
	
	\subsection{Proof of Theorem \ref{thm:performance_certification_protocol}}
	\begin{proof}
		\emph{Soundness}---For any property-free state $\rho \in \freeset$, we have $\eta_{\psi}(\rho)\le 0$ by Lemma \ref{lem:eta_witness}. Moreover, $\bE[\omega_i - \fid(\psi_{\bmz_i})] = \eta_{\psi}(\rho)$.
		Applying the median-of-means estimator (Proposition \ref{prop:mom-estimator}) to $T$ independent estimates $\{\omega_i - \fid(\psi_{\bmz_i})\}_{i=1}^{T}$ yields a final estimator $\omega$ satisfying
		\begin{equation}\label{eq:error_MoM}
			\abs{\omega - \eta_{\psi}(\rho)} < \sigma \sqrt{\frac{27 \ln (\delta^{-1})}{T}}.
		\end{equation}
		with probability at least $1-\delta$, where $\sigma^{2}=4^{n_A}+1$ by Corollary~\ref{col:var_conditional_fidelity}.
		With the choice of $T$ in Eq.~\eqref{eq:T_value}, we obtain, with probability at least $1-\delta$,
		\begin{equation}
			\omega < \eta_{\psi}(\rho)+\LQ(\psi)/{3} \le \eta^\ast,
		\end{equation}
		So the protocol outputs $\reject$.  
		
		\emph{Completeness}---Since the experimental state $\rho$ is $\varepsilon$-close to the ideal state $\psi$ in trace distance, expectation values of bounded observables are also close, ensuring that $\eta_{\psi}(\rho)$ is close to the ideal value. Specifically, for a given $\rho$ satisfying $\dtr(\rho,\psi) \le \varepsilon$, we have that 
		\begin{equation}
			\begin{split}
				\eta_{\psi}(\rho) &= \tr(\propwitness \psi) + \tr(\propwitness (\rho-\psi)) \\
				&\ge \LQ(\psi) - \norm{\propwitness}_{\infty} \norm{(\rho-\psi)}_1\\
				&\ge  \LQ(\psi) - 2\dtr(\rho,\psi) \\
				&\ge \frac{2}{3}\LQ(\psi).
			\end{split}
		\end{equation}
		
		Choosing $T$ as in Eq.~\eqref{eq:T_value} ensures
		\begin{equation}
			\omega > \eta_{\psi}(\rho) - \frac{\LQ(\psi)}{3} \ge \eta^\ast
		\end{equation}
		with probability at least $1-\delta$. 
	\end{proof}
	
	\section{Additional proofs on complexity certification} \label{app:additional_proofs_complexity}
	
	\subsection{Proof of Lemma \ref{lem:entanglement_free_projected_state}: entanglement of low-complexity states}
	\begin{proof}
		Because $C_2(\phi)\le d$, the state $\ket{\phi}$ can be prepared from $\ket{0}^{\otimes n}$ by a two-dimensional depth-$d$ circuit
		\begin{equation}
			U = \prod_{i=1}^d \left(\bigotimes_j V_{i,j}\right)
		\end{equation}
		where the gates $\{V_{i,j}\}_{j}$ in the $i$-th layer act on disjoint pairs of nearest-neighbour qubits or on single qubits. Define the backward light cone of the subsystem $B$ layer by layer:
		\begin{equation}
			\begin{split}
				\cL_{B,d+1} &= B, \\
				\cL_{B,i} &= \cL_{B,i+1}\cup \left[\bigcup \left\{\mathrm{supp}(V_{i,j}) \;\big|\; \mathrm{supp}(V_{i,j}) \cap \cL_{B,i+1} \neq \emptyset\right\}\right].
			\end{split}
		\end{equation}
		Thus $\mathcal{L}_{B,1}$ contains all qubits that can interact with $B$ through at most $d$ layers.  By construction of $L_1$ and $R_1$ (see Fig.~\ref{fig:projected_ensemble_complexity}), 
		\begin{equation}\label{eq:backward_lightcone_L1R1}
			(\cL_{B,1} \cap L) \subseteq L_1, \quad (\cL_{B,1} \cap R) \subseteq R_1
		\end{equation}
		
		The unitary $V$ is defined as the product of all gates whose supports lie outside the backward light cone of $B$:
		\begin{equation}
			V = \prod_{i=1}^d \left(\bigotimes_{\mathrm{supp}(V_{i,j})\not\subseteq \cL_{B,i}} V_{i,j}\right),
		\end{equation}
		The remaining two-qubit gates define the unitary
		\begin{equation}
			W = \prod_{i=1}^d \left(\bigotimes_{\mathrm{supp}(V_{i,j})\subseteq \cL_{B,i}} V_{i,j}\right).
		\end{equation}
		By construction, $U = VW$, and $V$ acts only on subsystem $A$. Consequently, $V$ commutes with the operator $\bra{\bmz}_B \otimes I_A$ for every measurement outcome $\bmz \in \Bzrange$. The unnormalised projected state can therefore be written as
		\begin{equation}\label{eq:commute_V_Z}
			\begin{split}
				\ket{\widetilde{\phi_{\bmz}}} &= (\bra{\bmz}_{B} \otimes I_A) \ket{\phi} \\
				&=  (\bra{\bmz}_{B} \otimes I_A)  VW\ket{0}_{AB} \\
				&= V (\bra{\bmz}_{B} \otimes I_A) W\ket{0}_{AB} 
			\end{split}
		\end{equation}
		Eq.~\eqref{eq:backward_lightcone_L1R1} implies that $W$ acts non-trivially only on $B L_1 R_1 $.  Consequently,
		\begin{equation}
			\begin{split}
				V (\bra{\bmz}_{B} \otimes I_A) W\ket{0}_{AB}&= V [(\bra{\bmz}_{B}  \ket{\phi_1}_{BL_1L_2})\ket{0}_{L_2R_2}] \\
				&= V \left(\ket{\widetilde{\varphi}}_{L_1R_1} \ket{0}_{L_2R_2}\right)
			\end{split}
		\end{equation}
		Normalizing this state, we obtain
		\begin{equation}\label{eq:projected_state_low_complexity}
			\ket{\phi_{\bmz}} =  V \left(\ket{\varphi}_{L_1R_1} \otimes \ket{0}_{L_2R_2}\right).
		\end{equation}
		
		The circuit $V$ has depth $d$, and in each layer at most $2w$ two-qubit gates can cross the cut $L\mid R$, so at most $2wd$ such gates appear in total.  Each gate increases the bipartite entanglement entropy by at most $2$~{\cite[Lemma~4]{Harrow2021SeparationOTOC}}.  Hence
		\begin{equation}
			\begin{split}
				E_{L\mid R}(\phi_{\bmz}) &\le E_{L\mid R}(\ket{\varphi}_{L_1R_1} \otimes \ket{0}_{L_2R_2}) + 4wd \\
				&= E_{L_1\mid R_1}(\varphi_{L_1R_1}) + 4wd \\
				&\le 8wd.
			\end{split}
		\end{equation}
		where we used $E_{L_1\mid R_1}(\varphi_{L_1R_1})\le|L_{1}|\le 4w d$.  This completes the proof.
	\end{proof}
	
	\subsection{Proof of Theorem \ref{thm:2Dcomplexity}: certifying unitary circuit complexity}
	
	We first establish the following lemma, which gives an upper bound on the maximal fidelity of highly entangled projected states.
	
	\begin{lemma}\label{lem:fid_proj_entanglement}
		Let $\phi$ be a pure state on subsystem~$A$ with bipartite entanglement $E_{L\mid R}(\phi)=e>8wd+1$.  Then
		\begin{equation}\label{eq:fid_bound_entanglement}
			\fid(\phi) \le 1 - \left(\frac{e-1}{4w^2} - \frac{2d}{w}\right)^2.
		\end{equation}
	\end{lemma}
	\begin{proof}
		For any pure state $\varphi \in \freeprojset$, 
		\begin{equation}
			\begin{split}
				E_{L\mid R}(\phi) - E_{L\mid R}(\varphi) &\le  \abs{E_{L\mid R}(\phi) - E_{L\mid R}(\varphi)} \\
				& \le 2 \dtr(\phi_L, \varphi_L) \abs{L} + 1.
			\end{split}
		\end{equation}
		where the last step uses Fannes’ inequality~\cite[Theorem~11.6]{Nielsen2012QCQI}.  Because trace distance is non-increasing under partial trace, $\dtr(\phi_L, \varphi_L) \le \dtr(\phi,\varphi)$. With $\abs{L} = 2w^2$ and $E_{L\mid R}(\varphi) \le 8wd$ as given in Lemma \ref{lem:entanglement_free_projected_state}, we obtain
		\begin{equation}
			\dtr(\phi,\varphi) \ge \frac{e - 8wd - 1}{4w^2} = \frac{e-1}{4w^2} - \frac{2d}{w}.
		\end{equation}
		Finally,
		\begin{equation}
			\tr(\phi\varphi) = 1-\dtr^2(\phi,\varphi) \le 1 - \left(\frac{e-1}{4w^2}-\frac{2d}{w}\right)^2,
		\end{equation}
		which completes the proof.
	\end{proof}
	
	Then we prove the following theorem.
	\begin{theorem}
		Let $\psi$ be the target state and consider the bipartition $A\cup B$ shown in Fig.~\ref{fig:projected_ensemble_complexity}, with $w=\eta d$ for some $\eta>4$.  Assume projected ensemble of $\psi$ on $A$ satisfies
		\begin{equation}\label{eq:ineq_prob_projected_state_entanglement}
			\Pr_{\bmz \sim p_{\psi}(\bmz)}\left[E_{L \mid R}(\psi_{\bmz}) \ge c \abs{L} + 1\right] \ge p > 0,   
		\end{equation}
		for a constant $c>4\eta^{-1}$.  Then Protocol~\ref{prot:certification} certifies that the circuit complexity of $\psi$ exceeds $d$ with constant failure probability~$\delta$, using
		\begin{equation}
			T = \cO\left(\frac{2^{\cO(\eta^2d^2)}}{(c-4\eta^{-1})^4 p^2} \ln(\delta^{-1})    \right)
		\end{equation}
		copies of $\psi$ and achieving robustness $\varepsilon = \Omega\bigl((c-4\eta^{-1})^2p\bigr)$.
	\end{theorem}
	\begin{proof}[Proof of Theorem \ref{thm:2Dcomplexity}]
		Lemma~\ref{lem:fid_proj_entanglement} and Eq.~\eqref{eq:ineq_prob_projected_state_entanglement} give
		\begin{equation}
			\Pr_{\bmz \sim p_{\psi}(\bmz)}\left[\fid(\psi_{\bmz}) \le 1 - \tfrac{1}{4}(c-4\eta^{-1})^2 \right] \ge p > 0,  
		\end{equation}
		which implies $\LQ(\psi) \ge \frac{1}{4}(c-4\eta^{-1})^2 p$.
		Substituting this into Theorem~\ref{thm:performance_certification_protocol} with variance $\sigma^{2}=4^{n_A}=2^{\cO(\eta^{2}d^{2})}$ produces the stated sample complexity and robustness.
	\end{proof}
	
	Theorem~\ref{thm:2Dcomplexity} follows immediately by substituting $c-4\eta^{-1}=\Omega(1)$ and $p=\Omega(1)$ into the preceding theorem.
	
	\subsection{Proof of Lemma \ref{lem:ent_adaptive_circuit}: entanglement in states of low measurement-assisted circuit complexity}
	
	\begin{proof}
		
		We analyze the preparation of $\ket{\phi}$, considering more and more layers while tracking the expected bipartite entanglement entropy $\bE_{m}[E_{L\mid(R\cup B)}(\phi_{m})]$, where $m$ denotes the outcomes of mid-circuit measurements in the considered layers, and $\phi_{m}$ denotes the state after performing the considered layers given the measurement outcomes $m$.
		In each layer of unitary gates, at most $6w$ two-qubit gates can cross the cut $L\mid(R\cup B)$. Each such gate can increase the bipartite entanglement entropy by at most $2$~\cite[Lemma~4]{Harrow2021SeparationOTOC}. Therefore, the increase in entanglement entropy across the cut per layer is at most $12w$.
		For measurement rounds, local projective measurements cannot increase entanglement entropy on average~\cite{Vidal2000EntanglementMonotones}. Therefore,  after a measurement layer, the expected entanglement entropy across the cut does not increase.
		After $d$ layers in total, combining the above two observations yields
		\begin{equation}
			E_{L\mid(R\cup B)}(\phi)
			=\bE_{m}\left[E_{L\mid(R\cup B)}(\phi_{m})\right]
			\;\leq\;12wd.
		\end{equation}
		Finally, applying again the monotonicity of entanglement entropy under local measurements when measuring $B$, we obtain
		\begin{equation}
			\bE_{\bmz \sim p_{\phi}(\bmz)}[E_{L\mid R}(\phi_{\bm z})]\leq E_{L\mid (R\cup B)}(\phi)\leq 12wd.
		\end{equation}
	\end{proof}

	\subsection{Proof of Theorem \ref{thm:2Dadaptive_complexity}: certifying measurement-assisted circuit complexity}
	
	We begin by analyzing the performance of the linear witness $\tilde{O}$ given in Eq.~\eqref{eq:adaptive_complexity_witness}. Define the set
	\begin{equation}
		\pfreeprojset \coloneqq\Bigl\{\phi:\;E_{L\mid R}(\phi)\le\frac{12wd}{1-p}\Bigr\}.
	\end{equation}

	We call a set $\mathcal{S}$ a $p$-likely free projected-state set with respect to property $\mathsf{P}$ if, for every pure $\mathsf{P}$-free state $\psi\in\freeset$,
	\begin{equation}
		\Pr_{\psi_{\bm z}\sim\mathcal{E}(\psi)}\bigl[\psi_{\bm z}\in \mathcal{S}\bigr]\ge p.
	\end{equation}
	The set $\pfreeprojset$ is a $p$-likely free projected-state set for the property $\mathsf{P}$ of ``measurement-assisted circuit complexity $C_2^{anc}> d$''. 
	This follows from Lemma~\ref{lem:ent_adaptive_circuit} together with Markov’s inequality.
	
	We define the maximal fidelity of a state against this $p$-likely-free set as
	\begin{equation}
		\fidp(\psi_{\bm z}) \coloneqq\max_{\rho\in\pfreeprojset}\tr(\rho\,\psi_{\bm z}),
	\end{equation}
	and its cumulative distribution function
	\begin{equation}
		\culprobp\coloneqq\Pr_{\psi_{\bm z}\sim\mathcal{E}(\psi)}\left[\fidp(\psi_{\bm z})\le t\right].
	\end{equation}
	With these definitions, the linear witness is formally defined as
	\begin{equation}
		\threswitnessnewp \coloneqq \sum_{\bm z}\bigl[t+(1-t)\Theta_t\bigl(\fidp(\psi_{\bm z})\bigr)\bigr]\, \ketbra{\psi_{\bm z}}_A\otimes \ketbra{\bm z}_B.
	\end{equation}
	\begin{lemma}[Performance of witness $\threswitnessnewp$]\label{lem:performance_witnessp}
		For any threshold $t\in (0,1)$: (i) (Soundness): If $\rho\in\freeset$, then $\tr(\threswitnessnewp \rho) \le t + (1-p)(1-t)$. (ii) (Completeness): For the target $\psi$, $\tr(\threswitnessnewp \psi) = t+(1-t)\culprobp$. When the gap 
		\begin{equation}
			\Delta(\psi,t,p)\coloneqq (1-t)(\culprobp-(1-p)) > 0,
		\end{equation}
		the witness $\threswitnessnewp$ separates $\psi$ from $\freeset$, certifying that $\psi$ has measurement-assisted circuit complexity greater than $d$.
	\end{lemma}
	
	\begin{proof}
		By definition, any free state $\rho\in\freeset$ satisfies
		\begin{equation}
			\begin{split}
				\tr(\threswitnessnewp \rho)& = \sum_{\bmz: \fidp(\psi_{\bmz}) \le t} p_{\rho}(\bmz) \tr(\psi_{\bmz }\rho_{\bmz})+t\sum_{\bmz: \fidp(\psi_{\bmz}) > t} p_{\rho}(\bmz) \tr(\psi_{\bmz }\rho_{\bmz})\\
				&\le \sum_{\bmz: \fidp(\psi_{\bmz}) \le t} p_{\rho}(\bmz) [t+(1-t)\cdot \mathbf{1}[\rho_{\bmz}\notin \pfreeprojset]]+t\sum_{\bmz: \fidp(\psi_{\bmz}) > t} p_{\rho}(\bmz)\\
				&\le  t+(1-p)(1-t).
			\end{split}
		\end{equation}
		For the target state $\psi$,
		\begin{equation}
			\begin{split}
				\tr(\threswitnessnewp\psi)&=\sum_{\bmz} p_{\psi}(\bmz) [t +(1-t) \Theta_t\bigl(\fidp(\psi_{\bm z})\bigr) ]\\
				&= t+(1-t)\culprobp,
			\end{split}
		\end{equation}
	\end{proof}
	
	Then we establish the lower bound of the witness gap for high-complexity states.
	
	\begin{lemma}\label{lem:adaptive_spectral_gap}
		Let $\psi$ be the target state and consider the bipartition $A\cup B$ shown in Fig.~\ref{fig:projected_ensemble_complexity}, with $w=\eta d$ for some $\eta>6$. Suppose the projected ensemble of $\psi$ on $A$ satisfies
		\begin{equation}
			\Pr_{\bm z\sim p_\psi(\bm z)}\left[E_{L\mid R}(\psi_{\bm z})\ge c|L|+1\right]\ge p>0,
		\end{equation}
		for a constant $c>6p^{-1}\eta^{-1}$. Then, setting
		\begin{equation}
			t=1-\frac{1}{4}\Bigl[c-\frac{6}{\eta(1-p')}\Bigr]^{2},
			\qquad
			p'=1-\frac{p+6\eta^{-1}c^{-1}}2,
		\end{equation}
		the witness gap satisfies
		\begin{equation}\label{eq:gap_adaptive_complexity}
			\Delta(\psi,t,p') \ge \frac{c(\eta cp-6)^3}{8\eta (\eta cp+6)^2}.
		\end{equation}
	\end{lemma}
	
	\begin{proof}
		Let $p'\in(0,1)$ to be determined with $c>6/[\eta(1-p')]$. First, by an argument parallel to that of Lemma~\ref{lem:fid_proj_entanglement}, for any state $\phi$ satisfying $E_{L\mid R}(\phi) > c\abs{L}+1$, we obtain
		\begin{equation}
			\fidpprime(\phi) \le 1 - \left(\frac{e-1}{4w^2} - \frac{3d}{(1-p')w}\right)^2.
		\end{equation}
		
		Substituting the entanglement bound $e \ge c\abs{L}+1$ then yields
		\begin{equation}
			\Pr_{\bmz \sim p_{\psi}(\bmz)}\left[\fidpprime(\psi_{\bmz}) \le 1 - \tfrac{1}{4}[c-6\eta^{-1}/(1-p')]^2 \right] \ge p > 0.
		\end{equation}
		Hence for $t = 1- \frac{1}{4}[c-6\eta^{-1}/(1-p')]^{2}$,
		\begin{equation}
			\culprobpprime = \Pr_{\bmz \sim p_{\psi}(\bmz)}\left[\fidpprime(\psi_{\bmz}) \le t \right] \ge p > 0.
		\end{equation}
		By Lemma \ref{lem:performance_witnessp}, the gap is therefore
		\begin{equation}
			\Delta(\psi,t,p') = (1-t)[\culprobpprime-(1-p')]\ge \tfrac{1}{4}[c-6\eta^{-1}/(1-p')]^{2}[p-(1-p')].
		\end{equation}
		When $\eta c p > 6$, choosing $p'=1-(p+6\eta^{-1}c^{-1})/2$ 
		gives
		\begin{equation}
			\Delta(\psi,t,p') \ge \frac{c(\eta cp-6)^3}{8\eta (\eta cp+6)^2}.
		\end{equation}
	\end{proof}
	
	\begin{proof}[Proof of Theorem \ref{thm:2Dadaptive_complexity}]
		Choose $t,p'$ as specified in Lemma~\ref{lem:adaptive_spectral_gap}, and define $O\coloneqq \tilde{O}(\psi,t,p'), a \coloneqq t+(1-p')(1-t) \ge \sup_{\rho \in \freeset} \tr(\tilde{O} \rho)$ and $b \coloneqq \tr(\tilde{O} \psi)$ according to Lemma~\ref{lem:performance_witnessp}.
		From Lemma~\ref{lem:adaptive_spectral_gap} we know that the gap $\Delta\coloneqq \Delta(\psi,t,p')=b-a=\Omega(1)$.  
		Set the decision threshold to $\eta^\ast\;\coloneqq\;a+\tfrac{\Delta}{3}$.
		
		\textit{Soundness}-For any $\rho\in\freeset$, we have $\tr(\tilde{O}\rho)\le a$. Estimating this quantity to an additive error $\varepsilon=\Delta/3$ ensures
		$\omega < \tr(\tilde{O}\rho)+\varepsilon \le \eta^\ast$, so all property-free states are rejected with high probability.
		
		\textit{Completeness}-For any state $\rho$ satisfying $\mathrm{d}_{\mathrm{tr}}(\rho,\psi)\le \varepsilon$, 
		\begin{equation}
			\tr(\tilde{O}\rho)-\varepsilon
			\ge \tr(\tilde{O} \psi) - \norm{\tilde{O}}_{\infty} \norm{(\rho-\psi)_{+}}_1 -\varepsilon = \eta^\ast.
		\end{equation}
		Here, for a Hermitian operator $M=M_{+}-M_{-}$, $M_{+}$ is its positive part. The inequality follows because $\tilde{O}$ is positive semidefinite. Therefore, the estimator satisfies
		\begin{equation}
			\omega > \tr(\tilde{O}\rho)-\varepsilon \ge \eta^\ast
		\end{equation}
		So $\rho$ is accepted with high probability.
		
		By Proposition \ref{prop:mom-estimator}, the sample complexity required to achieve additive error $\varepsilon$ is given by 
		\begin{equation}
			T = \cO\left(\frac{\sigma^2 \ln(\delta^{-1})}{\varepsilon^2}\right) = \cO(4^{n_A}) = \exp(\cO(\eta^2 d^2)),
		\end{equation}
		where $\sigma^2\le 4^{n_A}+1$. For the two-dimensional setting with $n_A=\mathcal{O}(\eta^2 d^2)$ and constant failure probability, this gives $T = \exp(\cO(\eta^2d^2))$, which completes the proof.
	\end{proof}
	
	\section{Performance analysis of complexity certification for generic states}\label{app:performance_certification_states}
	In this section, we prove the performance of our complexity certification schemes for generic quantum states, including Haar-random states and brickwork-circuit states. 
	
	\subsection{Deep-thermalized and Haar-random states}
	We first establish the performance guarantees for deep-thermalized states. We begin by introducing the basic definitions of quantum state designs and their associated entanglement properties. We then use these tools to establish the desired performance guarantees.
	
	\subsubsection{Entanglement in quantum state designs}
	Quantum state designs are defined as approximations to the Haar random distribution \cite{Brandao2016LocalRandomCircuits}. In this context, we denote by \(\mathrm{Haar}(n)\) the Haar random distribution over $n$-qubit pure states. Given a state ensemble \(\mathcal{E} = \{p_{\psi}, \psi\}\), we define its \(k\)-th moment as
	\begin{equation}
		\rho_{\mathcal{E}}^{(k)} \;=\; \bE_{\psi \sim \mathcal{E}} \left[(\ketbra{\psi})^{\otimes k}\right].
	\end{equation}
	Following \cite{Cotler2023EmergentDesign}, we define an \(\epsilon\)-approximate \(k\)-design as follows:
	\begin{definition}[$\epsilon$-approximate $k$-design]
		An $n$-qubit state ensemble $\mathcal{E}$ is an $\epsilon$-approximate $k$-design if
		\begin{equation}
			\norm{\rho_{\mathcal{E}}^{(k)} - \rho_{\mathrm{Haar}(n)}^{(k)}}_1 \le \epsilon.
		\end{equation}
		where  $\norm{\cdot}_1$ denotes the trace norm. 
	\end{definition}
	Note that an $\epsilon$-approximate \(k_1\)-design is also an $\epsilon$-approximate \(k_2\)-design for any $k_2 \leq k_1$. There is also an alternate definition of approximate state designs, using the name ``frame operator'' (FO).
	
	\begin{definition}[$\epsilon$-FO-approximate $k$-design, {\cite[Definition 10]{Liu2018EntanglementDesign}}]
		An $n$-qubit state ensemble $\mathcal{E}$ is an $\lambda$-FO-approximate $k$-design if
		\begin{equation}
			\norm{\rho_{\mathcal{E}}^{(k)} - \rho_{\mathrm{Haar}(n)}^{(k)}}_1 \le \lambda D_k^{-1}.
		\end{equation}
		where $D_k = \binom{k+N-1}{k}$ and $N= 2^n$.
	\end{definition}
	It follows that the ensemble \(\mathcal{E}\) is an \(\epsilon\)-approximate \(k\)-design if and only if it is a \(\lambda\)-FO-approximate \(k\)-design with \(\lambda = \epsilon D_k\). 
	
	Intuitively, states in a state design exhibit high entanglement, a fact that is rigorously established in \cite{Liu2018EntanglementDesign}. To show this, we consider the generalized R\'enyi entropy.  In the following, all logarithms are taken to base 2.
	
	\begin{definition}
		The quantum R\'enyi entropy of a density matrix \(\rho\) is defined as
		\begin{equation}
			S^{(\alpha)}(\rho) \;=\; \frac{1}{1-\alpha} \log \tr\bigl[\rho^{\alpha}\bigr].
		\end{equation}
	\end{definition}
	
	Taking the limit as \(\alpha \rightarrow 1\), the R\'enyi entropy converges to the von Neumann entropy:
	\begin{equation}
		S(\rho) \coloneqq \lim_{\alpha \rightarrow 1} S^{(\alpha)}(\rho) = -\tr[\rho \log \rho].
	\end{equation}
	By Jensen's inequality,
	\begin{equation}
		\sum_{i=1}^D p_i \log p_i \le \log \sum_{i=1}^D p_i^2,
	\end{equation}
	it follows that $S(\rho) \ge S^{(2)}(\rho)$.
	
	Here, we fix the system to be $LR$, where both $L$ and $R$ consist of $m$ qubits, with a total of $2m$ qubits. We now present the following result.
	
	\begin{lemma}[Purity in a state design] \label{lem:design_entanglement}
		Let \(\mathcal{E}\) be an \(\epsilon\)-approximate \(2\)-design on system \(LR\). For sufficiently large \(m\), it holds that
		\begin{equation}
			\bE_{\psi \sim \mathcal{E}} \tr(\psi_L^2) \le \frac{d_A+d_B}{d_Ad_B+1} + \epsilon\le 2^{-m+1}+\epsilon. 
		\end{equation}
	\end{lemma}
	\begin{proof}
		The ensemble \(\mathcal{E}\) is also a \(\lambda\)-FO-approximate \(2\)-design with \(\lambda = \epsilon D_2\). From Theorem 26 in \cite{Liu2018EntanglementDesign},  we have
		\begin{equation}
			\bE_{\psi \sim \mathcal{E}} \tr(\psi_L^2) \le \bE_{\psi \sim \mathrm{Haar}(2m)} \tr(\psi_L^2)+\frac{\lambda}{D_2}= \frac{d_A+d_B}{d_Ad_B+1}  + \epsilon\le2^{-m+1}+\epsilon,
		\end{equation}
		where we use $\bE_{\psi \sim \mathrm{Haar}(2m)} \tr(\psi_L^2)= \frac{d_A+d_B}{d_Ad_B+1}$~\cite[Eq.~(4.11)]{Liu2018EntanglementDesign}.
	\end{proof}
	
	\subsubsection{Proof of Proposition \ref{prop:gap_deep_thermalize}: deep-thermalized states}
	\begin{proof}
		Since $E_{L\mid R}(\phi) = S(\phi_L) \ge S^{(2)}(\phi_L)$, we can bound $\Pr_{\phi \sim \mathcal{E}_{\psi}}\big[E_{L\mid R}(\phi) \ge c|L| + 1\big]$ as follows:
		\begin{equation}
			\Pr_{\phi\sim\mathcal{E}_{\psi}}\Bigl[E_{LR}\left(\phi\right) \ge c |L| + 1\Bigr]\ge\Pr_{\phi\sim\mathcal{E}_{\psi}}\Bigl[S^{(2)}(\phi_L) \ge c |L| + 1\Bigr]=\Pr_{\phi\sim\mathcal{E}_{\psi}}\Bigl[-\log\tr(\phi_L^2) \ge c |L| + 1 \Bigr]=\Pr_{\phi\sim\mathcal{E}_{\psi}}\Bigl[\tr(\phi_L^2) \le 2^{-(c |L| + 1)}\Bigr].
		\end{equation}
		By Markov's inequality,
		\begin{equation}
			\Pr_{\phi\sim\mathcal{E}_{\psi}}\Bigl[\tr(\phi_L^2) > 2^{-(c |L| + 1)}\Bigr]\le\frac{\bE_{\phi\sim\mathcal{E}_{\psi}}\Bigl[\tr(\phi_L^2)\Bigr]}{2^{-(c |L| + 1)}}=2^{c |L| + 1}\bE_{\phi\sim\mathcal{E}_{\psi}}\Bigl[\tr(\phi_L^2)\Bigr].
		\end{equation}
		Hence,
		\begin{equation}\label{eq:bound_entanglement_from_purity}
			\begin{split}
				\Pr_{\phi\sim\mathcal{E}_{\psi}}\Bigl[E_{L\mid R}\left(\phi\right) \ge c |L| + 1\Bigr]
				&\ge\Pr_{\phi\sim\mathcal{E}_{\psi}}\Bigl[\tr(\phi_L^2) \le 2^{-(c |L| + 1)}\Bigr]\\
				&=1-\Pr_{\phi\sim\mathcal{E}_{\psi}}\Bigl[\tr(\phi_L^2) > 2^{-(c |L| + 1)}\Bigr]\\
				&\ge 1-2^{c |L| + 1}\bE_{\phi\sim\mathcal{E}_{\psi}}\Bigl[\tr(\phi_L^2)\Bigr].
			\end{split}
		\end{equation}
		Lemma~\ref{lem:design_entanglement} shows that
		\begin{equation}
			\bE_{\phi \sim \mathcal{E}_{\psi}}\Bigl[\tr(\phi_L^2)\Bigr]\le 2^{-|L|+1}+\epsilon = \cO(2^{-|L|}).
		\end{equation}
		Choosing, for example, $\eta = 5$ and $c = \tfrac{9}{10}$, we have
		\begin{equation}
			\Pr_{\phi\sim\mathcal{E}_{\psi}}\Bigl[E_{L\mid R}\left(\phi\right) \ge  c |L| + 1\Bigr]\ge 1-2^{c|L| + 1}\Bigl[2^{-(|L| + 1)} + \epsilon \Bigr] = 1 - \cO\Bigl(2^{-(1-c)|L|}\Bigr) = \Omega(1).
		\end{equation}
		Therefore, the sample complexity in Theorem~\ref{thm:2Dcomplexity} is $\exp(\mathcal{O}(d^2))$.
		
		Moreover, for example, when choosing $\eta=7$ and $c=\frac9{10}$, we still have $\Pr_{\phi\sim\mathcal{E}_{\psi}}\Bigl[E_{L\mid R}\left(\phi\right) \ge  c |L| + 1\Bigr]\ge 1-o(1)$. In this case, the gap given by \eqref{eq:gap_adaptive_complexity} is $\Omega(1)$ and hence the observable $\tilde{O}(\psi,t,p')$ can be used to certify the measurement-assisted circuit complexity with $\cO(1)$ sample complexity.
	\end{proof}
	
	\subsubsection{Haar-random states}
	For Haar-random states, we use the following result on deep thermalization. 
	\begin{fact}[Haar-random states deeply thermalize, {\cite[Theorem 1]{Cotler2023EmergentDesign}}]\label{fact:Haar_random_deep_thermalization}
		Let \(\ket{\psi}\) be chosen Haar-randomly from the set of \(n\)-qubit pure states. Then, the ensemble \(\mathcal{E}_{\psi}\) forms an \(\epsilon\)-approximate \(k\)-design with probability at least \(1-\delta\), provided that 
		\begin{equation}
			|B| = \Omega\left(k |A| + \log\left(\frac{1}{\epsilon}\right) + \log \log \left(\frac{1}{\delta}\right)\right).
		\end{equation}
	\end{fact}
	
	For complexity certification task, by choosing $k = 2$ and $\epsilon = \Theta(2^{-n_A/2})$, we see that $n$-qubit Haar-random states with $n = \Omega(|A|) = \Omega(d^2)$ satisfy the requirement in Proposition~\ref{prop:gap_deep_thermalize} with probability at least $1-\delta$, and thus their circuit complexity can be sample-efficiently certified.
	
	\subsection{Brickwork-circuit states}
	We now establish the performance guarantee for brickwork-circuit states. 
	In an $D$-dimensional brickwork circuit, qubits are labeled by $D$-dimensional coordinate vectors in $[L]^D$, where $L$ is the edge length. 
	Within every $2D$ consecutive layers, the $(2i - 1)$-th and $2i$-th layers, for $1 \le i \le D$, each act on adjacent qubits that differ only in their $i$-th coordinate. 
	Specifically, in the $i$-th coordinate, these two layers respectively couple qubits with coordinates $(2x_i,\, 2x_i + 1)$ and $(2x_i - 1,\, 2x_i)$.
	
	We first show that the projected states of shallow-depth brickwork-circuit states are highly entangled. 
	Building on this, we then derive the corresponding performance guarantee for such states.
	
	\subsubsection{Entanglement in projected states}
	Here, we establish the entanglement growth result for general finite dimension $D$.  
	We begin by using the result on entanglement growth for a single-qubit random unitary.
	
	\begin{fact}[Entanglement growth under random unitary, {\cite[Appendix B]{Nahum2017EntanglementGrowth}}]\label{fact:ent_growth}
		Let $\ket{\phi_t}_{1234}$ be a pure state, and define $\ket{\phi_{t+1}}_{1234} = U_{23} \ket{\phi_t}_{1234}$, where $U_{23}$ is a two-qubit unitary sampled from the Haar measure. Assume subsystems $2$ and $3$ each have dimension $2$. Then, over the Haar-random choice of $U_{23}$,
		\begin{equation}
			\bE_{U_{23}} \tr(\rho_{t+1,12}^2) 
			= \frac{2}{5} \bigl( \tr(\rho_{t,1}^2) + \tr(\rho_{t,123}^2) \bigr).
		\end{equation}
	\end{fact}
	
	Then, we establish the entanglement growth of the projected ensemble for random brickwork-circuit states. 
	We denote $\mathcal{E}$ as the random state ensemble $\{p(\psi), \psi\}$. 
	The projected ensemble is defined as $\mathcal{E}_{A} \coloneqq \{p(\psi)\,p_{\psi}(\bm{z}), \psi_{\bm{z}}\}$, where $A$ is a region of size $(2w)^D$ in the lattice (for the 2D case, see Fig.~\ref{fig:projected_ensemble_complexity}). 
	We have the following result:
	
	\begin{lemma}[High entanglement in $D$-dimensional random brickwork-circuit states]\label{lem:gap_random_highd_circuit}
		Let $\mathcal{E}$ be the ensemble of output states from a random $D$-dimensional depth-$d'$ circuit. 
		Let $s$ be an integer with $w \ge 4s$. 
		Then,
		\begin{equation}
			\bE_{\psi \sim \mathcal{E}_A}[\tr(\psi_L^2)]\le(4/5)^{(2w-4s)^{D-1}\min(s,\lfloor d'/2D\rfloor)}.
		\end{equation}
	\end{lemma}

	\begin{proof}
		Recall the definition of the backward light cone of the subsystem $B$:
		\begin{equation}
			\begin{split}
				\cL_{B,d'+1} &= B, \\
				\cL_{B,i} &= \cL_{B,i+1}\cup \left[\bigcup \left\{\mathrm{supp}(V_{i,j}) \;\big|\; \mathrm{supp}(V_{i,j}) \cap \cL_{B,i+1} \neq \emptyset\right\}\right].
			\end{split}
		\end{equation}
		
		The whole $D$-dimensional circuit unitary $U$ can be partitioned into two parts $V W$, where $V$ is defined as the product of all two-qubit gates in the last $2D s'$ layers (with $s' = \min(s, \lfloor d'/2D \rfloor)$) whose supports lie entirely outside the backward light cone of $B$:
		\begin{equation}
			V = \prod_{i=d'-2Ds'+1}^{d'} \left(\bigotimes_{\mathrm{supp}(V_{i,j})\nsubseteq \cL_{B,i}} V_{i,j}\right),
		\end{equation}
		and $W$ is defined as the product of the remaining unitaries. Similar to Eq.~\eqref{eq:projected_state_low_complexity}, one can show that, given $W$ and measurement outcome $\bm{z}$, there exists a state $\ket{\phi_{\bm{z}}}$ such that
		\begin{equation}
			\ket{\psi_{\bm z}}=V\ket{\phi_{\bm z}}.
		\end{equation}
		
		We now show that, when $W$ and $\bm{z}$ are fixed, under the randomness of $V$,which consists of Haar-random two-qubit unitaries,
		\begin{equation}\label{eq:pf-highd}
			\bE_V[\tr(\psi_L^2)]\le(4/5)^{(2w-4s)^{D-1}\min(s,\lfloor d'/2D\rfloor)},
		\end{equation}
		where $\ket{\psi}$ denotes $\ket{\psi_{\bm{z}}}$ with a slight abuse of notation. The idea is to iteratively expand the term using Fact~\ref{fact:ent_growth}.
		We write $V$ as the product of unitaries $V = V_m \cdots V_1$, ordered by layer from $V_1$ (in layer $d' - 2Ds' + 1$) to $V_m$ (in layer $d'$), where each $V_i$ is drawn independently from the Haar measure. Let $\ket{\psi^{(i)}} = V_i \cdots V_1 \ket{\phi_{\bm{z}}}$, so that
		\begin{equation}
			\ket{\psi^{(0)}} = \ket{\phi_{\bm{z}}}, 
			\quad \ket{\psi^{(i)}} = V_i \ket{\psi^{(i-1)}}, 
			\quad \ket{\psi} = V_m \cdots V_1 \ket{\phi_{\bm{z}}} = \ket{\psi^{(m)}}.
		\end{equation}
		
		We start with $\bE_V[\tr(\psi_L^2)]=\bE_{V_m,\cdots,V_1}\left[\tr({\psi^{(m)}_L}^2)\right]$.
		At each step, we pick a term of the form $c\bE_{V_i,\cdots,V_1}\left[\tr({\psi^{(i)}_X}^2)\right]$ for some coefficient $c > 0$, subset $X$, and $i > 0$, and “decrease $i$” by integrating out $V_i$ using Fact~\ref{fact:ent_growth}. This process acts as follows (illustrated in Fig.~\ref{fig:pf-rand-D-circuit-expand}):
		\begin{itemize}
			\item If $V_i$ acts on two qubits both inside $X$ or both outside $X$, then
			\begin{equation}
				c\bE_{V_i,\cdots,V_1}\left[\tr({\psi^{(i)}_X}^2)\right]=c\bE_{V_i}\bE_{V_{i-1},\cdots,V_1}\left[\tr(\left(V_i\psi^{(i-1)}V_i^\dag\right)_X^2)\right]=c\bE_{V_{i-1},\cdots,V_1}\left[\tr({\psi^{(i-1)}_X}^2)\right],
			\end{equation}
			and we replace the term accordingly.
			\item Otherwise, $V_i$ acts on one qubit $x \in X$ and one qubit $x' \notin X$. By Fact~\ref{fact:ent_growth}, we replace this term by
			\begin{equation}
				\frac25c\left(\bE_{V_{i-1},\cdots,V_1}\left[\tr({\psi^{(i-1)}_{X\cup\{x'\}}}^2)\right]+\bE_{V_{i'-1},\cdots,V_1}\left[\tr({\psi^{(i-1)}_{X\backslash\{x\}}}^2)\right]\right).
			\end{equation}
			where the two resulting terms still have the desired form, but the sum of coefficients becomes $\frac{4}{5}c$. In this case, we say “the coefficient decreases.”
		\end{itemize}
		
		\begin{figure}
			\centering
			\subfigure[When $V_i$ acts on two qubits both inside $X$ or both outside $X$:  $c\bE_{V_i,\cdots,V_1}\left\lbrack\tr({\psi^{(i)}_X}^2)\right\rbrack=c\bE_{V_{i-1},\cdots,V_1}\left\lbrack\tr({\psi^{(i-1)}_X}^2)\right\rbrack$]{\includegraphics[width=0.396\textwidth]{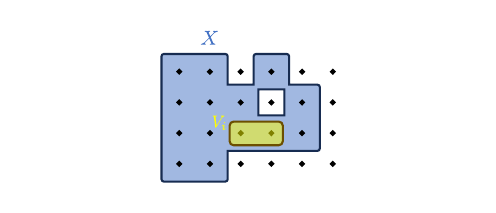}}
			\hfill
			\subfigure[When $V_i$ acts on one qubit in $X$ and the other outside $X$: $c\bE_{V_i,\cdots,V_1}\left\lbrack\tr({\psi^{(i)}_X}^2)\right\rbrack=\frac25c\left(\bE_{V_{i-1},\cdots,V_1}\left\lbrack\tr({\psi^{(i-1)}_{X_1}}^2)\right\rbrack+\bE_{V_{i-1},\cdots,V_1}\left\lbrack\tr({\psi^{(i-1)}_{X_2}}^2)\right\rbrack\right)$]{\includegraphics[width=0.594\textwidth]{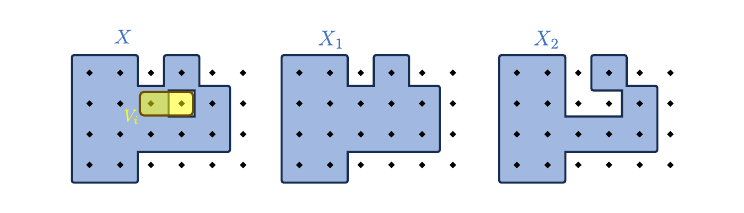}}
			\caption{Illustration of ``expanding terms'' in the proof of Lemma~\ref{lem:gap_random_highd_circuit}.}
			\label{fig:pf-rand-D-circuit-expand}
		\end{figure}
		
		Now, we expand $\bE_V[\tr(\psi_L^2)]=\bE_{V_m,\cdots,V_1}\left[\tr({\psi^{(m)}_L}^2)\right]$ and decrease the coefficient round by round. In each round, we enumerate every current term and reduce its index $i$ until its coefficient decreases.
		If this process continues for $r$ rounds before reaching $i = 0$, then the sum of the coefficients becomes $(\frac{4}{5})^r$.
		Moreover, since $\tr({\psi^{(i)}_X}^2)\le1$, the averaged purity is bounded by $(\frac{4}{5})^r$. 
		We will show that there are at least $(2w - 4s)^{D-1} s'$ rounds.
		
		Let $i_l$ ($0 \le l \le 2D s'$) be the index such that $V_{i_l+1}, \dots, V_m$ are the unitaries in the last $l$ layers, i.e., layers $d' - l + 1$ to $d'$. By definition, we have $i_0 = m$ and $i_{2D s'} = 0$. The unitaries can be grouped by their layer as
		\begin{equation}
			V=V_m\cdots V_1=(V_{i_0}V_{i_0-1}\cdots V_{i_1+1})(V_{i_1}V_{i_1-1}\cdots V_{i_2+1})\cdots(V_{i_{2Ds'-1}}V_{i_{2Ds'-1}-1}\cdots V_{i_{2Ds'}+1}),
		\end{equation}
		where $V_{i_l}, V_{i_l - 1}, \dots, V_{i_{l+1} + 1}$ are the unitaries in layer $d' - l$. 
		We have the following claim:
		
		\begin{claim}\label{claim:rand-d-complexity-X-range}
			When $i$ is decreased to $i_{2D g}$ for $0 \le g < s'$, the partition set $X$ of each term satisfies 
			$ L_{-2g} \subseteq X \subseteq L_{+2g}$,
			where $L_{\pm t}$ is defined in Fig.~\ref{fig:pf-rand-D-circuit-defs}.
		\end{claim}
		
		\begin{figure}
			\centering
			\includegraphics[width=0.8\textwidth]{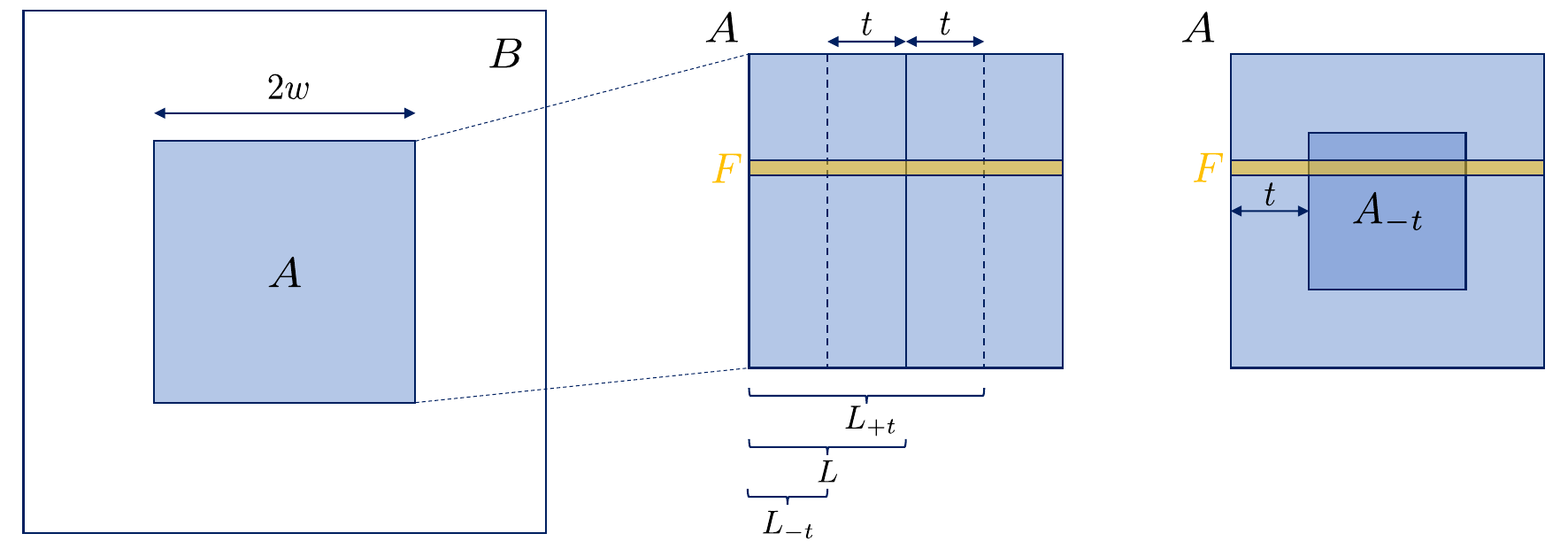}
			\caption{Definitions used in Claims~\ref{claim:rand-d-complexity-X-range} and~\ref{claim:rand-d-complexity-ent-increase}. 
				Here, $L_{\pm t}$ denotes the set of coordinates in $A$ whose first coordinate lies in $[1,\, w \pm t]$. 
				A fiber $F$ is defined as a one-dimensional line of qubits in $A$ that differ only in their first coordinate. 
				$A_{-t}$ denotes the set of coordinates at least distance $t$ from the boundary of $A$.
			}
			\label{fig:pf-rand-D-circuit-defs}
		\end{figure}
		
		\begin{proof}
			We proceed by induction. The claim holds trivially for $g = 0$. Consider a term with $i = i_l$ and $L_{-t} \subseteq X \subseteq L_{+t}$, and suppose $i$ is then decreased to $i_{l+1}$, meaning we examine all the unitaries in layer $d' - l + 1$. Let $X'$ be the partition set of any term after decreasing $i$. There exists $j \in [D]$ such that every gate in the $(d' - l + 1)$-th layer is aligned along the $j$-th coordinate; that is, each gate acts on qubits $x, x'$ satisfying $x'_j = x_j + 1$ and $x'_{j'} = x_{j'}$ for $j' \ne j$. 
			By the ``decrease $i$'' rule, if $j \ne 1$, then $L_{-t} \subseteq X' \subseteq L_{+t}$; otherwise $L_{-(t+1)} \subseteq X' \subseteq L_{+(t+1)}$.
			
			From the circuit architecture, in every $2D$ consecutive layers, there are exactly $2$ layers where all unitaries are aligned with the first coordinate. Therefore, after each such occurrence, $t$ increases by at most $2$. This proves the claim by induction.
		\end{proof}
		
		By the construction of $D$-dimensional brickwork circuits,
		\begin{equation}
			\mathcal{L}_{B,d'-i} \cap A_{-2s'} = \varnothing,
		\end{equation}
		for all $0 \le i \le 2D s'$, where $A_{-t}$ is defined in Fig.~\ref{fig:pf-rand-D-circuit-defs}. Hence, by the definition of $V$, for all adjacent qubits $x, x' \in A_{-2s'}$ and $0 \le g < s'$, there exists some $V_i$ acting on $x, x'$ in a layer from $d' - 2Dg + 1$ to $d'$, i.e., $i_{2Dg} \ge i > i_{2D(g+1)}$.
		We now show that, when $i$ of a term is decreased to $i_{2Dg}$ for $0 \le g < s'$, before $i$ is decreased to $i_{2D(g+1)}$, the coefficient can be decreased at least $(2w - 4s)^{D-1}$ times. It suffices to focus on one term with partition set $X$. To establish this, we define a fiber $F$ as a one-dimensional line of qubits in $A$ that differ only in their first coordinate (Fig.~\ref{fig:pf-rand-D-circuit-defs}). We have the following claim:
		
		\begin{claim}\label{claim:rand-d-complexity-ent-increase}
			For any fiber $F$ intersecting $A_{-2s'}$, for any term with set $X$, $X \cap F$ changes at least once before $i$ decreases from $i_{2Dg}$ to $i_{2D(g+1)}$.
		\end{claim}
		
		\begin{proof}
			By Claim~\ref{claim:rand-d-complexity-X-range}, $L_{-2g} \subseteq X \subseteq L_{+2g}$. This implies the existence of adjacent $x, x' \in \left(L_{+(2g+1)} \backslash L_{-(2g+1)}\right) \cap F$ such that $x \in X$ and $x' \notin X$. Since $(2g+1) < 2s' \le w - 2s'$, both $x, x'$ belong to $A_{-2s'}$, and thus there exists some $V_{i'}$ acting on $x, x'$ with $i_{2Dg} \ge i' > i_{2D(g+1)}$. If $X \cap F$ has never changed before $i$ decreases to $i'$, then for the current term with partition set $X$, we have $x \in X$ and $x' \notin X$. At this point, when $i$ is decreased, $X \cap F$ changes. This proves the claim.
		\end{proof}
		
		There are at least $(2w - 4s')^{D-1}$ such fibers. Moreover, whenever the coefficient decreases, exactly one qubit is either added to $X$ or removed from $X$. Therefore, the coefficient decreases at least $(2w - 4s)^{D-1}$ times when $i$ of a term is decreased to $i_{2Dg}$ for $0 \le g < s'$.
		In total, we can decrease the coefficients for at least $(2w - 4s)^{D-1} s'$ rounds. From the earlier discussion, this yields the bound in Eq.~\eqref{eq:pf-highd}. Including the randomness over $W$ and $\bm{z}$ completes the proof of Lemma~\ref{lem:gap_random_highd_circuit}.
	\end{proof}
	
	\subsubsection{Proof of Proposition \ref{prop:brickwork_state}}
	
	\begin{proof}
		Let
		\begin{equation}
			p_{\psi} \coloneqq \Pr_{\bm z \sim \mathcal{E}_{\psi}}\left[ E_{L \mid R}(\psi_{\bm z}) \ge c|L| + 1 \right],
		\end{equation}
		We first compute the expectation of $p_{\psi}$ over $\psi \sim \mathcal{E}$:
		\begin{equation}
			\begin{split}
				\bE_{\psi\sim\mathcal{E}}[p_{\psi}]& =\bE_{\psi\sim\mathcal{E}}\Bigl[\Pr_{\psi_{\bmz}\sim \mathcal{E}_{\psi}}\bigl[E_{L\mid R}\left(\psi_{\bmz}\right) \ge c |L| + 1\bigr] \Bigr] \\
				& =\Pr_{\psi\sim\mathcal{E},\psi_{\bmz}\sim \mathcal{E}_{\psi}}\Bigl[E_{L\mid R}\left(\psi_{\bmz}\right) \ge c |L| + 1\Bigr]\\
				&=\Pr_{\psi\sim\mathcal{E}_A}\Bigl[E_{L\mid R}\left(\psi\right) \ge c |L| + 1\Bigr] \\
				&\ge 1-2^{c |L| + 1}\bE_{\psi\sim\mathcal{E}_{A}}\Bigl[\tr(\psi_L^2)\Bigr]\\
				&\ge 1 - 2^{c|L|+1} (4/5)^{(2w-4s)\min(s,\lfloor d'/4\rfloor)}.
			\end{split}
		\end{equation}
		where the fourth line uses Eq.~\eqref{eq:bound_entanglement_from_purity} and the last inequality applies Lemma~\ref{lem:gap_random_highd_circuit}.
		
		For $s=\frac14 w$, $d'\ge w=\eta d$ (where $\eta$ is a constant to be determined later), a sufficiently small constant $c$, and a large $w$,  
		\begin{equation}
			\bE_{\psi\sim\mathcal{E}}[p_{\psi}] \ge 1 - 2^{-\Omega(w^2)}.
		\end{equation}
		Then,
		\begin{equation}
			\Pr_{\psi\sim\mathcal{E}}\left[p_{\psi} \ge 0.5\right] \ge 1 - 2 \cdot 2^{-\Omega(w^2)}.
		\end{equation}
		By Theorem~\ref{thm:2Dcomplexity}, when $p_{\psi} > 0.5$, we can certify circuit complexity $d$ with sample complexity $\exp(\mathcal{O}(w^2)) = \exp(\mathcal{O}(d^2))$, as $\eta$ is also a constant.  
		By choosing $w$ to be a large constant such that $1 - 2 \cdot 2^{-\Omega(w^2)} > 0.99$, we obtain the desired result.
		
		Moreover, whenever the constant $\eta\ge14c^{-1}$, when $p_{\psi}>0.5$, $\eta p_{\psi}c\ge 7$. In this case, the gap given by \eqref{eq:gap_adaptive_complexity} is $\Omega(1)$ and hence the observable $\tilde{O}(\psi,t,p)$ can be used to certify the measurement-assisted circuit complexity with $\cO(1)$ sample complexity.
	\end{proof}
	
	\section{Additional analysis on entanglement certification}\label{app:proof_entanglement_certification}
	
	
	\subsection{Proof of Proposition \ref{prop:fully_inseparable_entanglement}}
	\begin{proof}
		\textit{Soundness}---If $\rho$ is not fully inseparable, let $A \mid B$ be a separable bipartition. Then there exists some $i\in A$ such that $i+1 \in B$. Then  By Proposition~\ref{prop:bipartite_entanglement} (random-basis variant), when testing $A_1=\{i\}$ and $B_1=\{i+1\}$, the protocol outputs $\reject$ with probability at least $1-\delta$. Therefore, among all tests, at least one test rejects with probability at least $1-\delta$.
		
		\textit{Completeness}---If the input state $\rho$ satisfies  $\dtr(\rho,\psi)\leq\varepsilon := \min_i\left\{\frac{\widetilde{\mathrm{LE}}_{\psi}(i,i+1)}{6}\right\} = \Omega(1)$ (see Theorem~\ref{thm:performance_certification_protocol} and Eq.~\eqref{eq:fully_inseparable_requirement}), set the failure probability parameter to $\delta'=\delta/(n-1)$. By Proposition~\ref{prop:bipartite_entanglement} (random-basis variant), with
		$T=\cO\bigl(\log (\delta'^{-1})\bigr)=\cO\bigl(\log n+\log \delta^{-1}\bigr)$, each individual test for a fixed $1\le i\le n$ outputs $\accept$ with probability at least $1-\delta'$. A union bound then implies that all tests (for $2\le i\le n$) output $\accept$ simultaneously with probability at least $1-\delta$.
	\end{proof}
	
	\subsection{Localizable entanglement of Haar-random states}
	Here, we prove that the localizable entanglement of Haar-random states is bounded below by a constant with overwhelmingly high probability, thereby validating the performance of our protocol on generic states. 
	
	\begin{lemma}[\revise{Localizable entanglement of Haar-random states}]\label{lem:Haar-random-states-localizable-entanglement}
		Let $\ket{\psi}$ be chosen Haar-randomly from the set of $n$-qubit pure states. 
		Then, with probability at least $1-\delta$ for $\delta = \exp(-\exp(\Omega(n)))$, we have
		\begin{equation}
			\mathrm{LE}_{\psi}(1,2) = \Omega(1), 
			\qquad 
			\widetilde{\mathrm{LE}}_{\psi}(1,2) = \Omega(1).
		\end{equation}
	\end{lemma}
	
	\begin{proof}
		Fix $\epsilon = \tfrac{1}{10}$. By Fact~\ref{fact:Haar_random_deep_thermalization}, the projected ensemble $\mathcal{E}_{\psi}$ on qubits $\{1,2\}$ is an $\epsilon$-approximate 2-design with probability $1-\delta_0$, where $\delta_0 = \exp\bigl(-\exp(\Omega(n))\bigr)$. 
		Then, by Lemma~\ref{lem:design_entanglement}, 
		\begin{equation}
			\bE_{\phi \sim \mathcal{E}_{\psi}} \bigl[\tr(\phi_1^2)\bigr] \le \frac{4}{5} + \epsilon = \frac{9}{10}.
		\end{equation}
		Now we show for the constant $e_0 = \frac{1}{20}$, it holds that $\mathrm{LE}_{\psi}(1,2) \ge e_0$. 
		Denote by $\lambda_1, \lambda_2$ the Schmidt coefficients of the projected state $\phi$ with $\lambda_1 \ge \lambda_2$. We have $\lambda_2 = 1-\lambda_1$, $\tr(\phi_1^2) = \lambda_1^2 + \lambda_2^2$ and $\fid(\phi) = \lambda_1$. 
		Therefore,
		\begin{equation}
			\begin{split}
				\mathrm{LE}_{\psi}(1,2) &= \bE_{\phi \sim \mathcal{E}_{\psi}}[1 - \lambda_1] \\
				&\ge \bE_{\phi \sim \mathcal{E}_{\psi}}[\lambda_1(1-\lambda_1)] \\
				&= \tfrac{1}{2}\,\bE_{\phi \sim \mathcal{E}_{\psi}}\left[2\lambda_1(1-\lambda_1) - 1\right] + \tfrac{1}{2} \\
				&= -\tfrac{1}{2}\,\bE_{\phi \sim \mathcal{E}_{\psi}}[\tr(\phi_1^2)] + \tfrac{1}{2} \\
				&\ge \tfrac{1}{20}.
			\end{split}
		\end{equation}
		Hence $\mathrm{LE}_{\psi}(1,2) \ge e_0$ with probability at least $1-\delta_0$.
		
		To lower bound $\widetilde{\mathrm{LE}}_{\psi}(1,2)$, define $\mathrm{LE}_{\psi}^{\bmb}$ for $\bmb \in \{X,Y,Z\}^{n_B}$ as the localizable entanglement under measurement basis $\bmb$ on $B$. 
		By local unitary invariance of the Haar measure, $\mathrm{LE}_{\psi}^{\bmb}(1,2) \ge e_0$ for Haar-random $\psi$ with probability at least $1-\delta_0$. 
		Let $f(\psi,\bmb) \coloneqq \mathbf{1}[\mathrm{LE}_{\psi}^{\bmb}(1,2) \ge e_0]$. Then for $\bmb$ chosen uniformly from $\{X,Y,Z\}^{n_B}$,
		\begin{equation}
			\bE_{\psi}\bE_{\bmb} f(\psi,\bmb) = \bE_{\bmb}\bE_{\psi} f(\psi,\bmb) \ge 1-\delta_0.
		\end{equation}
		By Markov’s inequality,
		\begin{equation}
			\Pr_{\psi}\left[\bE_{\bmb} f(\psi,\bmb) \le 1-\sqrt{\delta_0}\right] \le \sqrt{\delta_0}.
		\end{equation}
		Therefore, with probability at least $1 - \delta_1$ for $\delta_1 \coloneqq \sqrt{\delta_0} = \exp\bigl(-\exp(\Omega(n))\bigr)$,
		\begin{equation}
			\begin{split}
				\bE_{\bmb} f(\psi,\bmb) &> 1-\delta_1,\\
				\widetilde{\mathrm{LE}}_{\psi}(1,2) &= \bE_{\bmb}\,\mathrm{LE}_{\psi}^{\bmb}(1,2) \\
				&> (1-\delta_1)e_0 \\
				&= \Omega(1).
			\end{split}
		\end{equation}
		
		Combining the bounds, we conclude that with probability at least $1-\delta$ for $\delta = \delta_0 + \delta_1 = \exp\bigl(-\exp(\Omega(n))\bigr)$,
		\begin{equation}
			\mathrm{LE}_{\psi}(1,2) = \Omega(1), 
			\quad 
			\widetilde{\mathrm{LE}}_{\psi}(1,2) = \Omega(1).
		\end{equation}
	\end{proof}
	
	By a union bound, with probability at least $1-\delta'$ for $\delta' = n\delta = \exp\bigl(-\exp(\Omega(n))\bigr)$, we have $\widetilde{\mathrm{LE}}_{\psi}(i,i+1) = \Omega(1)$ for all $i \in [n-1]$. 
	This establishes that our certification protocol for fully inseparable entanglement is sample-efficient for Haar-random states. 
	
	\subsection{Details of the numerical experiments on localizable entanglement}
	\subsubsection{Stabilizer states}
	In Fig.~\ref{fig:StabilizerEntanglement}, we sampled $N_s = 5$ random stabilizer states for system sizes $n \in \{32, 64, 128, 256\}$, and $N_s = 10$ for $n=16$.
	For each stabilizer state, we set the failure probability to $\delta = 0.01$ and the additive error to $\varepsilon = 0.05$.
	To guarantee that all sites are simultaneously estimated up to an additive error $\varepsilon$ with a global failure probability of at most $\delta$, we applied a union bound by setting $\delta' = \frac{\delta}{n N_s}$. Accordingly, we set the parameters for the median-of-means estimator (Proposition~\ref{prop:mom-estimator}) to $K = \lceil \tfrac{9}{2} \ln (\delta'^{-1}) \rceil \approx 60$ and $B = \lceil \frac{6 \sigma^2}{\varepsilon^2} \rceil = 600$, assuming a conservative variance upper bound of $\sigma^2 = 0.25$ of $\fid$ across different projected states.
	
	A total of $T = BK$ samples $\{\bmx_i\}_{i=1}^T$ were obtained by performing random Pauli measurements on the sampled stabilizer states. Each measurement sample $\bmx_i \in \mathsf{X}^n$ is reused across all sites $j$ to construct the projected state on the nearest-neighbor pair $(i, i+1)$.
	Subsequently, the localizable entanglement $\widetilde{\mathrm{LE}}_{\psi}(i,i+1)$ is obtained by computing the median-of-means estimator of $\fid$ of the obtained projected states on $(i, i+1)$.
	Note that $\sigma^2=0.25$ overestimates the actual variance of $\fid$. Consequently, the actual additive error is much smaller than the target $\varepsilon$.

	\subsubsection{Hamiltonian systems}
	In Fig.~\ref{fig:HamiltonianEntanglement}, we report sample means and plot error bars as the standard error of the mean, computed from $N=1000$ independently sampled projected states. The maximal fidelity to a two-qubit product state is evaluated as $\lambda_{1}$, the square of the maximal Schmidt coefficient of the two-qubit projected state. Ground states of the XXZ chain are obtained by DMRG for system sizes $n=32$ and $n=64$, while the $J_{1}$–$J_{2}$ chain is treated by DMRG at $n=20$. All other data sets shown in the figures are produced by exact diagonalization.

	\revise{
		\section{Additional analysis on magic certification}
		\label{app:proof_magic_certification}
		
		In this appendix, we prove that when the subsystem $A$ consists of a single qubit, the localizable magic of a Haar-random state is lower-bounded by a constant with overwhelmingly high probability.
		
		\begin{lemma}[Localizable magic of Haar-random states]
			\label{lem:Haar-random-states-localizable-magic}
			Let $A$ be a single-qubit subsystem, and $B$ be an $n_B$-qubit subsystem. If $\ket{\psi} \in \mathcal{H}_A \otimes \mathcal{H}_B$ is a Haar-random pure state, then with probability at least $1 - \exp(-\Omega(2^{n_B}))$, 
			\begin{equation}
				\min\left\{\mathrm{LM}(\psi), \widetilde{\mathrm{LM}}(\psi)\right\} \ge \frac{1}{24}.
				\label{eq:Haar-localizable-magic}
			\end{equation}
		\end{lemma}
		
		\begin{proof}
			Let $\mathcal{S}_1=\{\ket{0},\ket{1},\ket{+},\ket{-},\ket{+i},\ket{-i}\}$ be the set of single-qubit stabilizer states. For a pure single-qubit state $\ket{\phi}$, we define
			\begin{equation}
				m(\phi)
				\coloneqq 1-\max_{\ket{s}\in\mathcal{S}_1}
				\abs{\braket{s}{\phi}}^2.
				\label{eq:single-qubit-magic-gap}
			\end{equation}
			Since the free projected set $\mathrm{STAB}_1$ is the convex hull of $\mathcal{S}_1$, we have $m(\phi)=1-\fid(\phi)$.
			
			For a fixed $\ket{s}\in\mathcal{S}_1$ and a Haar-random state $\ket{\phi}\in\mathbb{C}^2$, the overlap
			$X_s\coloneqq\abs{\braket{s}{\phi}}^2$ is uniformly distributed on $[0,1]$. Thus, for any
			$0\le t\le 1/6$, a union bound gives
			\begin{equation}
				\Pr_{\phi}[m(\phi)\ge t]
				\ge 1-\sum_{\ket{s}\in\mathcal{S}_1}
				\Pr_{\phi}[X_s>1-t]
				=1-6t.
			\end{equation}
			Consequently, 
			\begin{equation}
				\mu_1\coloneqq\bE_{\phi\sim\operatorname{Haar}(n_A)}[m(\phi)]
				=\int_0^1\Pr_{\phi}[m(\phi)\ge t]\,\mathrm{d}t
				\ge\int_0^{1/6}(1-6t)\,\mathrm{d}t
				=\frac{1}{12}.
				\label{eq:single-qubit-Haar-magic-mean}
			\end{equation}
			
			Fix a product Pauli basis $\bmb\in\{X,Y,Z\}^{n_B}$ on $B$, and denote
			the corresponding localizable magic under basis $\bmb$ as $\mathrm{LM}_{\bmb}$.
			Using the Gaussian representation of Haar measure, let
			$\{g_{\bmz}\}_{\bmz}$ be independent standard complex Gaussian vectors
			in $\mathbb{C}^2$, and set
			$R_{\bmz}=\norm{g_{\bmz}}_2^2$ and
			$\ket{u_{\bmz}}=\ket{g_{\bmz}}/\norm{g_{\bmz}}_2$.  Then
			\begin{equation}
				\ket{\psi}
				=\frac{1}{\sqrt{\sum_{\bmz}R_{\bmz}}}
				\sum_{\bmz}\sqrt{R_{\bmz}}\,
				\ket{u_{\bmz}}_A\otimes\ket{\bmz}_{\bmb}.
				\label{eq:Haar-Gaussian-blocks}
			\end{equation}
			The directions $\ket{u_{\bmz}}$ are independent Haar-random states in
			$\mathcal{H}_A$ and are independent of $\{R_{\bmz}\}_{\bmz}$. Therefore,
			\begin{equation}
				\bE_{\psi}\!\left[
				\mathrm{LM}_{\bmb}(\psi)\mid\{R_{\bmz}\}_{\bmz}
				\right]
				=\sum_{\bmz}\frac{R_{\bmz}}{\sum_{\bmz'}R_{\bmz'}}
				\bE_{u_{\bmz}}m(u_{\bmz})
				=\mu_1.
				\label{eq:basis-resolved-LM-mean}
			\end{equation}
			Thus, $\bE_{\psi}[\mathrm{LM}_{\bmb}(\psi)]=\mu_1$ for every basis $\bmb$.
			In particular, 
			\begin{equation}
				\bE_{\psi}[\mathrm{LM}(\psi)]
				=\bE_{\psi}[\widetilde{\mathrm{LM}}(\psi)]
				=\mu_1,
				\qquad
				\widetilde{\mathrm{LM}}(\psi)
				=3^{-n_B}\!\sum_{\bmb\in\{X,Y,Z\}^{n_B}}
				\mathrm{LM}_{\bmb}(\psi).
				\label{eq:LM-means-and-random-basis-average}
			\end{equation}
			
			It remains to prove concentration.  For any
			$\boldsymbol{s}=(s_{\bmz})_{\bmz}\in\mathcal{S}_1^{d_B}$, where
			$d_B=2^{n_B}$, set
			\begin{equation}
				P_{\boldsymbol{s},\bmb}
				\coloneqq\sum_{\bmz}
				\ketbra{s_{\bmz}}_A\otimes\ketbra{\bmz}_{\bmb}.
			\end{equation}
			This is an orthogonal projector, and Eq.~\eqref{eq:single-qubit-magic-gap}
			gives
			\begin{equation}
				\mathrm{LM}_{\bmb}(\psi)
				=1-\max_{\boldsymbol{s}\in\mathcal{S}_1^{d_B}}
				\bra{\psi}P_{\boldsymbol{s},\bmb}\ket{\psi}.
				\label{eq:LM-projector-form}
			\end{equation}
			For any unit vectors $\ket{\psi},\ket{\varphi}$ and any orthogonal
			projector $P$,
			\begin{equation}
				\abs{\bra{\psi}P\ket{\psi}-\bra{\varphi}P\ket{\varphi}}
				\le 2\norm{\psi-\varphi}_2.
			\end{equation}
			Hence every $\mathrm{LM}_{\bmb}$ is $2$-Lipschitz on the unit sphere of
			$\mathbb{C}^{2d_B}$, and so is their average
			$\widetilde{\mathrm{LM}}$.  Then, Levy's lemma~\cite[Theorem 7.37]{watrous2018theory} implies, for a
			universal $c_0>0$ and every $\varepsilon>0$,
			\begin{equation}
				\begin{aligned}
					\Pr_{\psi}\!\left[
					\abs{\mathrm{LM}(\psi)-\mu_1}\ge\varepsilon\right]
					&\le 2\exp(-c_0d_B\varepsilon^2),\\
					\Pr_{\psi}\!\left[
					\abs{\widetilde{\mathrm{LM}}(\psi)-\mu_1}\ge\varepsilon\right]
					&\le 2\exp(-c_0d_B\varepsilon^2).
				\end{aligned}
				\label{eq:LM-concentration-compact}
			\end{equation}
			Because $\mu_1\ge1/12$, for either quantity to be smaller than $1/24$, it must deviate by at least $1/24$ from its mean. 
			Taking $\varepsilon=1/24$ in Eq.~\eqref{eq:LM-concentration-compact} and applying a union bound concludes the proof.
		\end{proof}
	}
	
	\section{Additional proofs on fidelity certification}\label{app:proof_spectral_gap}
	
	\subsection{Proof of Proposition \ref{prop:fidelity_performance}}
	
	\begin{proof}
		Set the decision boundary as $\eta^* = 1 - (1-c)\Delta(1-F)$, and the estimation error as $\varepsilon = c\Delta(1-F)$. 
		Accounting for statistical fluctuations of the estimator $\omega$ of $\tr(O_{\psi}\rho)$ (see Eq.~\eqref{eq:claim-distance-from-mean-MoM} and Corollary~\ref{col:var_conditional_fidelity}), the number of copies required to guarantee precision $\varepsilon$ with failure probability at most $\delta$ is
		\begin{equation}
			T = \frac{27 \ln(\delta^{-1})(4^{n_A}+1)}{c^{2}\Delta^{2}(1-F)^{2}}.
		\end{equation}
		
		When $\tr(\psi\rho) \leq F$, we have
		\begin{equation}
			\tr(O_{\psi}\rho)
			\leq
			F\cdot 1 + (1-F)(1-\Delta)
			=
			1-\Delta(1-F).
		\end{equation}
		Therefore, with probability at least $1-\delta$, the estimator $\omega$ satisfies
		\begin{equation}
			\omega < \tr(O_{\psi}\rho) + \varepsilon \leq \eta^{*},
		\end{equation}
		so the protocol outputs $\reject$.

		Conversely, when $\tr(\psi\rho) \ge 1-(1-2c)\Delta(1-F)$, note that the spectrum of $O_{\psi}$ is non-negative and has eigenvalue $1$ for $\ket{\psi}$. Therefore, 
		\begin{equation}
			\tr(O_{\psi}\rho) \geq \tr(\psi\rho) \geq 1-(1-2c)\Delta(1-F).
		\end{equation}
		Hence, with probability at least $1-\delta$,
		\begin{equation}    
			\omega > \tr(O_{\psi}\rho)-\varepsilon \ge\eta^{\ast}.
		\end{equation}
		and the protocol will output $\accept$.
	\end{proof}
	
	\subsection{Proof of Theorem~\ref{thm:constant-gap-graph-states}}
	
	Let $G=(V,E)$ be a simple undirected graph on $n$ vertices $V=\{1,\dots,n\}$ with symmetric adjacency matrix $A\in\bF_2^{n\times n}$ with zero diagonal.
	The associated $n$-qubit graph state $\ket{G}$ is the unique $+1$ common eigenstate of the stabilizer generators
	\begin{equation}
		K_i \;=\; X_i \prod_{j\in N(i)} Z_j,\qquad i\in[n],
	\end{equation}
	where $N(i)=\{j:\,A_{ij}=1\}$ is the neighborhood of $i$.
	Equivalently, writing an $n$-qubit Pauli as
	\begin{equation}\label{eq:stabilizer_H_XZ_form}
		P_j = i^{p_j}\prod_{k=1}^n X_k^{H^X_{k,j}} Z_k^{H^Z_{k,j}},
	\end{equation}
	the stabilizer of $\ket{G}$ is generated by $\{P_1,\dots,P_n\}$ with
	\begin{equation}
		H^X = I_n,\qquad H^Z = A,
	\end{equation}
	i.e., the $j$-th generator has an $X$ on qubit $j$ and $Z$'s on its neighbors.
	A \emph{random graph state} is obtained by sampling $A$ uniformly at random from all symmetric matrices in $\bF_2^{n\times n}$ with zero diagonal, and setting $H^Z=A$ (with $H^X=I_n$).
	
	\begin{proof}[Proof of Theorem~\ref{thm:constant-gap-graph-states}]
		Operationally, measuring $O_{\psi}$ can be viewed as follows: independently sample a local Pauli basis ($X$, $Y$, or $Z$) for each qubit in $B$, measure $B$ in the resulting product basis, and then (conditioned on the outcome $\bmx$) either output $0$ if $\tilde p_{\psi}(\bmx)=0$ or project $A$ onto $\ket{\psi_{\bmx}}$.
		For a fixed choice of local bases $\bmb\in\Bbrange$ (with $\sfB=\{\{0,1\},\{+,-\},\{+i,-i\}\}$), define
		\begin{align}
			O_{\bmb,\psi} = \sum_{\bmx \in \bmb,\tilde{p}_\psi(\bmx)\ne0} \ketbra{\psi_{\bmx}}_A \otimes \ketbra{\bmx}_B,
		\end{align}
		so that
		\begin{align}
			O_{\psi} = 3^{-n_B} \sum_{\bmb\in\Bbrange} O_{\bmb,\psi}=\bE_{\bmb}\left[O_{\bmb,\psi}\right],
		\end{align}
		where $\bmb$ is uniform over $\Bbrange$.
		
		Let $\mathcal S=\mathrm{Stab}(\psi)$ be generated by $\{P_1,\dots,P_n\}$.  
		Fix a basis choice $\bmb$, and let $\mathcal S_{\mathrm{comm},\bmb}\subseteq \mathcal S$ denote the subgroup of stabilizers that commute with all local Pauli measurements specified by $\bmb$ on $B$. Since the projectors $\{\ketbra{\bmx}_B\}_{\bmx\in\bmb}$ are pairwise orthogonal, we have $O_{\bmb,\psi}^2=O_{\bmb,\psi}$, and hence $O_{\bmb,\psi}$ is a projector.
		
		Moreover, $O_{\bmb,\psi}$ is exactly the projector onto the joint $(+1)$-eigenspace of $\mathcal S_{\mathrm{comm},\bmb}$. First, for any $g\in\mathcal S_{\mathrm{comm},\bmb}$ and any outcome $\bmx\in\bmb$ with $\tilde p_\psi(\bmx)\neq 0$, we have
		\begin{equation}
			g (I\otimes \ketbra{\bmx}_B)\ket{\psi}
			=(I\otimes \ketbra{\bmx}_B)\, g\ket{\psi}
			=(I\otimes \ketbra{\bmx}_B)\ket{\psi},
		\end{equation}
		where we used $[g,\ketbra{\bmx}_B]=0$ and $g\ket{\psi}=\ket{\psi}$. Thus each branch state $\ket{\psi_{\bmx}}_A\otimes\ket{\bmx}_B$ lies in the joint $(+1)$-eigenspace of $\mathcal S_{\mathrm{comm},\bmb}$.
		
		Conversely, let $\ket v=\sum_{\bmx\in\bmb}\widetilde{\ket{v_{\bmx}}}_A\otimes\ket{\bmx}_B$ satisfy $g\ket v=\ket v$ for all $g\in\mathcal S_{\mathrm{comm},\bmb}$. Since every $g\in\mathcal S_{\mathrm{comm},\bmb}$ commutes with $\ketbra{\bmx}_B$, each component $\widetilde{\ket{v_{\bmx}}}_A\otimes\ket{\bmx}_B$ is also stabilized by $\mathcal S_{\mathrm{comm},\bmb}$. In addition, it is stabilized by the measurement-outcome constraints $\{(-1)^{\bmx_i}P_{\bmb_i}\}_{i\in B}$, where $P_{\bmb_i}\in\{X_i,Y_i,Z_i\}$ is the Pauli observable associated with the basis $\bmb_i$. The state stabilized by $\mathcal S_{\mathrm{comm},\bmb}$ together with these commuting single-qubit Paulis is unique. Therefore,
		\begin{equation}
			\widetilde{\ket{v_{\bmx}}}_A\otimes\ket{\bmx}_B \ \propto\ (I\otimes\ketbra{\bmx}_B)\ket{\psi}
			\ \propto\ \ket{\psi_{\bmx}}_A\otimes\ket{\bmx}_B,
		\end{equation}
		and hence $\bra{v}O_{\bmb,\psi}\ket v = 1$ lies in $\mathrm{Ran}(O_{\bmb,\psi})$. Combining above, $O_{\bmb,\psi}$ is precisely the projector onto the joint $(+1)$-eigenspace of $\mathcal S_{\mathrm{comm},\bmb}$.
		
		
		For each $x\in\bF_2^n$, let $\ket{\psi_x}$ be the stabilizer state whose stabilizer group is generated by
		$\{(-1)^{x_1}P_1,\dots,(-1)^{x_n}P_n\}$.
		Then $\{\ket{\psi_x}\}_{x\in\bF_2^n}$ forms an orthonormal basis.
		Because $O_{\bmb,\psi}$ is the projector onto the subspace stabilized by $\mathcal S_{\mathrm{comm},\bmb}$, we have $\bra{\psi_x}O_{\bmb,\psi}\ket{\psi_x} \in \{0,1\}$
		and $\bra{\psi_x}O_{\bmb,\psi}\ket{\psi_x}=1$ iff every $g\in\mathcal S_{\mathrm{comm},\bmb}$ also stabilizes $\ket{\psi_x}$.
		To make this condition explicit, write each generator in $(H^X,H^Z)$ form in Eq.~\eqref{eq:stabilizer_H_XZ_form}.
		Any stabilizer element can be written as
		\begin{equation}
			g_z := P_1^{z_1}\cdots P_n^{z_n}
			\ \propto\
			\prod_{k=1}^n X_k^{(H^X z)_k} Z_k^{(H^Z z)_k},
			\qquad z\in\bF_2^n.
		\end{equation}
		Such $g_z$ stabilizes $\ket{\psi_x}$ iff $(-1)^{z^\top x}=1$, i.e., $z^\top x=0$ in $\bF_2$.
		For each $i\in B$, let 
		\begin{align}
			m_i^T=\begin{cases}e_i^TH^X,&\bmb_i=\{0,1\},\\e_i^TH^Z,&\bmb_i=\{+,-\},\\e_i^TH^X+e_i^TH^Z,&\bmb_i=\{+i,-i\},\\\end{cases}
		\end{align}
		where $e_i$ denote the unit vector having $1$ on the $i$-th entry. 
		A direct Pauli commutation check shows that $g_z$ commutes with the local measurement on qubit $i$ iff $m_i^\top z=0$ in $\bF_2$.
		Therefore, $\bra{\psi_x}O_{\bmb,\psi}\ket{\psi_x}=1$ iff ``for any $z\in\bF_2^n$, $m_i^Tz=0$ for all $i\in B$ implies $z^Tx=0$'', which is equivalent to $x\in\spn\{m_i\mid i\in B\}$.

		Since $O_{\psi}=\bE_{\bmb}[O_{\bmb,\psi}]$ and each $O_{\bmb,\psi}$ is diagonal in the basis $\{\ket{\psi_x}\}$, $O_{\psi}$ is also diagonal in this basis.
		The eigenvalue corresponding to $\ket{\psi_x}$ is
		\begin{align}
			\Pr_m[x\in\spn\{m_i\mid i\in B\}].
		\end{align}
		where $m_i$ is independently and uniformly picked from $\{(H^X)^Te_i,(H^Z)^Te_i,(H^X)^Te_i+(H^Z)^Te_i\}$. For random graph states, $H^X=I$ and $H^Z$ is uniformly and randomly picked from all symmetric matrices with zero diagonal. Therefore,  the result follows from the following Lemma~\ref{lem:gap-rand-graph-linear-algebra}.
	\end{proof}
	
	\begin{lemma}\label{lem:gap-rand-graph-linear-algebra}
		Let $c>0$ be any constant. Pick $A$ uniformly from random symmetric matrices in $\bF_2^{n\times n}$ with zero diagonal. Let $a_i$ denote the $i$-th row (or column) of $A$, and let $e_i$ denote the unit vector having $1$ on the $i$-th entry. 
		Then, with probability at least $1-\exp(-\Omega(n))$ over the choice of $A$, the following holds simultaneously for all nonzero $x\in\bF_2^n$:
		$\Pr_m[x\in\spn\{m_i\mid2\le i\le n\}]\le\frac12+\frac13+c$, where $m_i$ is independently and uniformly picked from $\{e_i,a_i,e_i+a_i\}$.
	\end{lemma}
	
	\begin{proof}
		For any matrix $A$ and vector $x$, we have
		\begin{align}
			\Pr_m[x\in\spn\{m_i\mid2\le i\le n\}]\le\bE_m\left[\sum_{k_2,\cdots,k_n\in\{0,1\}}1_{k_2m_2+\cdots+k_nm_n=x}\right]=2^{n-1}\Pr_{m'}[m'_2+\cdots+m'_n=x],\label{eq:pf-const-gap-rand-eq}
		\end{align}
		where $m'_i$ is picked from the following distribution: $0$ with probability $\frac12$, $e_i$ with probability $\frac16$, $a_i$ with probability $\frac16$, $e_i+a_i$ with probability $\frac16$; and each $m'_i$ is picked independently.
		
		We have another way to choose $m'$: let $m'_i=f_iz_i$, where $f_i$ is $1$ with probability $\frac23$ and $0$ with probability $\frac13$, $z_i$ is uniformly picked from $\{0,e_i,a_i,e_i+a_i\}$; and each $f_i,z_i$ are independent. Let
		\begin{align}
			V_f=\spn\{e_i,a_i\mid f_i=1\},\qquad |f|=\sum_{i=2}^nf_i,
		\end{align}
		where the addition is in $\bZ$. Then, $\dim V_f\le2|f|$, and
		\begin{align}\begin{split}
				\Pr_{m'}[m'_2+\cdots+m'_n=x]&=\bE_f\Pr_z[f_2z_2+\cdots+f_nz_n=x]=\bE_f\Pr_{y\gets V_f}[y=x]=\bE_f[1_{x\in V_f}2^{-{\dim V_f}}]\\&=\bE_f[1_{x\in V_f}2^{-2|f|}]+\bE_f[1_{x\in V_f}(2^{-{\dim V_f}}-2^{-2|f|})].
			\end{split}\label{eq:pf-const-gap-eq-two-terms}\end{align}
		
		Let $f'$ be another random vector drawn from the following distribution: $f'_i$ is $1$ with probability $\frac13$ and $0$ with probability $\frac23$. Then, for any function $X(f)$ of $f$,
		\begin{align}\label{eq:pf-const-gap-new-dist}
			2^{n-1}\bE_f[X(f)2^{-2|f|}]=\sum_{f\in\{0,1\}^{n-1}}2^{n-1}\frac{2^{|f|}}{3^{n-1}}X(f)2^{-2|f|}=\sum_{f\in\{0,1\}^{n-1}}\frac{2^{n-1-|f|}}{3^{n-1}}X(f)=\bE_{f'}[X(f')].
		\end{align}
		Use this equality in \eqref{eq:pf-const-gap-rand-eq} and \eqref{eq:pf-const-gap-eq-two-terms}, we have
		\begin{align}\begin{split}
				\Pr_m[x\in\spn\{m_i\mid2\le i\le n\}]&\le2^{n-1}\Pr_{m'}[m'_2+\cdots+m'_n=x]\\&=2^{n-1}(\bE_f[1_{x\in V_f}2^{-2|f|}]+\bE_f[1_{x\in V_f}(2^{-{\dim V_f}}-2^{-2|f|})])\\&=\bE_{f'}[1_{x\in V_{f'}}]+2^{n-1}\bE_f[1_{x\in V_f}(2^{-{\dim V_f}}-2^{-2|f|})]\\&\le\Pr_{f'}[x\in V_{f'}]+2^{n-1}\bE_f[2^{-{\dim V_f}}-2^{-2|f|}].
		\end{split}\end{align}
		When $x=0$, it gives
		\begin{align}\begin{split}
				2^{n-1}\Pr_{m'}[m'_2+\cdots+m'_n=0]&=\bE_{f'}[1_{0\in V_{f'}}]+2^{n-1}\bE_f[1_{0\in V_f}(2^{-{\dim V_f}}-2^{-2|f|})]\\&=1+2^{n-1}\bE_f[2^{-{\dim V_f}}-2^{-2|f|}].
		\end{split}\end{align}
		Therefore, 
		\begin{align}\begin{split}
				\Pr_m[x\in\spn\{m_i\mid2\le i\le n\}]&\le\Pr_{f'}[x\in V_{f'}]+2^{n-1}\bE_f[2^{-{\dim V_f}}-2^{-2|f|}]\\&=\Pr_{f'}[x\in V_{f'}]+2^{n-1}\Pr_{m'}[m'_2+\cdots+m'_n=0]-1.
		\end{split}\end{align}
		The result can be derived from the following two lemmas: Lemma~\ref{lem:pf-const-gap-0-concentrate} and Lemma~\ref{lem:pf-const-gap-le-one-third}.
	\end{proof}
	
	
	\begin{lemma}\label{lem:pf-const-gap-0-concentrate}
		For any constant $c_1>0$, $2^{n-1}\Pr_{m'}[m'_2+\cdots+m'_n=0]\le1+\frac12+c_1$ with probability at least $1-\exp(-{\Omega(n)})$ (with respect to the randomness of $A$), where $m'_i$ is picked from the following distribution: $0$ with probability $\frac12$, $e_i$ with probability $\frac16$, $a_i$ with probability $\frac16$, $e_i+a_i$ with probability $\frac16$; and each $m'_i$ is picked independently.
	\end{lemma}
	
	\begin{lemma}\label{lem:pf-const-gap-le-one-third}
		For any constant $c_2>0$, with probability at least $1-\exp(-{\Omega(n)})$ (with respect to the randomness of $A$), for all non-zero $x\in\bF_2^n$, $\Pr_{f'}[x\in V_{f'}]\le\frac13+c_2$, where $f'$ follows $\Ber(1/3)^{n-1}$.
	\end{lemma}
	
	In the following, for any set $S,T$ of indices, let $A_{S,T}$ denote the submatrix of $A$ by keeping the rows with indices in $S$ and columns with indices in $T$, i.e., $A_{S,T}=(A_{i,j})_{i\in S,j\in T}$; similarly, for vector $a$, let $a_S$ be the vector by keeping the entries in $a$ with indices in $S$, i.e., $a_S=(a_i)_{i\in S}$.
	
	
	Before proving these two lemmas, we first prove another lemma.
	
	\begin{lemma}\label{lem:rand-symm-matrix-vec}
		Let $x_1,\cdots,x_k$ be $k$ linearly independent vectors in $\bF_2^n$, and $d_1,\cdots,d_k\in \bF_2^n$. Then,
		\begin{align}
			\Pr_A[\forall 1\le i\le k, Ax_i=d_i]\le2^{-k(n-k)}.
		\end{align}
	\end{lemma}
	\begin{proof}
		Let $X$, $D$ be $n\times k$ matrices whose $i$-th column is $x_i$, $d_i$, respectively. ``$\forall 1\le i\le k, Ax_i=d_i$'' is equivalent to $AX=D$. Since $\{x_i\}$ is linearly indepedent, there exists some indices set $I$ of size $k$ such that $X_{I,[k]}$ is an invevrtible matrix. Let $\bar I$ denote the complementary of $I$. $AX=D$ implies
		\begin{align}
			D_{\bar I,[k]}=(AX)_{\bar I,[k]}=A_{\bar I,[n]}X=A_{\bar I,I}X_{I,[k]}+A_{\bar I,\bar I}X_{\bar I,[k]},
		\end{align}
		which is equivalent to
		\begin{align}
			A_{\bar I,I}=(D_{\bar I,[k]}-A_{\bar I,\bar I}X_{\bar I,[k]})X_{I,[k]}^{-1}.
		\end{align}
		Here, $A_{\bar I,I}$ is independent of $A_{\bar I,\bar I}$, each entry of $A_{\bar I,I}$ is independently uniformly chosen from $\bF_2$. Hence, we can first choose $A_{\bar I,\bar I}$; given $A_{\bar I,\bar I}$, the probabililty that the above equation holds is $2^{-k(n-k)}$. Therefore, $\Pr_A[(AX)_{\bar I,[k]}=D_{\bar I,[k]}]=2^{-k(n-k)}$ and
		\begin{align}
			\Pr_A[\forall 1\le i\le k, Ax_i=d_i]\le\Pr_A[(AX)_{\bar I,[k]}=D_{\bar I,[k]}]=2^{-k(n-k)}.
		\end{align}
	\end{proof}
	
	\begin{proof}[Proof of Lemma~\ref{lem:pf-const-gap-0-concentrate}]
		We compute the expectation and variance of $X=\Pr_{m'}[m'_2+\cdots+m'_n=0]$. Here, we pick $m'$ using the following method: first pick $y_i=y_i(a,r)$ from the following expressions: $0$ with probability $\frac12$, $e_i$ with probability $\frac16$, and $a+re_i$ with probability $\frac13$; then plug in $a_i$ and a random bit $r_i$ to get $m_i'=y_i(a_i,r_i)$. This procedure produces the same distribution. Let
		\begin{align}
			d_y=\sum_{i=2}^n1_{y_i=e_i}e_i,\qquad x_y=\sum_{i=2}^n1_{y_i=a+re_i}e_i.
		\end{align}
		Then, for any $i$, $(d_y)_i$ and $(x_y)_i$ cannot both be $1$, and
		\begin{align}
			m'_2+\cdots+m'_n=y_2(a_2,r_2)+\cdots+y_n(a_n,r_n)=d_y+Ax_y+\sum_{i=2}^nr_i(x_y)_ie_i.
		\end{align}
		
		To compute $\bE_A[X]$, we have
		\begin{align}
			\bE_A[X]=\bE_A\Pr_{m'}[m'_2+\cdots+m'_n=0]=\Pr_{y,A,r}[y_2(a_2,r_2)+\cdots+y_n(a_n,r_n)=0].
		\end{align}
		We first pick $y$, and then compute the probability of $m'_2+\cdots+m'_n=0$ under the randomness of $A$ and $r$.
		
		If for all $i$, $y_i$ is not $a_i+r_ie_i$, then $y_2(a_2,r_2)+\cdots+y_n(a_n,r_n)=0$ if and only if $y_i=0$ for all $i$. The probability of such event is $2^{-(n-1)}$.
		
		Otherwise, assume $y_i=a+re_i$ for some $i$. The probability of this case is $1-(\frac23)^{n-1}$. Now,
		\begin{gather}\begin{aligned}
				(m'_2+\cdots+m'_n)_{\overline{\{i\}}}&=(d_y)_{\overline{\{i\}}}+(Ax_y)_{\overline{\{i\}}}+\left(\sum_{j=2}^nr_j(x_y)_je_j\right)_{\overline{\{i\}}}\\&=(d_y)_{\overline{\{i\}}}+A_{\overline{\{i\}},\overline{\{i\}}}(x_y)_{\overline{\{i\}}}+A_{\overline{\{i\}},i}(x_y)_i+\sum_{j\ne i}r_j(x_y)_je_j\\&=A_{\overline{\{i\}},i}+(d_y)_{\overline{\{i\}}}+A_{\overline{\{i\}},\overline{\{i\}}}(x_y)_{\overline{\{i\}}}+\sum_{j\ne i}r_j(x_y)_je_j;
			\end{aligned}\\
			(m'_2+\cdots+m'_n)_i=(d_y)_i+(Ax_y)_i+\left(\sum_{j=2}^nr_j(x_y)_je_j\right)_i=r_ie_i+(d_y)_i+(Ax_y)_i.
		\end{gather}
		We randomly pick $A_{\overline{\{i\}},\overline{\{i\}}}$ under the constraints ``$A$ is symmetric'' and ``$A_{jj}=0$ for all $j$''; and the entries of $r_j$ ($j\ne i$). Then we randomly pick $A_{\overline{\{i\}},i}$. After that, $(m'_2+\cdots+m'_n)_{\overline{\{i\}}}$ is determined, and it is zero with probability $2^{-(n-1)}$. Then, we randomly pick $r_i$, and with probability $\frac12$, $(m'_2+\cdots+m'_n)_i=0$. Therefore, in this case, when $y$ is given,
		\begin{align}\label{eq:pf-lemma-0-random-A-uniform}
			\Pr_{A,r}[y_2(a_2,r_2)+\cdots+y_n(a_n,r_n)=0]=2^{-n}.
		\end{align}
		
		In conclusion,
		\begin{align}\label{eq:pf-lemma-0-expectation}
			\bE_A[X]=\Pr_{y,A,r}[y_2(a_2,r_2)+\cdots+y_n(a_n,r_n)=0]=2^{-(n-1)}+\left(1-\left(\frac23\right)^{n-1}\right)2^{-n}.
		\end{align}
		
		Now we bound $\bE_A[X^2]$. We have
		\begin{align}\begin{split}
				\bE_A[X^2]&=\bE_A[\Pr_{m'}[m'_2+\cdots+m'_n=0]\Pr_{m''}[m''_2+\cdots+m''_n=0]]\\&=\Pr_{y,y',A,r,r'}[y_2(a_2,r_2)+\cdots+y_n(a_n,r_n)=y'_2(a_2,r'_2)+\cdots+y'_n(a_n,r'_n)=0].
		\end{split}\end{align}
		We first pick $y,y'$, and then compute the probability of $m'_2+\cdots+m'_n=m''_2+\cdots+m''_n=0$ under the randomness of $A$ and $r,r'$.
		
		If for all $i$, $y_i$ is not $a+re_i$, then $y_2(a_2,r_2)+\cdots+y_n(a_n,r_n)=0$ if and only if $y_i=0$ for all $i$, and it is independent of $A$. Then, we can use the results when computing $\bE_A[X]$. The case ``for all $i$, $y'_i$ is not $a_i+r_ie_i$'' is similar. These two cases have contribution $(2^{-(n-1)})^2+2\cdot2^{-(n-1)}\cdot\left(1-\left(\frac23\right)^{n-1}\right)2^{-n}$.
		
		Now, we assume that there exists some $i,i'$ such that $y_i=a_i+r_ie_i$, $y'_{i'}=a_{i'}+r_{i'}e_{i'}$. Then, $x_y,x_{y'}\ne0$.
		
		First, we compute the probability of the case with $x_y=x_{y'}$:
		\begin{align}\begin{split}
				&\Pr_{y,y',A,r,r'}[y_2(a_2,r_2)+\cdots+y_n(a_n,r_n)=y'_2(a_2,r'_2)+\cdots+y'_n(a_n,r'_n)=0\land x_y=x_{y'}\ne0]\\=&\Pr_{y,y',A,r,r'}\left[d_y+Ax_y+\sum_{i=2}^nr_i(x_y)_ie_i=d_{y'}+Ax_{y'}+\sum_{i=2}^nr_i'(x_{y'})_ie_i=0\land x_y=x_{y'}\ne0\right]\\=&\Pr_{y,y',A,r,r'}\left[d_y+Ax_y+\sum_{i=2}^nr_i(x_y)_ie_i=0\land x_y=x_{y'}\ne0\land d_y+\sum_{i=2}^nr_i(x_y)_ie_i=d_{y'}+\sum_{i=2}^nr_i'(x_{y'})_ie_i\right].
		\end{split}\end{align}
		Here, $d_y+\sum_{i=2}^nr_i(x_y)_ie_i=d_{y'}+\sum_{i=2}^nr_i'(x_{y'})_ie_i$ is equivalent to that for all $2\le i\le n$, $(d_y)_i+r_i(x_y)_i=(d_{y'})_i+r'_i(x_{y'})_i$. Here, $(d_y)_i$ ($(d_{y'})_i$) and $(x_y)_i$ ($=(x_{y'})_i$) are not both $1$, so it is equivalent to that $d_y=d_{y'}$ and $r_i=r'_i$ for all $i$ with $(x_y)_i=1$. Therefore,  by utilizing \eqref{eq:pf-lemma-0-random-A-uniform},
		\begin{align}\begin{split}
				&\Pr_{y,y',A,r,r'}[y_2(a_2,r_2)+\cdots+y_n(a_n,r_n)=y'_2(a_2,r'_2)+\cdots+y'_n(a_n,r'_n)=0\land x_y=x_{y'}\ne0]
				\\=&\Pr_{y,y',A,r,r'}\left[d_y+Ax_y+\sum_{i=2}^nr_i(x_y)_ie_i=0\land x_y=x_{y'}\ne0\land d_y+\sum_{i=2}^nr_i(x_y)_ie_i=d_{y'}+\sum_{i=2}^nr_i'(x_{y'})_ie_i\right]
				\\=&\bE_{y,A,r}\left[\begin{aligned}
					&1\left[x_y\ne0\land d_y+Ax_y+\sum_{i=2}^nr_i(x_y)_ie_i=0\right]\\&\quad\cdot\Pr_{y',r'}\left[x_y=x_{y'}\land d_y=d_{y'}\land r_i=r'_i\text{ for all }i\text{ with }(x_y)_i=1\right]
				\end{aligned}\right]
				\\=&\bE_{y,A,r}\left[1\left[x_y\ne0\land d_y+Ax_y+\sum_{i=2}^nr_i(x_y)_ie_i=0\right]2^{-\sum_{i=2}^n(x_y)_i}\Pr_{y'}[x_y=x_{y'}\land d_y=d_{y'}]\right]\\=&\bE_y\left[1_{x_y\ne0}\Pr_{A,r}\left[d_y+Ax_y+\sum_{i=2}^nr_i(x_y)_ie_i=0\right]2^{-\sum_{i=2}^n(x_y)_i}\Pr_{y'}[y=y']\right]\\=&\bE_y\left[1_{x_y\ne0}2^{-n}2^{-\sum_{i=2}^n(x_y)_i}\Pr_{y'}[y=y']\right]\\=&2^{-n}\left(\bE_y\left[2^{-\sum_{i=2}^n(x_y)_i}\Pr_{y'}[y=y']\right]-\bE_y\left[1_{x_y=0}2^{-\sum_{i=2}^n(x_y)_i}\Pr_{y'}[y=y']\right]\right)\\=&2^{-n}\left(\left(\frac1{2^2}+\frac1{6^2}+\frac1{3^2}\times\frac12\right)^{n-1}-\left(\frac1{2^2}+\frac1{6^2}\right)^{n-1}\right)=2^{-n}\left(3^{-(n-1)}-\left(\frac5{18}\right)^{n-1}\right).
		\end{split}\end{align}
		
		Then, we compute the probability of the case with $x_y\ne x_{y'}$, and there exists some $i$ such that $(x_y)_i=(x_{y'})_i=1$. The probability of entering this case is
		\begin{align}\begin{split}
				&\Pr_{y,y'}[x_y,x_{y'}\ne0\land x_y\ne x_{y'}\land\exists i,(x_y)_i=(x_{y'})_i=1]\\=&\Pr_{y,y'}[x_y\ne x_{y'}\land\exists i,(x_y)_i=(x_{y'})_i=1]\\=&1-\Pr_{y,y'}[\forall i,((x_y)_i=0\lor (x_{y'})_i=0)]-\Pr_{y,y'}[x_y=x_{y'}]+\Pr_{y,y'}[x_y=x_{y'}=0]\\=&1-\left(1-\frac1{3^2}\right)^{n-1}-\left(\frac1{3^2}+\frac{2^2}{3^2}\right)^{n-1}+\left(\frac{2^2}{3^2}\right)^{n-1}=1-\left(\frac89\right)^{n-1}-\left(\frac59\right)^{n-1}+\left(\frac49\right)^{n-1}.
		\end{split}\end{align}
		Now we assume that $y,y'$ are given and compute
		\begin{align}\begin{split}
				&\Pr_{A,r,r'}[y_2(a_2,r_2)+\cdots+y_n(a_n,r_n)=y'_2(a_2,r'_2)+\cdots+y'_n(a_n,r'_n)=0]\\=&\Pr_{A,r,r'}\left[d_y+Ax_y+\sum_{k=2}^nr_k(x_y)_ke_k=d_{y'}+Ax_{y'}+\sum_{k=2}^nr_k'(x_{y'})_ke_k=0\right].
		\end{split}\end{align}
		Since $x_y\ne x_{y'}$, there exists some $j$ such that $(x_y)_j\ne(x_{y'})_j$. Without loss of generality, we assume that $(x_y)_j=0$, $(x_{y'})_j=1$. Because $(x_y)_i=(x_{y'})_i=1$, $i\ne j$. Then,
		\begin{align}
			&\begin{aligned}\left(d_y+Ax_y+\sum_{k=2}^nr_k(x_y)_ke_k\right)_{\overline{\{i,j\}}}&=A_{\overline{\{i,j\}},\overline{\{i,j\}}}(x_y)_{\overline{\{i,j\}}}+A_{\overline{\{i,j\}},i}(x_y)_i+A_{\overline{\{i,j\}},j}(x_y)_j\\&+(d_y)_{\overline{\{i,j\}}}+\sum_{k\ne i,j}r_k(x_y)_ke_k\\&=A_{\overline{\{i,j\}},i}+A_{\overline{\{i,j\}},\overline{\{i,j\}}}(x_y)_{\overline{\{i,j\}}}+(d_y)_{\overline{\{i,j\}}}+\sum_{k\ne i,j}r_k(x_y)_ke_k;\end{aligned}\label{eq:pf-lemma-0-s1-other}\\
			&\begin{aligned}\left(d_y+Ax_y+\sum_{k=2}^nr_k(x_y)_ke_k\right)_i&=(d_y)_i+A_{i,[n]}(x_y)+r_i(x_y)_i=r_i+(d_y)_i+A_{i,[n]}(x_y);\end{aligned}\label{eq:pf-lemma-0-s1-i}\\
			&\begin{aligned}\left(d_y+Ax_y+\sum_{k=2}^nr_k(x_y)_ke_k\right)_j&=(d_y)_j+A_{j,\overline{\{i,j\}}}(x_y)_{\overline{\{i,j\}}}+A_{j,i}(x_y)_i+A_{j,j}(x_y)_j+r_j(x_y)_j\\&=A_{j,i}+(d_y)_j+A_{j,\overline{\{i,j\}}}(x_y)_{\overline{\{i,j\}}};\end{aligned}\label{eq:pf-lemma-0-s1-j}\\
			&\begin{aligned}\left(d_{y'}+Ax_{y'}+\sum_{k=2}^nr'_k(x_{y'})_ke_k\right)_{\overline{\{i,j\}}}&=A_{\overline{\{i,j\}},\overline{\{j\}}}(x_{y'})_{\overline{\{j\}}}+A_{\overline{\{i,j\}},j}(x_{y'})_j+(d_{y'})_{\overline{\{i,j\}}}+\sum_{k\ne i,j}r'_k(x_{y'})_ke_k\\&=A_{\overline{\{i,j\}},j}+A_{\overline{\{i,j\}},\overline{\{j\}}}(x_{y'})_{\overline{\{j\}}}+(d_{y'})_{\overline{\{i,j\}}}+\sum_{k\ne i,j}r'_k(x_{y'})_ke_k;\end{aligned}\label{eq:pf-lemma-0-s2-other}\\
			&\begin{aligned}\left(d_{y'}+Ax_{y'}+\sum_{k=2}^nr'_k(x_{y'})_ke_k\right)_i&=(d_{y'})_i+A_{i,[n]}x_y+r'_i(x_{y'})_i=r'_i+(d_{y'})_i+A_{i,[n]}x_y;\end{aligned}\label{eq:pf-lemma-0-s2-i}\\
			&\begin{aligned}\left(d_{y'}+Ax_{y'}+\sum_{k=2}^nr'_k(x_{y'})_ke_k\right)_j&=(d_{y'})_j+A_{j,[n]}x_y+r'_j(x_{y'})_i=r'_j+(d_{y'})_j+A_{j,[n]}x_y.\end{aligned}\label{eq:pf-lemma-0-s2-j}
		\end{align}
		We consider the following procedure to sample $A$, $r$, $r'$: uniformly randomly sample $r_k,r'_k$ for all $k\ne i,j$, and uniformly randomly sample $A_{\overline{\{i,j\}},\overline{\{i,j\}}}$ under the constraints that ``$A_{\overline{\{i,j\}},\overline{\{i,j\}}}$ is symmetric and $A_{k,k}=0$ for all $k\ne i,j$''; set $A_{i,i}=A_{j,j}=0$; uniformly randomly sample $A_{\overline{\{i,j\}},i}$ and let $A_{i,\overline{\{i,j\}}}$ be its transpose, and now \eqref{eq:pf-lemma-0-s1-other} is determined and it is $0$ with probability $2^{-(n-2)}$; uniformly randomly sample $A_{\overline{\{i,j\}},j}$ and let $A_{j,\overline{\{i,j\}}}$ be its transpose, and now \eqref{eq:pf-lemma-0-s2-other} is determined and it is $0$ with probability $2^{-(n-2)}$; uniformly randomly sample $A_{i,j}=A_{j,i}$, and now \eqref{eq:pf-lemma-0-s1-j} is determined and it is $0$ with probability $2^{-1}$; uniformly randomly sample $r_i,r_j,r_i',r_j'$, and now \eqref{eq:pf-lemma-0-s1-i}, \eqref{eq:pf-lemma-0-s2-i}, \eqref{eq:pf-lemma-0-s2-j} are determined and they are all $0$ with probability $2^{-3}$. One can verify that this procedure produces the correct distribution of $A,r,r'$. Therefore, for fixed $y,y'$ with $x_y\ne x_{y'},x_y,x_{y'}\ne0$ and $(x_y)_i=(x_{y'})_i=1$ for some $i$, the probability of $m'_2+\cdots+m'_n=m''_2+\cdots+m''_n=0$ is $2^{-2n}$ under the randomness of $A$ and $r,r'$.
		
		Finally, we compute the probability of the case with $x_y\ne x_{y'}$, and there does not exist $i$ such that $(x_y)_i=(x_{y'})_i=1$. The probability of entering this case is
		\begin{align}\begin{split}
				&\Pr_{y,y'}[x_y,x_{y'}\ne0\land x_y\ne x_{y'}\land\forall i,((x_y)_i\ne(x_{y'})_i\lor (x_y)_i=(x_{y'})_i=0)]\\=&\Pr_{y,y'}[x_y,x_{y'}\ne0\land\forall i,((x_y)_i\ne(x_{y'})_i\lor (x_y)_i=(x_{y'})_i=0)]\\=&\Pr_{y,y'}[\forall i,((x_y)_i\ne(x_{y'})_i\lor (x_y)_i=(x_{y'})_i=0)]-\Pr_{y,y'}[x_y=0]-\Pr_{y,y'}[x_{y'}=0]+\Pr_{y,y'}[x_y=x_{y'}=0]\\=&\left(1-\frac1{3^2}\right)^{n-1}-2\left(\frac23\right)^{n-1}+\left(\frac{2^2}{3^2}\right)^{n-1}=\left(\frac89\right)^{n-1}-2\left(\frac23\right)^{n-1}+\left(\frac49\right)^{n-1}.
		\end{split}\end{align}
		In this case, $x_y\ne x_{y'}$ and $x_y,x_{y'}\ne0$, so $x_y$ and $x_{y'}$ are linearly independent. We pick $r,r'$ first then $A$, and utilize Lemma~\ref{lem:rand-symm-matrix-vec}, then we get
		\begin{align}\begin{split}
				&\Pr_{A,r,r'}[y_2(a_2,r_2)+\cdots+y_n(a_n,r_n)=y'_2(a_2,r'_2)+\cdots+y'_n(a_n,r'_n)=0]\\=&\Pr_{A,r,r'}\left[d_y+Ax_y+\sum_{i=2}^nr_i(x_y)_ie_i=d_{y'}+Ax_{y'}+\sum_{i=2}^nr_i'(x_{y'})_ie_i=0\right]\\=&\bE_{r,r'}\left[\Pr_A\left[Ax_y=-d_y-\sum_{i=2}^nr_i(x_y)_ie_i\land Ax_{y'}=-d_{y'}-\sum_{i=2}^nr_i'(x_{y'})_ie_i\right]\right]\le2^{-2(n-2)}.
		\end{split}\end{align}
		
		In summary,
		\begin{align}\begin{split}
				\bE[X^2]&\le(2^{-(n-1)})^2+2\cdot2^{-(n-1)}\cdot\left(1-\left(\frac23\right)^{n-1}\right)2^{-n}\\&+2^{-n}\left(3^{-(n-1)}-\left(\frac5{18}\right)^{n-1}\right)\\&+2^{-2n}\left(1-\left(\frac89\right)^{n-1}-\left(\frac59\right)^{n-1}+\left(\frac49\right)^{n-1}\right)\\&+2^{-2(n-2)}\left(\left(\frac89\right)^{n-1}-2\left(\frac23\right)^{n-1}+\left(\frac49\right)^{n-1}\right).
		\end{split}\end{align}
		Together with \eqref{eq:pf-lemma-0-expectation}, we have
		\begin{align}
			\Var[X]\le2^{-2n}\exp(-\Omega(n)).
		\end{align}
		
		Therefore, by Chebyshev's inequality,
		\begin{align}\begin{split}
				&\Pr_A\left[2^{n-1}\Pr_{m'}[m'_2+\cdots+m'_n=0]>1+\frac12+c_1\right]\\=&\Pr_A\left[2^{n-1}X>1+\frac12+c_1\right]\le\Pr_A\left[2^{n-1}X>2^{n-1}\bE[X]+c_1\right]\le\frac{2^{2(n-1)}}{c_1^2}\Var[X]=\exp(-\Omega(n)).
		\end{split}\end{align}
		Hence, $2^{n-1}\Pr_{m'}[m'_2+\cdots+m'_n=0]\le1+\frac12+c_1$ with probability at least $1-\exp(-{\Omega(n)})$ (with respect to the randomness of $A$).
	\end{proof}
	
	\begin{proof}[Proof of Lemma~\ref{lem:pf-const-gap-le-one-third}]
		Let $T$ denote the event ``for all non-zero $x\in\bF_2^n$, $\Pr_{f'}[x\in V_{f'}]\le\frac13+c_2$'', and we need to prove that $T$ happens with probability at least $1-\exp(-\Omega(n))$. Here, $f'$ follows $\Ber(1/3)^{n-1}$.
		
		Assume that $T$ does not happen, that is, there exists a non-zero $x\in\bF_2^n$ such that $\Pr_{f'}[x\in V_{f'}]>\frac13+c_2$. We consider the following procedure: for each $2\le i\le n$, independently uniformly sample $p_i$ from $\{0,1,2\}$. For $k=0,1,2$, let $(f_k)_i=1_{p_i=k}$; then $f_k$ follows $\Ber(1/3)^{n-1}$. By assumption, $\Pr[x\in V_{f_k}]>\frac13+c_2$, and hence
		\begin{align}\begin{split}
				1+3c_2&<\Pr[x\in V_{f_0}]+\Pr[x\in V_{f_1}]+\Pr[x\in V_{f_2}]=\bE[1[x\in V_{f_0}]+1[x\in V_{f_1}]+1[x\in V_{f_2}]]\\&\le\bE\left[1+\binom{1[x\in V_{f_0}]+1[x\in V_{f_1}]+1[x\in V_{f_2}]}{2}\right]\\&=1+(\Pr[x\in V_{f_0}\cap V_{f_1}]+\Pr[x\in V_{f_1}\cap V_{f_2}]+\Pr[x\in V_{f_2}\cap V_{f_0}]).
		\end{split}\end{align}
		By symmetry, $\Pr[x\in V_{f_0}\cap V_{f_1}]=\Pr[x\in V_{f_1}\cap V_{f_2}]=\Pr[x\in V_{f_2}\cap V_{f_0}]$. Therefore,
		\begin{align}
			\Pr[x\in V_{f_0}\cap V_{f_1}]>c_2. 
		\end{align}
		
		We independently sample another $q$ following the same distribution as $p$, and let $(g_k)_i=1_{q_i=k}$ for $k=0,1,2$. Then,
		\begin{align}\begin{split}
				c_2^2&<\Pr_{p,q}[x\in V_{f_0}\cap V_{f_1}\land x\in V_{g_0}\cap V_{g_1}]=\Pr_{p,q}[x\in V_{f_0}\cap V_{f_1}\cap V_{g_0}\cap V_{g_1}]\\&\le\Pr_{p,q}[V_{f_0}\cap V_{f_1}\cap V_{g_0}\cap V_{g_1}\ne\{0\}].
		\end{split}\end{align}
		Let $X=\Pr_{p,q}[V_{f_0}\cap V_{f_1}\cap V_{g_0}\cap V_{g_1}\ne\{0\}]$. When event $T$ does not happen, $X>c_2^2$.
		
		Now we bound the expectation of $X$ under the randomness of $A$, that is,
		\begin{align}
			\bE_A[X]=\Pr_{p,q,A}[V_{f_0}\cap V_{f_1}\cap V_{g_0}\cap V_{g_1}\ne\{0\}].
		\end{align}
		We first pick $p,q$ and then $A$. Let $I_k=\{i\mid p_i=k\}$, $J_k=\{i\mid q_i=k\}$ for $k=0,1,2$. $I_0\cap I_1=\varnothing$, $J_0\cap J_1=\varnothing$. $i\in I_k\cap J_{k'}$ if and only if $p_i=k,q_i=k'$, which happens with probability $\frac19$; ``$p_i=k,q_i=k'$'' and ``$p_j=k,q_j=k'$'' are independent for $i\ne j$. Therefore, by Chernoff's bound and union bound, with probability $1-\exp(-\Omega(n))$, $|I_k\cap J_{k'}|<(\frac19+0.01)n$ for all $k,k'\in\{0,1,2\}$.
		
		Now we fix $p,q$ with assumption $|I_k\cap J_{k'}|<(\frac19+0.01)n$ for all $k,k'\in\{0,1,2\}$, which implies $0<|I_k|,|J_k|<(\frac13+0.03)n$ for all $k\in\{0,1,2\}$; and we bound $\Pr_A[V_{f_0}\cap V_{f_1}\cap V_{g_0}\cap V_{g_1}\ne\{0\}]$. Here,
		\begin{align}\begin{split}
				V_{f_k}&=\spn\{e_i,a_i\mid (f_k)_i=1\}=\spn\{e_i,a_i\mid p_i=k\}=\spn\{e_i,a_i\mid i\in I_k\}\\&=\left\{\left.\sum_{i\in I_k}(d_ie_i+x_ia_i)\right\vert x,d\in\bF_2^{I_k}\right\}=\left\{Ax+d\left\vert x,d\in\bF_2^{I_k}\right.\right\}
		\end{split}\end{align}
		and similar for $V_{g_k}$. Then,
		\begin{align}\begin{split}
				&\Pr_A[V_{f_0}\cap V_{f_1}\cap V_{g_0}\cap V_{g_1}\ne\{0\}]\\\le&\bE_A\left[\sum_{x_1,d_1\in\bF_2^{I_0}}\sum_{x_2,d_2\in\bF_2^{I_1}}\sum_{x_3,d_3\in\bF_2^{J_0}}\sum_{x_4,d_4\in\bF_2^{J_1}}1[Ax_1+d_1=Ax_2+d_2=Ax_3+d_3=Ax_4+d_4\ne0]\right]\\=&\sum_{x_1,d_1\in\bF_2^{I_0}}\sum_{x_2,d_2\in\bF_2^{I_1}}\sum_{x_3,d_3\in\bF_2^{J_0}}\sum_{x_4,d_4\in\bF_2^{J_1}}\Pr_A[Ax_1+d_1=Ax_2+d_2=Ax_3+d_3=Ax_4+d_4\ne0].
		\end{split}\end{align}
		We bound $\Pr_A[Ax_1+d_1=Ax_2+d_2=Ax_3+d_3=Ax_4+d_4\ne0]$ for fixed $x_1,d_1\in\bF_2^{I_0},x_2,d_2\in\bF_2^{I_1},x_3,d_3\in\bF_2^{J_0},x_4,d_4\in\bF_2^{J_1}$; denote it as $p_{x,d}$. By Lemma~\ref{lem:rand-symm-matrix-vec},
		\begin{align}\begin{split}
				p_{x,d}&=\Pr_A[Ax_1+d_1=Ax_2+d_2=Ax_3+d_3=Ax_4+d_4\ne0]\\&\le\Pr_A[A(x_1+x_2)=d_1+d_2\land A(x_3+x_4)=d_3+d_4\land A(x_1+x_3)=d_1+d_3]\\&\le 2^{-k_{x,d}(n-k_{x,d})}=\cO(2^{-k_{x,d}n})
		\end{split}\end{align}
		(in $\bF_2$, plus and minus are the same), where
		\begin{align}
			k_{x,d}=\dim\spn\{x_1+x_2,x_3+x_4,x_1+x_3\}.
		\end{align}
		
		We only focus on those $x,d$ such that $p_{x,d}\ne0$, that is, there exists some $A$ such that $Ax_1+d_1=Ax_2+d_2=Ax_3+d_3=Ax_4+d_4\ne0$.
		
		When $x_1=x_2=0$, $Ax_1+d_1=Ax_2+d_2$ implies $d_1=d_2$; but $d_1\in\bF_2^{I_0}$, $d_2\in\bF_2^{I_1}$ and $I_0\cap I_1=\varnothing$. Hence, $d_1=d_2=0$, and $Ax_1+d_1=0$. Therefore, in this case $p_{x,d}=0$. We ignore this case, that is, we have $x_1\ne0$ or $x_2\ne0$. Similarly, $x_3\ne0$ or $x_4\ne0$.
		
		Since $x_1\in\bF_2^{I_0}$, $x_2\in\bF_2^{I_1}$ and $I_0\cap I_1=\varnothing$, $x_1\ne x_2$ (otherwise $x_1=x_2=0$); similarly, $x_3\ne x_4$.
		
		Now we discuss $k_{x,d}$ and count the number of $x,d$ such that $p_{x,d}\ne0$, that is, there exists some $A$ such that $Ax_1+d_1=Ax_2+d_2=Ax_3+d_3=Ax_4+d_4\ne0$. Here, $x_1+x_2=x_1-x_2\ne0$, $x_3+x_4\ne0$, so $k_{x,d}\in\{1,2,3\}$. Assume that $p_{x,d}\ne0$.
		\begin{itemize}
			\item $k_{x,d}=3$. The number of possible values of $x_1$ is $|\bF_2^{I_0}|=2^{|I_0|}<2^{(1/3+0.03)n}$. Similarly, $x_2$, $x_3$, $x_4$, $d_1$, $d_2$, $d_3$, $d_4$ each has less than $2^{(1/3+0.03)n}$ possible values. Hence,
			\begin{align}
				\sum_{x,d:k_{x,d}=3}p_{x,d}\le\sum_{x,d:k_{x,d}=3,p_{x,d}\ne0}\cO\left(2^{-3n}\right)\le2^{8(1/3+0.03)n}\cO\left(2^{-3n}\right)=\exp(-\Omega(n)).
			\end{align}
			\item $k_{x,d}=1$. Then, $x_1+x_2=x_3+x_4$; $x_1+x_3=x_1+x_2$ or $x_1+x_3=0$. The first case implies $x_2=x_3$ and then $x_1=x_4$, and the second case implies $x_1=x_3$ and then $x_2=x_4$.
			
			We assume $x_1=x_4$, $x_2=x_3$. Since $Ax_1+d_1=Ax_4+d_4$, $d_1=d_4$. Similarly $d_2=d_3$. Because $x_1\in\bF_2^{I_0}$, $x_4\in\bF_2^{J_1}$, $x_1=x_4\in\bF_2^{I_0}\cap\bF_2^{J_1}=\bF_2^{I_0\cap J_1}$, which has $2^{|I_0\cap J_1|}<2^{(1/9+0.01)n}$ possible values. Similarly, $x_2=x_3$, $d_1=d_4$, $d_2=d_3$ have less that $2^{(1/9+0.01)n}$ possible values. Hence, for this case, the number of possible $x,d$ is less than $2^{(4/9+0.04)n}$.
			
			The case for $x_1=x_3$, $x_2=x_4$ is similar. Therefore,
			\begin{align}
				\sum_{x,d:k_{x,d}=1}p_{x,d}\le\sum_{x,d:k_{x,d}=1,p_{x,d}\ne0}\cO\left(2^{-n}\right)<2\times2^{(4/9+0.04)n}\cO\left(2^{-n}\right)=\exp(-\Omega(n)).
			\end{align}
			\item $k_{x,d}=2$ with $x_1+x_2=x_3+x_4$. Since $x_1+x_2\in\bF_2^{I_0\cup I_1}=\bF_2^{I_0}\oplus\bF_2^{I_1}$, $x_3+x_4\in\bF_2^{J_0\cup J_1}=\bF_2^{J_0}\oplus\bF_2^{J_1}$, $x_1,x_2,x_3,x_4$ are determined by $x_1+x_2=x_3+x_4\in\bF_2^{I_0\cup I_1}\cap\bF_2^{J_0\cup J_1}=\bF_2^{(I_0\cap J_0)\cup(I_0\cap J_1)\cup(I_1\cap J_0)\cup(I_1\cap J_1)}$, which has $2^{|I_0\cap J_0|+|I_0\cap J_1|+|I_1\cap J_0|+|I_1\cap J_1|}<2^{(4/9+0.04)n}$ possible values.
			
			Since $Ax_1+d_1=Ax_2+d_2=Ax_3+d_3=Ax_4+d_4$, $A(x_1+x_2)+(d_1+d_2)=A(x_3+x_4)+(d_3+d_4)$ and hence $d_1+d_2=d_3+d_4$. By the same reason, the number of possible values of $(d_1,d_2,d_3,d_4)$ is less than $2^{(4/9+0.04)n}$. Therefore,
			\begin{align}
				\sum_{\substack{x,d:k_{x,d}=2\\x_1+x_2=x_3+x_4}}p_{x,d}\le\sum_{\substack{x,d:k_{x,d}=1,p_{x,d}\ne0\\x_1+x_2=x_3+x_4}}\cO\left(2^{-2n}\right)<2^{2(4/9+0.04)n}\cO\left(2^{-2n}\right)=\exp(-\Omega(n)).
			\end{align}
			\item $k_{x,d}=2$ with $x_1+x_2\ne x_3+x_4$. Since $x_1+x_2,x_3+x_4\ne0$, they are linearly independent, and hence $x_1+x_3\in\spn\{x_1+x_2,x_3+x_4\}=\{0,x_1+x_2,x_3+x_4,x_1+x_2+x_3+x_4\}$. It implies that at least one of $x_1+x_3$, $x_2+x_3$, $x_1+x_4$, $x_2+x_4$ is $0$.
			
			If $x_1+x_3=0$, then $x_1=x_3\in\bF_2^{I_0}\cap\bF_2^{J_0}=\bF_2^{I_0\cap J_0}$, which has $2^{|I_0\cap J_1|}<2^{(1/9+0.01)n}$ possible values. Together with $Ax_1+d_1=Ax_3+d_3$, $d_1=d_3$ and similarly $d_1=d_3$ has no more than $2^{(1/9+0.01)n}$ possible values. $x_2,x_4,d_2,d_4$ each has less than $2^{(1/3+0.03)n}$ possible values. In total, in this case there are no more than $2^{2(1/9+0.01)n+4(1/3+0.03)n}=2^{(14/9+0.14)n}$ possible $(x,d)$'s.
			
			The cases for $x_2+x_3=0$, $x_1+x_4=0$, $x_2+x_4=0$ are similar. Therefore,
			\begin{align}
				\sum_{\substack{x,d:k_{x,d}=2\\x_1+x_2\ne x_3+x_4}}p_{x,d}\le\sum_{\substack{x,d:k_{x,d}=1,p_{x,d}\ne0\\x_1+x_2\ne x_3+x_4}}\cO\left(2^{-2n}\right)<4\times2^{(14/9+0.14)n}\cO\left(2^{-2n}\right)=\exp(-\Omega(n)).
			\end{align}
		\end{itemize}
		
		In conclusion,
		\begin{align}
			\Pr_A[V_{f_0}\cap V_{f_1}\cap V_{g_0}\cap V_{g_1}\ne\{0\}]=\sum_{x,d}p_{x,d}\le\exp(-\Omega(n))
		\end{align}
		whenever $|I_k\cap J_{k'}|<(\frac19+0.01)n$ for all $k,k'\in\{0,1,2\}$ (which happens with probability $1-\exp(-\Omega(n))$ under the randomness of $p,q$). After including the randomness of $p,q$ and apply union bound, we have
		\begin{align}
			\bE_A[X]=\Pr_{p,q,A}[V_{f_0}\cap V_{f_1}\cap V_{g_0}\cap V_{g_1}\ne\{0\}]\le\exp(-\Omega(n)).
		\end{align}
		By Markov's inequality,
		\begin{align}
			\Pr_A[T\text{ does not happen}]\le \Pr\left[X>c_2^2\right]\le c_2^{-2}\bE_A[X]\le\exp(-\Omega(n)),
		\end{align}
		which finalizes the proof.
	\end{proof}
	
	\revise{
		\subsection{Non-zero gap for almost all states}
		In this subsection, we prove that the spectral gap $\Delta(O_{\psi})$ is strictly positive almost surely.
		\begin{theorem}\label{thm:spec_gap}
			Let $\ket{\psi}$ be drawn Haar-randomly from the $n$-qubit Hilbert space $\mathcal{H}_{d}$ with $d = 2^n$ and $n_A \ge 1$.  Then, with probability~$1$, $   \Delta(O_{\psi}) > 0$.
		\end{theorem}
	}
	
	Firstly, we show that a Haar-random pure state is almost surely entangled.
	\begin{lemma}\label{lem:Haar_entangled}
		Let $\ket{\psi}$ be drawn Haar-randomly from the bipartite Hilbert space $\mathcal{H}_{A}\otimes\mathcal{H}_{B}$ with local dimensions $d_A,d_B\ge 2$.  Then, with probability~$1$, $\ket{\psi}$ is entangled across the cut $A\mid B$.
	\end{lemma}
	
	\begin{proof}
		The set of all pure states in $\mathcal{H}_{AB}$, modulo a global phase, forms the complex projective space $\mathrm{CP}^{d-1}$ with real dimension $2d-2$, where $d=d_Ad_B$.  
		In contrast, the set of product (separable) states forms the Cartesian product $\mathrm{CP}^{d_A-1}\times\mathrm{CP}^{d_B-1}$, whose real dimension is $(2d_A-2)+(2d_B-2)$.  
		Because $d_A,d_B\ge 2$, we have
		\begin{equation}
			(2d_A-2)+(2d_B-2)<2d_Ad_B-2,
		\end{equation}
		so the manifold of separable states has strictly lower dimension than the manifold of all pure states.  
		A lower-dimensional submanifold has Haar measure zero within the larger manifold, implying that the probability of sampling a separable state under the Haar measure is zero.  Hence, a Haar-random pure state is almost surely entangled.
	\end{proof}
	
	The notation $O_{\psi}$ suppresses the dependence on the size $n_A$ of subsystem $A$. Here, we write $O_{\psi,n_A}$ to make this dependence explicit. We now show that when $B$ consists of a single qubit (i.e., $n_A=n-1$), the observable $O_{\psi,n-1}$ has a strictly positive spectral gap almost surely.
	
	\begin{lemma}\label{lem:Haar_spectral_singleB}
		Let $\ket{\psi}$ be drawn Haar-randomly from the $n$-qubit Hilbert space $\mathcal H_{AB}$ with $n_B=1$. Then
		\begin{equation}
			\Delta\bigl(O_{\psi,n-1}\bigr)>0
		\end{equation}
		with probability~1.
	\end{lemma}

	\begin{proof}
		Lemma~\ref{lem:Haar_entangled} guarantees that $\ket{\psi}$ is entangled across $A\mid B$ almost surely.  Observe that
		\begin{equation}
			O_{\psi,n-1}=\tfrac13\bigl(O_X+O_Y+O_Z\bigr),
		\end{equation}
		where 
		\begin{equation}\label{eq:Oz_def}
			O_Z
			=\psi_0\otimes\ketbra{0}
			+\psi_1\otimes\ketbra{1},
		\end{equation}
		with
		\begin{equation}
			\psi_0=\frac{\tr_B[(I_A\otimes\ketbra{0})\,\psi]}
			{\tr[(I_A\otimes\ketbra{0})\,\psi]},
			\qquad
			\psi_1=\frac{\tr_B[(I_A\otimes\ketbra{1})\,\psi]}
			{\tr[(I_A\otimes\ketbra{1})\,\psi]},
		\end{equation}
		and $O_X,O_Y$ are defined analogously in the $\{\ket{\pm}\}$ and $\{\ket{\pm i}\}$ bases.  Entanglement implies $\psi_0\neq\psi_1$; otherwise $\ket{\psi}$ would be separable.
		
		To prove $\Delta(O_{\psi,n-1})>0$, we show the maximal eigenvalue of $O_{\psi,n-1}$ is nondegenerate.  Suppose a pure state $\ket{\phi}$ attains this maximum so that $\tr(O_{\psi,n-1}\phi)=1$.  Then
		\begin{equation}
			\tr(O_X\phi)=\tr(O_Y\phi)=\tr(O_Z\phi)=1.
		\end{equation}
		The condition $\tr(O_Z\phi)=1$ implies that when subsystem $B$ of $\ket{\phi}$ is measured in the computational basis, the projected states on subsystem $A$ are precisely $\psi_0$ (for outcome $\ket{0}$) and $\psi_1$ (for outcome $\ket{1}$). Hence, 
		\begin{equation}
			\ket{\phi} = a_0\ket{0}_B\ket{\psi_{0}}_A+a_1\ket{1}_B\ket{\psi_{1}}_A, \quad a_0,a_1 \in \mathbb{C}.
		\end{equation}
		The uncertainty principle dictates that $\phi_B$ cannot simultaneously be an eigenstate of at least two of the operators $\sigma_X, \sigma_Y, \sigma_Z$. Assuming without loss of generality that $\phi_B$ is not an eigenstate of $\sigma_X$ and $\sigma_Z$, it follows that $a_0, a_1 \neq 0$.
		Similarly, for the entangled state $\ket{\psi}$, we have
		\begin{equation}
			\ket{\psi} = b_0 \ket{0}_B\ket{\psi_0}_A + b_1 \ket{1}_B\ket{\psi_1}_A, \quad b_0,b_1\neq0.
		\end{equation}
		Next, $\tr(O_X\phi)=1$ forces the projected states of $\ket{\phi}$ and $\ket{\psi}$ after measuring $B$ in $\{\ket{\pm}\}$ to coincide.  Thus
		\begin{equation}
			a_0\ket{\psi_0}+a_1\ket{\psi_1}
			\propto
			b_0\ket{\psi_0}+b_1\ket{\psi_1},
		\end{equation}
		so $a_0=eb_0$ and $a_1=eb_1$ for some $e\in\mathbb{C}$. Therefore $\ket{\phi}=e \ket{\psi}$, and normalization implies $|e|=1$.  Hence $\ket{\phi}$ and $\ket{\psi}$ differ only by a global phase, proving the maximum eigenvalue is nondegenerate and $\Delta(O_{\psi,n-1})>0$ almost surely.
	\end{proof}
	
	\begin{proof}[Proof of Theorem \ref{thm:spec_gap}]
		We focus on the case $n_A=1$; the extension to $n_A>1$ is analogous. We proceed by induction on the number of qubits, $n$.
		For the base case, $n=2$. A Haar-random $\ket{\psi}$ on $2$ qubits satisfies $\Delta(O_{\psi,1})>0$ almost surely by Lemma~\ref{lem:Haar_spectral_singleB}.
		When $n \ge 3$, assume $\Delta(O_{\psi,1}) > 0$ holds for $(n-1)$-qubit Haar-random state $\psi$ almost surely. We now show it holds for $n$-qubit Haar-random states.
		
		Let $\ket{\psi}$ be Haar-random on $n$ qubits and $\psi_i$ for $i \in \mathsf{X}$ be the projected state corresponding to Pauli-basis measurement on the $n$-th qubit. Each projected state is Haar-randomly distributed on $n-1$ qubits by unitary invariance of the Haar measure. 
		By the induction assumption, $\Delta(O_{\psi_i,1})>0$ almost surely for each \(i\). Applying union bound shows this holds simultaneously for both outcomes with probability~1.
		
		Note that
		\begin{equation}
			O_{\psi,1}
			= \frac{1}{3}\bigl(\tilde O_X + \tilde O_Y + \tilde O_Z\bigr),
		\end{equation}
		where 
		\begin{equation}
			\tilde O_Z
			= O_{\psi_0,1}\otimes\ketbra{0}
			+O_{\psi_1,1}\otimes\ketbra{1},
		\end{equation}
		and $\tilde O_X,\tilde O_Y$ are defined analogously. Suppose $\phi$ is any pure state satisfying  $\tr(O_{\psi,1} \phi) = 1$, then
		\begin{equation}
			\tr(\tilde{O}_X \phi) = \tr(\tilde{O}_Y\phi)  = \tr(\tilde{O}_Z \phi) = 1.
		\end{equation}
		Since each $O_{\psi_i,1}$ satisfies $\Delta(O_{\psi_i,1}) > 0$, the above equalities implies 
		\begin{equation}
			\tr(O_X \phi) = \tr(O_Y\phi)  = \tr(O_Z \phi) = 1,
		\end{equation}
		which then implies $\phi=\psi$, with same reasoning as in Lemma~\ref{lem:Haar_spectral_singleB}.  Hence, the maximal eigenvalue of $O_{\psi,1}$ is nondegenerate almost surely.
	\end{proof}

	\revise{
		\subsection{Proof of Theorem~\ref{thm:large_spectral}}
		
		We first introduce the notation that will be used extensively throughout this proof. 
		Let $m = n_B = n - 1$ denote the number of qubits in the complementary subsystem, and define the corresponding Hilbert space dimension as $D = 2^m$. 
		We denote the set of single-qubit Pauli basis states as $\mathsf{X} = \{0, 1, +, -, +i, -i\}$, and define the corresponding projectors as $\pi_x = \ket{x}\bra{x}$ for $x \in \mathsf{X}^m$. 
		Let $\mu = 3^{-m}$ denotes a uniform weight, which corresponds to the probability of selecting a specific Pauli basis. 
		Consequently, 
		\begin{equation}
			\sum_{x \in \mathsf{X}^m} \mu \pi_x = I_D.
		\end{equation}
		We express the target pure state $\ket{\psi}$ as 
		\begin{equation}
			\ket{\psi} = \ket{0}\ket{u} + \ket{1}\ket{v}, \qquad \text{with} \quad \Vert{}u\Vert{}^2 + \Vert{}v\Vert{}^2 = 1.
		\end{equation}
		Consider an arbitrary state $\ket{\phi} = \ket{0}\ket{a} + \ket{1}\ket{b}$. 
		For any measurement outcome $x$, define the amplitude $u_x = \bra{x}\ket{u}$ (with analogous definitions for $v_x, a_x, b_x$), and let
		\begin{equation}
			p_x = \vert{}u_x\vert{}^2 + \vert{}v_x\vert{}^2, \qquad d_x = \vert{}u_x b_x - v_x a_x\vert{}^2.
		\end{equation}
		For Haar-random states, $p_x \neq 0$ for all $x$ almost surely. Direct substitution into the definition of the conditional fidelity observable $O_\psi$ yields 
		\begin{equation}
			\bra{\phi}(I - O_\psi)\ket{\phi} = \sum_{x \in \mathsf{X}^m} \mu \, \frac{d_x}{p_x}.
		\end{equation}
		To bound this gap from below, we construct an auxiliary unnormalized vector
		\begin{equation}
			\ket{w_x} = \left(-\overline{v_x}\ket{0} + \overline{u_x}\ket{1}\right) \otimes \ket{x},
		\end{equation}
		and define two positive semi-definite operators, $H_0$ and $H_1$, as follows:
		\begin{equation}\label{eq:H0H1}
			\begin{split}
				H_0 &= \sum_{x \in \mathsf{X}^m} \mu  \ketbra{w_x}, \\     
				H_1 &= \sum_{x \in \mathsf{X}^m} \mu p_x  \ketbra{w_x}.
			\end{split}
		\end{equation}
		By construction, $\bra{\phi}H_0\ket{\phi} = \sum_{x} \mu d_x$ and $\bra{\phi}H_1\ket{\phi} = \sum_{x} \mu p_x d_x$. 
		Applying the Cauchy-Schwarz inequality, we can lower-bound the spectral gap as follows:
		\begin{equation}\label{eq:Cauchy-Schwarz-gap}
			\bra{\phi}(I - O_\psi)\ket{\phi} \ge \frac{\bra{\phi}H_0\ket{\phi}^2}{\bra{\phi}H_1\ket{\phi}}.
		\end{equation}
		The problem is thus reduced to bounding the operators $H_0$ and $H_1$. We accomplish this through the following two lemmas:
		\begin{lemma} \label{lem:H0}
			With probability at least $1 - \exp(-\Omega(m))$ over the Haar-random choice of $\ket{\psi}$, the kernel of $H_0$ is spanned exclusively by the target state (i.e., $\ker H_0 = \mathbb{C}\ket{\psi}$), and its smallest non-zero eigenvalue is bounded from below by$$\lambda_{\min}^+(H_0) \ge \frac{c}{D},$$where $c > 0$ is a universal constant.
		\end{lemma}
		
		\begin{lemma} \label{lem:H1}
			With probability at least $1 - \exp(-\Omega(m))$ over the Haar-random choice of $\ket{\psi}$, the maximum eigenvalue of $H_1$ is bounded from above by$$H_1 \le \frac{C}{D^2} I,$$where $C > 0$ is a universal constant.
		\end{lemma}
		By substituting Lemma \ref{lem:H0} and Lemma \ref{lem:H1} into Eq.~\eqref{eq:Cauchy-Schwarz-gap}, we establish Theorem \ref{thm:large_spectral}. 
		In the remainder of this section, we provide the detailed proofs for these two lemmas.
		
		\subsection{Proof of Lemma~\ref{lem:H0}}
		
		In this subsection, we present the proof of Lemma~\ref{lem:H0}. We begin by introducing several standard concentration inequalities concerning Gaussian random matrices. Next, we demonstrate that the operator $H_0$ can be expressed in a specific Gaussian matrix form. Finally, we separately bound the diagonal and off-diagonal blocks of $H_0$, synthesizing these bounds to establish the lemma.
		
		\subsubsection{Gaussian random matrices}
		
		Let $Z \sim \mathcal{N}_{\mathbb{C}}(0,1)$ denote a standard complex Gaussian random variable, defined as $Z = (\xi + i\eta)/\sqrt{2}$, where $\xi$ and $\eta$ are independent standard real Gaussian variables distributed as $\mathcal{N}(0, 1)$. Accordingly, an $r \times s$ standard complex Gaussian matrix is defined as a matrix whose entries are i.i.d. standard complex Gaussian variables.
		The singular values of such matrices satisfy the following bounds.
		

		\begin{fact}[Singular value concentration,~{\cite[Proposition 6.33]{Aubrun2017AliceBob}}]\label{fact:gaussian-exponential-tail-bound}
			There exists a universal constant $c>0$. Let $G$ be an $s \times n$ standard complex Gaussian matrix with $s \ge n$. Then, for any $t > 0$, the largest singular value $\Vert{}G\Vert{}$ satisfy 
			\begin{equation}
				\Pr\left[\|G\|\ge \sqrt s+\sqrt n+t\right] \le e^{-ct^2}.
			\end{equation}
			If $s > n$, then for every $t > 4\sqrt{2 \log n} / (\sqrt{s/n}-1)$, the smallest singular value $s_{\min}(G)$ satisfy 
			\begin{equation}
				\Pr\left[s_{\min}(G)\le \sqrt s - \sqrt n- t\right] \le e^{-ct^2}.
			\end{equation}
		\end{fact}
		
		\begin{lemma}[Linear summations]\label{lem:linear-summation}
			Let $\{\gamma_j\}$ be a sequence of independent standard complex Gaussian variables, and let $\{A_j\}$ be a sequence of matrices, each of dimension at most $D \times D$. Then, for any $p \ge 2$,
			\begin{equation}
				\left(\bE\left\| \sum_j \gamma_j A_j \right\|^p\right)^{1/p} \le C \sqrt{p} \, D^{1/p} \max\left\{ \left\| \sum_j A_j A_j^\dagger \right\|^{1/2}, \left\|\sum_j A_j^\dagger A_j \right\|^{1/2} \right\},
			\end{equation}
			where $C > 0$ is a universal constant.
		\end{lemma}
		
		\begin{proof}
			Let $\|T\|_{S_p} = (\tr|T|^p)^{1/p}$ denote the Schatten-$p$ norm of a matrix $T$. For any matrix of dimension at most $D \times D$, the standard operator norm is bounded by the Schatten-$p$ norm such that $\|T\| \le \|T\|_{S_p} \le D^{1/p} \|T\|$. Then, we can bound the expectation as follows:
			\begin{equation}
				\begin{aligned}
					\left(\bE\left\| \sum_j \gamma_j A_j \right\|^p\right)^{1/p} &\le \left(\bE\left\| \sum_j \gamma_j A_j \right\|_{S_p}^p\right)^{1/p} \\
					&\le C \sqrt{p} \max\left\{ \left\| \left(\sum_j A_j A_j^\dagger\right)^{1/2} \right\|_{S_p}, \left\| \left(\sum_j A_j^\dagger A_j\right)^{1/2} \right\|_{S_p} \right\} \\
					&\le C \sqrt{p} \, D^{1/p} \max\left\{ \left\| \sum_j A_j A_j^\dagger \right\|^{1/2}, \left\| \sum_j A_j^\dagger A_j \right\|^{1/2} \right\}.
				\end{aligned}
			\end{equation}
			Here, the second inequality follows directly from the non-commutative Khintchine inequality (see Ref.~\cite[Proposition 19, Corollary 20]{Tropp2008ConditioningSubdictionaries}).
		\end{proof}
		
		\begin{lemma}[Complex Hanson-Wright inequality] \label{lem:hanson_wright}
			Let $g \in \mathbb{C}^D$ be a standard complex Gaussian random vector, and let $A \in \mathbb{C}^{D \times D}$. Then, for any $t > 0$,
			\begin{equation}
				\Pr\left( \left| g^\dagger A g - \bE[g^\dagger A g] \right| > t \right) \le 4 \exp\left( -c \min\left\{ \frac{t^2}{\norm{A}_F^2}, \frac{t}{\norm{A}} \right\} \right),
			\end{equation}
			where $c > 0$ is a universal constant, and $\|A\|_F$ denotes the Frobenius norm of $A$.
		\end{lemma}
		
		\begin{proof}
			We first reduce the complex quadratic form to a real one. Let $g = x + iy$, where $x, y \sim \mathcal{N}(0, I_D / 2)$ are independent real Gaussian vectors. We define a $2D$-dimensional standard real Gaussian vector $w = \sqrt{2} (x \oplus y) \sim \mathcal{N}(0, I_{2D})$. 
			
			By decomposing the matrix $A = A_R + i A_I$ into its real and imaginary parts, the quadratic form can be expanded as
			\begin{equation}
				g^\dagger A g = \frac{1}{2} w^T \tilde{A} w + i \frac{1}{2} w^T \tilde{B} w,
			\end{equation}
			where
			\begin{equation}
				\tilde{A} = \begin{pmatrix} A_R & -A_I \\ A_I & A_R \end{pmatrix}, \qquad \tilde{B} = \begin{pmatrix} A_I & A_R \\ -A_R & A_I \end{pmatrix}
			\end{equation}
			are deterministic real $2D \times 2D$ matrices.
			
			Applying the standard real Hanson-Wright inequality (see, e.g., Ref.~\cite[Theorem 1.1]{Rudelson2013}) to the real part $\frac{1}{2} w^T \tilde{A} w$, we obtain
			\begin{equation}
				\Pr\left( \left| w^T \tilde{A} w - \bE[w^T \tilde{A} w] \right| > t \right) \le 2 \exp\left( -c \min\left\{ \frac{t^2}{\norm{\tilde{A}}_F^2}, \frac{t}{\norm{\tilde{A}}} \right\} \right),
			\end{equation}
			and a similar bound applies to $\tilde{B}$. Because the spectrum of $\tilde{A}$ (and analogously $\tilde{B}$) exactly matches the spectrum of $A$ up to doubled multiplicities, their respective norms satisfy
			\begin{equation}
				\|\tilde{A}\| = \|A\|, \qquad \|\tilde{A}\|_F^2 = 2\|A\|_F^2.
			\end{equation}
			Therefore, the concentration bound for the real part scales directly with $\|A\|$ and $\|A\|_F$. Applying an identical argument to the imaginary part $\frac{1}{2} w^T \tilde{B} w$ and taking a union bound over both components yields the desired exponential tail bound, absorbing the scalar factors into the universal constant $c$.
		\end{proof}

		\subsubsection{Matrix representation of $H_0$}
		We define the channel $\Lambda$ as
		\begin{equation}
			\Lambda(X) = \sum_{x \in \mathsf{X}^m} \mu \, \bra{x} X \ket{x} \pi_x.
		\end{equation}
		For a single qubit, this channel takes the form $\Lambda_1(X) = (X + \tr(X)I)/3$. Since $\Lambda = \Lambda_1^{\otimes m}$, we can expand this operator over subsets of qubits as
		\begin{equation}\label{eq:subset-expansion}
			\Lambda(X) = 3^{-m} \sum_{S \subseteq [m]} \tr_{S^c}(X) \otimes I_{S^c}. 
		\end{equation}
		Furthermore, utilizing the single-qubit identity $\sum_{P \in \cP_1} P X P = 2\tr(X)I_2$, we can alternatively express $\Lambda(X)$ as a Pauli channel:
		\begin{equation}\label{eq:Lambda_Pauli}
			\Lambda(X) = \sum_{P \in \cP_m} p_P P X P, \qquad p_P = D^{-1} 3^{-\wt(P)}, 
		\end{equation}
		where $\cP_m = \{I, X, Y, Z\}^{\otimes m}$ denotes the $m$-qubit Pauli group and $\wt(P)$ is the Pauli weight of operator $P$.
		
		Let $g, h \in \mathbb{C}^D$ be independent standard complex Gaussian vectors. Define $T = \|g\|^2 + \|h\|^2$, and set the normalized vectors as $u = g/\sqrt{T}$ and $v = h/\sqrt{T}$.
		The projector $\ket{w_x}\bra{w_x}$ can be written in block matrix form with respect to the bipartition as
		\begin{equation}
			\ketbra{w_x} = \begin{pmatrix}
				\bra{x} v v^\dagger \ket{x} & -\bra{x} u v^\dagger \ket{x} \\
				- \bra{x} v u^\dagger \ket{x} & \bra{x} u u^\dagger \ket{x}
			\end{pmatrix} \otimes \pi_x.
		\end{equation}
		Consequently, the operator $H_0$ can be expressed compactly using the map $\Lambda$:
		\begin{equation} \label{eq:H0_matrix}
			\begin{split}
				H_0 &= \begin{pmatrix}
					\sum_{x } \mu \bra{x} v v^\dagger \ket{x} \pi_x & - \sum_{x } \mu \bra{x} u v^\dagger \ket{x} \pi_x \\
					-\sum_{x } \mu \bra{x} v u^\dagger \ket{x} \pi_x & \sum_{x } \mu \bra{x} u u^\dagger \ket{x} \pi_x
				\end{pmatrix} \\
				&= \frac{1}{T} \begin{pmatrix}
					\Lambda(h h^\dagger) & -\Lambda(g h^\dagger) \\
					-\Lambda(h g^\dagger) & \Lambda(g g^\dagger)
				\end{pmatrix}.
			\end{split}
		\end{equation}
		
		\subsubsection{Diagonal blocks}
		
		We first analyze the diagonal block $\Lambda(gg^\dagger)$. For any subset $S \subseteq [m]$, let $G_S$ denote the Gaussian matrix obtained by reshaping the vector $g$ across the bipartition $S|S^c$. Using the subset expansion of $\Lambda$, we have
		\begin{equation}\label{eq:expression-diagonal}
			\begin{split}
				\Lambda(gg^\dagger) &= 3^{-m} \sum_{S} \operatorname{tr}_{S^c}(gg^\dagger) \otimes I_{S^c} \\
				&= 3^{-m} \sum_{S} G_S G_S^\dagger \otimes I_{S^c}.
			\end{split}
		\end{equation}
		
		Setting $t = C'\sqrt{m}$ for a sufficiently large constant $C'$, we apply Fact~\ref{fact:gaussian-exponential-tail-bound}. A union bound ensures that with probability at least $1 - e^{-\Omega(m)}$, the operator norm bound $\| G_S \| \le \sqrt{2^{|S|}} + \sqrt{2^{m-|S|}} + t$ holds simultaneously for all $2^m$ subsets. Consequently,
		\begin{equation}
			\begin{split}
				\| \Lambda(gg^\dagger) \| &\lesssim 3^{-m} \sum_S (2^{|S|} + 2^{m-|S|} + C'm) = \mathcal{O}(1),
			\end{split}
		\end{equation}
		where $A \lesssim B$ denotes $A \le c B$ for some absolute constant $c > 0$.
		
		To lower-bound the smallest singular value, we rewrite Eq.~\eqref{eq:expression-diagonal} as a convex combination:
		\begin{equation}\label{eq:bound-g-norm}
			\Lambda(gg^\dagger) = \sum_S q_S \left(\frac{G_S G_S^\dagger}{2^{m-|S|}}\right) \otimes I_{S^c}, \qquad q_S = (1/3)^{|S|}(2/3)^{m-|S|}.
		\end{equation}
		The coefficients $q_S$ correspond to a binomial distribution over the subsets $S$, where each element is included independently with probability $1/3$. By standard concentration bounds, the total probability mass on subsets sized $m/4 \le |S| \le 5m/12$ is $1 - e^{-\Omega(m)}$. Uniformly over these typical subsets, Fact~\ref{fact:gaussian-exponential-tail-bound} guarantees that $s_{\min}(G_S G_S^\dagger / 2^{m-|S|}) \ge 1 - o(1)$. Therefore, with probability at least $1 - e^{-\Omega(m)}$, we obtain
		\begin{equation}\label{eq:bound-g-smin}
			s_{\min}(\Lambda(gg^\dagger)) \ge 1 - o(1).
		\end{equation}
		
		The same analysis also applies to the vector $h$. Combining the above, we conclude that with probability at least $1 - e^{-\Omega(m)}$, both diagonal blocks satisfy
		\begin{equation}\label{eq:diagonal-bound-final}
			(1-o(1))I \le \Lambda(gg^\dagger), \Lambda(hh^\dagger) \le C I.
		\end{equation}
		
		\subsubsection{The off-diagonal block}
		
		We now analyze the off-diagonal block $B = \Lambda(gh^\dagger)$. By Eq.~\eqref{eq:Lambda_Pauli}, we can expand $B$ as
		\begin{equation}\label{eq:B-decomposition}
			B = D^{-1}gh^\dagger + \sum_{P \neq I} p_P P g h^\dagger P.
		\end{equation}
		We choose a cutoff weight $r = C_0 \log m$ for a sufficiently large constant $C_0$, and split the summation into low-weight and high-weight Pauli contributions:
		\begin{equation}
			E_L = \sum_{1 \le \wt(P) \le r} p_P P g h^\dagger P, \qquad E_H = \sum_{\wt(P) > r} p_P P g h^\dagger P.
		\end{equation}
		The number of Pauli operators with weight at most $r$ is bounded by $\exp(\mathcal{O}((\log m)^2)) = D^{o(1)}$.
		
		We define the matrix $G \coloneqq [Pg]_P$, where each column is given by $Pg$ for $ 0\le \wt(P) \le r$. Setting $t = D^{2/3}$ and applying Lemma~\ref{lem:hanson_wright}, we can bound the matrix entries $(G^\dagger G)_{PQ} = g^\dagger P Q g$:
		\begin{equation}\label{eq:bound-diag-offdiag}
			\Pr( |(G^\dagger G)_{PQ} - \delta_{PQ} D| > t ) = e^{-\Omega(D^{1/3})}.
		\end{equation}
		Taking a union bound over all such $P$ and $Q$, Eq.~\eqref{eq:bound-diag-offdiag} holds uniformly for all pairs with probability at least $1 - e^{-D^{\Omega(1)}}$. An analogous concentration bound applies to the matrix $H \coloneqq [Ph]_P$.
		Applying Gershgorin's circle theorem to the rows of $G^\dagger G$ and $H^\dagger H$, the entry-wise bounds imply that the spectra are highly concentrated around $D$:
		\begin{equation}\label{eq:GH-bound}
			D^{-1}G^\dagger G = (1 + o(1))I, \qquad D^{-1}H^\dagger H = (1 + o(1))I.
		\end{equation}
		The low-weight component can be factored as $E_L = G W H^\dagger$, where $W$ is a diagonal matrix with $W_{PP} = p_P$ for $P \neq I$ and $W_{II} = 0$. Because $\max_{P \neq I} p_P = 1/(3D)$, we can bound the operator norm as
		\begin{equation}\label{eq:bound-lower-E}
			\begin{split}
				\norm{E_L} &\le \norm{G} \norm{W} \norm{H} \\
				&\le \sqrt{D(1+o(1))} \cdot \frac{1}{3D} \cdot \sqrt{D(1+o(1))} \\
				&= \frac{1}{3} + o(1).
			\end{split}
		\end{equation}
		
		For the high-weight component $E_H$, we write
		\begin{equation}
			E_H = \sum_j \overline{h_j} A_j, \qquad A_j = \sum_{\wt(P)>r} p_P P g e_j^\dagger P.
		\end{equation}
		We aim to bound $E_H$ using Lemma~\ref{lem:linear-summation}. First, with probability at least $1-e^{-\Omega(m)}$, we have
		\begin{equation}\label{eq:bound-low-weight-A}
			\begin{split}
				\sum_j A_j A_j^\dagger &= \sum_j \sum_{\wt(P), \wt(Q) > r} p_P p_Q P g e_j^\dagger P Q e_j g^\dagger Q \\
				&= D \sum_{\wt(P)>r} p_P^2 P g g^\dagger P \\
				&\le 3^{-r} \sum_{P} p_P P g g^\dagger P \\
				&= 3^{-r} \Lambda(g g^\dagger) \\
				&\le C 3^{-r} I,
			\end{split}
		\end{equation}
		where the final inequality follows from Eq.~\eqref{eq:diagonal-bound-final}.
		
		To bound $\sum_j A_j^\dagger A_j$, we introduce
		\begin{equation}
			A_j^{\mathrm{all}} = \sum_{P \in \mathcal{P}m} p_P P g e_j^\dagger P, \qquad A_j^{\le r} = \sum_{\wt(P) \le r} p_P P g e_j^\dagger P.
		\end{equation}
		Evaluating the full summation yields
		\begin{equation}\label{eq:sum-A-all}
			\begin{split}
				\sum_j (A_j^{\mathrm{all}})^\dagger A_j^{\mathrm{all}} &= \sum_{P,Q} p_P p_Q Q ( \sum_j e_j g^\dagger Q P g e_j^\dagger ) P \\
				&= \sum_{P,Q} p_P p_Q (g^\dagger Q P g) Q P \\
				&= \sum_R ( \sum_{P,Q: QP = \omega R} p_P p_Q \omega^2 ) (g^\dagger R g) R \\
				&= \sum_R c_R (g^\dagger R g) R,
			\end{split}
		\end{equation}
		where $c_R \coloneqq \sum_{P,Q: QP = \omega R} p_P p_Q \omega^2$.
		
		For single-qubit Paulis, these coefficients evaluate to
		\begin{equation}
			c_I = \frac{1}{4} + \frac{3}{36} = \frac{1}{3}, \qquad c_X = c_Y = c_Z = 2 \frac{1}{2} \frac{1}{6} - 2 \frac{1}{36} = \frac{1}{9}.
		\end{equation}
		Consequently,
		\begin{equation}
			\begin{split}
				\sum_R c_R (g^{\dagger} R g)R &= \frac{1}{3} \norm{g}^2 I + \frac{1}{9}[(g^{\dagger} X g)X + (g^{\dagger} Y g) Y + (g^{\dagger} Z g) Z]\\
				&= \frac{1}{3} \norm{g}^2 I  + \frac{1}{9}(2gg^{\dagger} - \norm{g}^2 I) \\
				&= \frac{2}{3} \left(\frac{\norm{g}^2 I + gg^{\dagger}}{3}\right) \\
				&= \frac{2}{3} \Lambda_1(gg^{\dagger}),
			\end{split}
		\end{equation}
		where we utilized the identity $\sum_{P \in \mathcal{P}_1} (g^\dagger P g) P = 2 g g^\dagger$ alongside Eq.~\eqref{eq:subset-expansion}. Substituting this into Eq.~\eqref{eq:sum-A-all}, we obtain
		\begin{equation}
			\sum_j (A_j^{\mathrm{all}})^\dagger A_j^{\mathrm{all}} = (2/3)^m \Lambda(g g^\dagger).
		\end{equation}
		
		Similarly, for the low-weight component, we have
		\begin{equation}
			\begin{split}
				\sum_j (A_j^{\le r})^\dagger A_j^{\le r} &= \sum_{\wt(P), \wt(Q) \le r} p_P p_Q (g^\dagger Q P g) Q P \\
				&= \sum_{\wt(P), \wt(Q) \le r} p_P p_Q (G^{\dagger}G)_{QP} Q P \\
				&\le D(1+o(1)) \sum_{\wt(P) \le r} p_P^2 I \\
				&\le D^{-1+o(1)}(1+o(1)) I,
			\end{split}
		\end{equation}
		where the first inequality uses Eq.~\eqref{eq:GH-bound}.
		
		Because $A_j = A_j^{\mathrm{all}} - A_j^{\le r}$, we can apply the Loewner inequality $(X-Y)^\dagger(X-Y) \le 2X^\dagger X + 2Y^\dagger Y$. Combined with Eq.~\eqref{eq:diagonal-bound-final}, this ensures that with probability at least $1-e^{-\Omega(m)}$,
		\begin{equation}\label{eq:bound-high-weight-A}
			\norm{\sum_j A_j^\dagger A_j} \le 2(2/3)^m \norm{\Lambda(g g^\dagger)} + 2 D^{-1+o(1)} = e^{-\Omega(m)}.
		\end{equation}
		
		Combining Eq.~\eqref{eq:bound-low-weight-A}, Eq.~\eqref{eq:bound-high-weight-A}, and Lemma~\ref{lem:linear-summation} evaluated at $p = m$, we obtain
		\begin{equation}
			\left(\bE \norm{E_H}^m\right)^{1/m} \coloneqq \epsilon = o(1)
		\end{equation}
		with probability $1-e^{-\Omega(m)}$. Conditional on this event, applying Markov's inequality yields
		\begin{equation}
			\begin{split}
				\Pr\left[ \norm{E_H} \ge e \epsilon \right] &= \Pr\left[ \norm{E_H}^m \ge e^m \epsilon^m \right] \\
				&\le \frac{\epsilon^m}{e^m \epsilon^m} = e^{-m}.
			\end{split}
		\end{equation}
		Therefore, a union bound guarantees that $\norm{E_H} = o(1)$ with probability at least $1-e^{-\Omega(m)}$.
		We now combine this result with the decomposition of $B$ in Eqs.~\eqref{eq:B-decomposition} and \eqref{eq:bound-lower-E}. Applying the Eckart-Young inequality, the second singular value of $B$ is bounded by
		\begin{equation}
			s_2(B) \le \norm{E_L} + \norm{E_H} = \frac{1}{3} + o(1).
		\end{equation}
		
		Let $A = \Lambda(hh^\dagger)$, $C = \Lambda(gg^\dagger)$, $K = A^{-1/2}BC^{-1/2}$, $M = \begin{pmatrix} I & -K \\ -K^\dagger & I \end{pmatrix}$, and $D_0 = \operatorname{diag}(A^{1/2}, C^{1/2})$. Then, the second singular value of $K$ satisfies
		\begin{equation}
			s_2(K) \le \norm{A^{-1/2}} \norm{C^{-1/2}} s_2(B) \le \frac{1}{3} + o(1) < \frac{1}{2}.
		\end{equation}
		Moreover, by Eq.~\eqref{eq:H0_matrix}, $H_0$ factorizes as $H_0 = \frac{1}{T} D_0 M D_0$, and by its definition in Eq.~\eqref{eq:H0H1}, $H_0(\ket{0}{g} + \ket{1}\ket{h})=0$.
		Consequently, $K$ has its largest singular value $s_1(K) = 1$. Thus, $H_0$ has a one-dimensional kernel, and the smallest non-zero eigenvalue of $M$ is bounded by $\lambda_{\min}^+(M) = 1 - s_2(K) > 1/2$.
		
		Furthermore, the non-zero spectrum of $D_0 M D_0$ is identical to the non-zero spectrum of $M^{1/2} D_0^2 M^{1/2}$. Since the diagonal blocks satisfy $D_0^2 \ge (1-o(1))I$, we have
		\begin{equation}
			\lambda_{\min}^+(H_0) \ge \frac{1}{T} \lambda_{\min}(D_0)^2 \lambda_{\min}^+(M) = \Omega(T^{-1}).
		\end{equation}
		
		Finally, by Fact~\ref{fact:gaussian-exponential-tail-bound}, the normalization constant 
		\begin{equation}\label{eq:T-concentration}
			T = \norm{g}^2 + \norm{h}^2 = \Theta(D)
		\end{equation}
		except with probability $e^{-cD}$. This establishes the lower bound $\lambda_{\min}^+(H_0) \ge \Omega(D^{-1})$ and completes the proof of Lemma~\ref{lem:H0}.
		
		\subsection{Proof of Lemma~\ref{lem:H1}}
		
		In this subsection, we present the proof of Lemma~\ref{lem:H1}. We begin by introducing several integration tools for complex Gaussian random variables, which we will subsequently use to calculate matrix moments. These tools primarily rely on Wick ordering and the Wick formula.
		
		\subsubsection{Wick ordering and Wick formula}
		A family of complex random variables $\{Z_1, \ldots, Z_r\}$ is termed \emph{proper} if
		\begin{equation}
			\bE[Z_i Z_j] = 0 \qquad \forall i, j.
		\end{equation}
		For the standard complex Gaussian vector $g \sim \mathcal{N}_{\mathbb{C}}(0, I_D)$ introduced above, $\bE[g g^T] = 0$. Hence, for any set of unit vectors $\{x_1, \ldots, x_r\}$, the random variables
		\begin{equation}
			g_{x_i} \coloneqq \bra{x_i}\ket{g}
		\end{equation}
		form a proper standard complex Gaussian family with covariance
		\begin{equation}\label{eq:H1-complex-covariance}
			C_{ij} \coloneqq \bE[g_{x_i} \overline{g_{x_j}}] = \bra{x_i}\ket{x_j}.
		\end{equation}
		
		For a standard proper complex Gaussian variable $Z$, its Wick-ordered monomials are defined via the generating function
		\begin{equation}\label{eq:H1-wick-generating-function}
			\exp(sZ + t\overline{Z} - st) = \sum_{p,q \ge 0} \frac{s^p t^q}{p!q!} :\!Z^p\overline{Z}^q\!: .
		\end{equation}
		In particular (see Ref.~\cite[Eqs.~(1.2), (1.3)]{IsmailSimeonov2015ComplexHermite}), the quartic Wick polynomial is given by
		\begin{equation}\label{eq:H1-quartic-wick-polynomial}
			:\!Z^2\overline{Z}^2\!: \,\, = |Z|^4 - 4|Z|^2 + 2.
		\end{equation}
		The key tool we use is the following formula.
		
		\begin{lemma}[Wick formula]\label{lem:wick-formula}
			Let $Z_1, \ldots, Z_r$ be jointly proper standard complex Gaussian variables, and let $C_{ij} = \bE[Z_i \overline{Z_j}]$. For nonnegative integers $p_i, q_i$, let $\mathcal{M}(\bm{p}, \bm{q})$ be the set of matrices $N = (n_{ij})_{i,j=1}^r$ with nonnegative integer entries satisfying the constraints
			\begin{equation}
				n_{ii} = 0, \qquad \sum_j n_{ij} = p_i, \qquad \sum_i n_{ij} = q_j.
			\end{equation}
			Then,
			\begin{equation}\label{eq:H1-block-wick-formula}
				\bE \prod_{i=1}^r :\!Z_i^{p_i}\overline{Z_i}^{q_i}\!: \,\, = \sum_{N \in \mathcal{M}(\bm{p}, \bm{q})} \frac{\prod_{i=1}^r p_i! q_i!}{\prod_{i \ne j} n_{ij}!} \prod_{i \ne j} C_{ij}^{n_{ij}}.
			\end{equation}
		\end{lemma}
		\begin{proof}
			This formula follows directly by extracting the coefficients of $\prod_i s_i^{p_i} t_i^{q_i}$ on both sides of the generating function identity. Utilizing the properness and unit variance of the variables, we have
			\begin{equation}\label{eq:H1-block-generating-function}
				\begin{split}
					\bE \prod_{i=1}^r \exp(s_i Z_i + t_i \overline{Z_i} - s_i t_i) &= \exp( \sum_{i,j} s_i C_{ij} t_j - \sum_i s_i t_i ) \\
					&= \exp( \sum_{i \ne j} s_i C_{ij} t_j ).
				\end{split}
			\end{equation}
			Here, the first equality uses the standard moment generating function identity $\bE[\exp(A)] = \exp(\frac{1}{2}\bE[A^2])$ for a zero-mean Gaussian exponent $A$.
			Expanding the left-hand side using Eq.~\eqref{eq:H1-wick-generating-function}, the coefficient of $\prod_i s_i^{p_i} t_i^{q_i}$ is exactly
			\begin{equation}
				\frac{1}{\prod_i p_i! q_i!} \bE \prod_{i=1}^r :\!Z_i^{p_i} \overline{Z_i}^{q_i}\!: .
			\end{equation}
			Expanding the right-hand side of Eq.~\eqref{eq:H1-block-generating-function} as a multivariate Taylor series, we obtain
			\begin{equation}
				\begin{split}
					\exp( \sum_{i \ne j} s_i C_{ij} t_j ) &= \prod_{i \ne j} \sum_{n_{ij}=0}^\infty \frac{(s_i C_{ij} t_j)^{n_{ij}}}{n_{ij}!} \\
					&= \sum_{N} \frac{\prod_{i \ne j} C_{ij}^{n_{ij}}}{\prod_{i \ne j} n_{ij}!} \prod_i s_i^{\sum_j n_{ij}} \prod_j t_j^{\sum_i n_{ij}},
				\end{split}
			\end{equation}
			where the sum runs over all matrices $N$ with zero diagonal ($n_{ii} = 0$). Matching the powers of $s_i$ and $t_j$ on both sides imposes the exact constraints $\sum_j n_{ij} = p_i$ and $\sum_i n_{ij} = q_j$, which directly recovers Eq.~\eqref{eq:H1-block-wick-formula}.
		\end{proof}
		
		Equivalently, the Wick formula can be interpreted through a labeled-diagrammatic representation: each block $i$ possesses $p_i$ labeled slots and $q_i$ labeled anti-slots. A directed edge $i \to j$ pairs one slot of each type and contributes a factor of $C_{ij}$. The condition $n_{ii} = 0$ removes all internal edges within each block $i$.
		
		\subsubsection{Concentration of $H_1$}
		We define 
		\begin{equation}\label{eq:H1-Rg-tilde}
			\widetilde{R}_g \coloneqq \sum_{x\in\mathsf{X}^m} \mu |g_x|^4 \pi_x, \qquad g_x = \bra{x}\ket{g}.
		\end{equation}
		Next, we use the quartic Wick polynomial $\mathfrak{h}(z)$ and its corresponding operator sum $Z_g$:
		\begin{equation}\label{eq:H1-Zg-definition}
			\mathfrak{h}(z) \coloneqq \,\,:\!z^2\overline{z}^2\!: \,\, = |z|^4 - 4|z|^2 + 2, \qquad Z_g \coloneqq \sum_{x \in \mathsf{X}^m} \mu \, \mathfrak{h}(g_x) \pi_x.
		\end{equation}
		Because $\sum_x \mu \pi_x = I_D$, we obtain the decomposition
		\begin{equation}\label{eq:H1-Rg-decomposition}
			\widetilde{R}_g = Z_g + 4\Lambda(gg^\dagger) - 2I_D.
		\end{equation}
		The operator norm of $\Lambda(gg^\dagger)$ is already bounded by Eq.~\eqref{eq:diagonal-bound-final}. Thus, it suffices to bound $Z_g$.
		
		Let $S_g$ denote the squared Frobenius norm of $Z_g$:
		\begin{equation}\label{eq:H1-Sg-definition}
			S_g \coloneqq \norm{Z_g}_F^2 = \sum_{x,y \in \mathsf{X}^m} \mu^2 \, \mathfrak{h}(g_x) \mathfrak{h}(g_y) |\bra{x}\ket{y}|^2.
		\end{equation}
		For any fixed pair of outcomes $x$ and $y$, we evaluate the expectation using Lemma~\ref{lem:wick-formula} with two blocks, setting $p_i = q_i = 2$.
		This leaves exactly two edges from the first block to the second, and two in the reverse direction. 
		Therefore,
		\begin{equation}\label{eq:H1-two-point-wick}
			\begin{split}
				\bE[\mathfrak{h}(g_x) \mathfrak{h}(g_y)] &= (2!)^2 \bE[g_x \overline{g_y}]^2 \bE[g_y \overline{g_x}]^2 \\
				&= 4|\bra{x}\ket{y}|^4.
			\end{split}
		\end{equation}
		Consequently, taking the expectation over $S_g$ factors into a tensor product structure:
		\begin{equation}\label{eq:H1-mean-Sg}
			\begin{split}
				\mathbb{E}[S_g] &= 4 \sum_{x,y \in \mathsf{X}^m} \mu^2 |\bra{x}\ket{y}|^6 \\
				&= 4 [ 3^{-2} \sum_{a,b \in \mathsf{X}} |\bra{a}\ket{b}|^6 ]^m \\
				&= 4.
			\end{split}
		\end{equation}
		The final equality follows by analyzing the single-qubit overlaps. For any fixed basis state $a \in \mathsf{X}$, exactly one of the six Pauli states $b \in \mathsf{X}$ yields a squared overlap of $1$, one yields a squared overlap of $0$, and the remaining four yield a squared overlap of $1/2$. Therefore, the one-qubit factor evaluates identically to $3^{-2} (6) [ 1 + 4 (1/2)^3 ] = 1$.
		
		We next bound the variance of $S_g$. Expanding Eq.~\eqref{eq:H1-Sg-definition} yields
		\begin{equation}\label{eq:H1-variance-expanded}
			\begin{split}
				\Var(S_g) &= \sum_{x_1, \ldots, x_4 \in \mathsf{X}^m} \mu^4 |\bra{x_1}\ket{x_2}|^2 |\bra{x_3}\ket{x_4}|^2 \\
				&\quad \times \operatorname{Cov}( \mathfrak{h}(g_{x_1})\mathfrak{h}(g_{x_2}), \mathfrak{h}(g_{x_3})\mathfrak{h}(g_{x_4}) ).
			\end{split}
		\end{equation}
		In the labeled-diagram representation of Eq.~\eqref{eq:H1-block-wick-formula}, each of the four Wick blocks contains two labeled slots and two labeled anti-slots, resulting in at most $8!$ possible diagrams. Because the covariance $\operatorname{Cov}( \mathfrak{h}(g_{x_1})\mathfrak{h}(g_{x_2}), \mathfrak{h}(g_{x_3})\mathfrak{h}(g_{x_4}) )$ subtracts the disconnected contributions $\bE[\mathfrak{h}(g_{x_1})\mathfrak{h}(g_{x_2})]\bE[\mathfrak{h}(g_{x_3})\mathfrak{h}(g_{x_4})]$, any surviving diagram must contain at least one edge crossing the bipartition $\{1, 2\} \mid \{3, 4\}$. Let $\mathcal{D}_{\mathrm{surv}}$ denote the set of these surviving diagrams.
		
		For each diagram $\beta \in \mathcal{D}_{\mathrm{surv}}$, we define its corresponding one-qubit factor as
		\begin{equation}\label{eq:H1-sbeta-definition}
			s_\beta \coloneqq 3^{-4} \sum_{a_1, \ldots, a_4 \in \mathsf{X}} |\bra{a_1}\ket{a_2}|^2 |\bra{a_3}\ket{a_4}|^2 \prod_{(i \to j) \in \beta} \bra{a_i}\ket{a_j}.
		\end{equation}
		Because each basis state $x_i \in \mathsf{X}^m$ is a product state, the inner products perfectly tensorize over the $m$ independent qubits. Consequently, the variance in Eq.~\eqref{eq:H1-variance-expanded} simplifies to
		\begin{equation}\label{eq:H1-variance-diagrams}
			\Var(S_g) = \sum_{\beta \in \mathcal{D}_{\mathrm{surv}}} s_\beta^m.
		\end{equation}
		
		It remains to establish a uniform upper bound on $|s_\beta|$. We augment the eight directed Wick edges of $\beta$ with the four directed edges originating from the two trace factors in Eq.~\eqref{eq:H1-sbeta-definition}, namely $1 \to 2$, $2 \to 1$, $3 \to 4$, and $4 \to 3$. By ignoring the edge orientations, this construction yields a loopless $6$-regular multigraph $G_\beta$ on four vertices, comprising exactly $12$ edges. 
		
		Let $A_1, \ldots, A_4$ be independent random variables uniformly distributed over the set of Pauli basis states $\mathsf{X}$, and define the overlap function
		\begin{equation}
			k(a,b) \coloneqq |\bra{a}\ket{b}|.
		\end{equation}
		Taking the absolute value of Eq.~\eqref{eq:H1-sbeta-definition} and applying the triangle inequality yields
		\begin{equation}\label{eq:H1-sbeta-graph-bound}
			\begin{split}
				|s_\beta| &\le 16 I(G_\beta), \\
				I(G) &\coloneqq \mathbb{E} \prod_{\{i,j\} \in E(G)} k(A_i, A_j),
			\end{split}
		\end{equation}
		where the prefactor $16 = (6/3)^4$ arises from multiplying the original weight $3^{-4}$ by $6^4$ to convert the explicit summations over the six Pauli states into expectations over the uniform random variables $A_i$.
		
		For any loopless multigraph $G$ with a maximum degree of at most six, we have the bound
		\begin{equation}\label{eq:H1-graphical-holder}
			I(G) \le (1/4)^{|E(G)|/6}.
		\end{equation}
		To establish this, we eliminate the vertices of the graph one at a time. Suppose the current vertex being removed has $r \le 6$ remaining incident edges, and its neighboring states are fixed as $b_1, \ldots, b_r$. Applying the generalized Hölder inequality yields
		\begin{equation}\label{eq:H1-holder-step}
			\begin{split}
				\mathbb{E}_A \prod_{\ell=1}^r k(A, b_\ell) &\le \prod_{\ell=1}^r ( \mathbb{E}_A k(A, b_\ell)^6 )^{1/6} \\
				&= (1/4)^{r/6}.
			\end{split}
		\end{equation}
		Here, for every fixed basis state $b \in \mathsf{X}$, evaluating the sixth moment explicitly gives
		\begin{equation}
			\mathbb{E}_A k(A, b)^6 = \frac{1}{6} [ 1 + 4 (1/2)^3 ] = \frac{1}{4}.
		\end{equation}
		Because each edge in the multigraph is removed exactly once during this vertex elimination process, accumulating these individual bounds directly proves Eq.~\eqref{eq:H1-graphical-holder}.
		
		We now examine the conditions for equality in Eq.~\eqref{eq:H1-graphical-holder} for a $6$-regular graph. Consider the first eliminated vertex $v$, which has exactly six incident edges. Achieving equality in Eq.~\eqref{eq:H1-holder-step} requires the six non-negative functions $k(\,\cdot\,,b_\ell)^6$ to be strictly proportional. Because their $L^1$ norms are identically $1/4$, these functions must be equal, which consequently requires all six neighboring states $b_\ell$ to coincide.
		
		If $v$ were connected to two distinct neighboring vertices, those vertices would be assigned distinct random states with constant probability, rendering the Hölder inequality strict. Therefore, saturating the inequality in Eq.~\eqref{eq:H1-graphical-holder} requires that all six edges incident on $v$ connect to a single adjacent vertex $w$. Because $w$ also has a maximum degree of six, the pair $\{v,w\}$ must form an isolated graph component.
		
		For the specific graph $G_\beta$, applying Eq.~\eqref{eq:H1-graphical-holder} with $|E(G_\beta)|=12$ yields the upper bound $|s_\beta| \le 1$. Reaching equality would force $G_\beta$ to be a union of exactly two disconnected components, each consisting of two vertices joined by six edges. Because $G_\beta$ is constructed with the fixed edges connecting vertices $1$ and $2$, as well as $3$ and $4$, these two isolated components must be precisely $\{1,2\}$ and $\{3,4\}$.
		However, this configuration implies that the underlying Wick diagram contains no edges crossing the bipartition $\{1,2\} \mid \{3,4\}$, which contradicts the definition of the surviving diagrams $\beta \in \mathcal{D}_{\mathrm{surv}}$. Thus, every surviving diagram must strictly satisfy $|s_\beta| < 1$:
		\begin{equation}\label{eq:H1-rho-definition}
			\rho \coloneqq \max_{\beta \in \mathcal{D}_{\mathrm{surv}}} |s_\beta| < 1.
		\end{equation}
		Combining Eq.~\eqref{eq:H1-variance-diagrams} with the uniform bound in Eq.~\eqref{eq:H1-rho-definition}, we obtain the variance bound:
		\begin{equation}\label{eq:H1-variance-final}
			\Var(S_g) \le |\mathcal{D}_{\mathrm{surv}}| \rho^m \le 8! \, \rho^m = e^{-\Omega(m)}.
		\end{equation}
		
		Because $Z_g$ is Hermitian, $\norm{Z_g} \le \norm{Z_g}_F = \sqrt{S_g}$. Combining this with the expectation $\mathbb{E}[S_g] = 4$ from Eq.~\eqref{eq:H1-mean-Sg}, Chebyshev's inequality yields
		\begin{equation}\label{eq:H1-Zg-bound}
			\Pr\left[ \norm{Z_g} > \sqrt{5} \right] \le \Pr\left[ S_g > 5 \right] \le \operatorname{Var}(S_g) = e^{-\Omega(m)}.
		\end{equation}
		Substituting this bound into Eq.~\eqref{eq:H1-Rg-decomposition} and utilizing the diagonal bound from Eq.~\eqref{eq:diagonal-bound-final}, we obtain
		\begin{equation}\label{eq:H1-Rg-bound}
			\norm{\widetilde{R}_g} \le C
		\end{equation}
		with probability at least $1 - e^{-\Omega(m)}$. The analogous bound $\norm{\widetilde{R}_h} \le C$ holds for the vector $h$.
		
		For the Haar-random target state $\ket{\psi}$, the outcome probabilities are given by
		\begin{equation}
			p_x = \frac{|g_x|^2 + |h_x|^2}{T}, \qquad T = \norm{g}^2 + \norm{h}^2.
		\end{equation}
		By Eq~\eqref{eq:T-concentration}, $T = \Theta(D)$ with probability at least $1 - e^{-\Omega(D)}$. We define the composite operator
		\begin{equation}
			R_\psi \coloneqq \sum_{x \in \mathsf{X}^m} \mu p_x^2 \pi_x.
		\end{equation}
		Applying $(a+b)^2 \le 2(a^2+b^2)$ alongside the bounds for $\widetilde{R}_g$ and $\widetilde{R}_h$, we can upper-bound $R_\psi$ as
		\begin{equation}\label{eq:H1-Rpsi-bound}
			R_\psi \le \frac{2}{T^2} \left( \widetilde{R}_g + \widetilde{R}_h \right) \le \frac{C}{D^2} I_D,
		\end{equation}
		where $C > 0$ is a universal constant. This inequality holds with probability at least $1 - e^{-\Omega(m)}$.
		
		Finally, observe that the unnormalized vector $\ket{w_x}$ satisfies $\norm{w_x}^2 = p_x$ and resides entirely within the subspace $\mathbb{C}^2 \otimes \operatorname{span}\{\ket{x}\}$. Consequently, its rank-one projector satisfies
		\begin{equation}
			p_x \ket{w_x}\bra{w_x} \le p_x^2 (I_2 \otimes \pi_x).
		\end{equation}
		Summing this inequality over all outcomes $x$ and applying Eq.~\eqref{eq:H1-Rpsi-bound}, we arrive at the final bound for $H_1$:
		\begin{equation}\label{eq:H1-final-bound}
			H_1 \le I_2 \otimes R_\psi \le \frac{C}{D^2} I_{2D}.
		\end{equation}
		This completes the proof of Lemma~\ref{lem:H1}.
		
	}
	
	
	\bibliographystyle{SciAdv}
	\bibliography{refProperty}
	
\end{document}